\documentclass[a4paper,11pt]{article}

\usepackage[svgnames,table]{xcolor}
\usepackage{fullpage}
\definecolor{ACMBlue}{rgb}{0,0,1}
\definecolor{ACMOrange}{rgb}{1,0.5,0}
\definecolor{ACMRed}{rgb}{1,0,0}
\definecolor{ACMPurple}{rgb}{1,0,1}

\usepackage{booktabs} 
\usepackage{multirow}
\usepackage{amsmath,amssymb,amsthm}
\usepackage[linesnumbered, ruled, lined, noend]{algorithm2e}
\usepackage{caption}
\usepackage{subcaption}
\usepackage{fancybox}
\usepackage{rotating}
\usepackage{tikz}
\usetikzlibrary{tikzmark, calc,arrows.meta, backgrounds} 
\usepackage{tikz-layers}
\usepackage{url}
\usepackage[sort&compress,numbers]{natbib}

\newtheorem{definition}{Definition}
\newtheorem{proposition}{Proposition}
\newtheorem{theorem}{Theorem}
\newtheorem{lemma}{Lemma}
\newtheorem{corollary}{Corollary}
\newtheorem{claim}{Claim}

\newcommand{\NDFAS}{\textsc{Negative DFAS}}

\newcommand{\Whard}[1]{$\mathsf{W}[#1]$-hard}
\newcommand{\MulticoloredClique}{\textsc{Multicolored Clique}}
\newcommand{\td}{\ensuremath\operatorname{\mathsf{td}}}
\newcommand{\pw}{\ensuremath\operatorname{\mathsf{pw}}}
\newcommand{\tw}{\ensuremath\operatorname{\mathsf{tw}}}
\newcommand{\problemdef}[3]{%
	\vspace{1em}
  \fbox { \parbox { .9\textwidth} {
    \textsc{\large {\hspace{1em} #1}}
    \begin{description}
        \item[Input: ] #2
        \item[Task:]\hspace{1mm} #3
    \end{description}}}\vspace{1em}
}

\usepackage[english]{babel}
\usepackage[utf8]{inputenc}
\usepackage[T1]{fontenc}
\addto\extrasenglish{%
}
\usepackage{thmtools}
\newenvironment{claimproof}{%
  \begin{proof}[Proof of Claim~\theclaim]%
}{%
	\end{proof}%
}

\AtBeginDocument{%
  \providecommand\BibTeX{{%
    \normalfont B\kern-0.5em{\scshape i\kern-0.25em b}\kern-0.8em\TeX}}}

\usepackage[colorlinks=true,linkcolor=blue,citecolor=red]{hyperref}
\usepackage[nameinlink]{cleveref}
\Crefname{claim}{Claim}{Claims}

\title{Resolving Infeasibility of Linear Systems:\\ A Parameterized Approach\thanks{{Preliminary versions of the results of the paper appeared 2019 and 2020 in \cite{gokePaper2019,GokeThesis}.}}}

\author{Krist{\'o}f B{\'e}rczi\thanks{MTA-ELTE Momentum Matroid Optimization Research Group and MTA-ELTE Egerv\'ary Research Group, Department of Operations Research, E\"otv\"os Lor\'and University, Budapest, Hungary. \texttt{kristof.berczi@ttk.elte.hu}}
   \and Alexander G{\"o}ke\thanks{Hamburg University of Technology, Institute for Algorithms and Complexity, Hamburg, Germany. \texttt{alexander.goeke@tuhh.de}}
   \and Lydia Mirabel Mendoza-Cadena\thanks{MTA-ELTE Momentum Matroid Optimization Research Group, Department of Operations Research, E\"otv\"os Lor\'and University, Budapest, Hungary, \texttt{lyd21@student.elte.hu}}
   \and Matthias Mnich\thanks{Hamburg University of Technology, Institute for Algorithms and Complexity, Hamburg, Germany. \texttt{matthias.mnich@tuhh.de}}}

\begin{document}

\date{~}

\maketitle

\begin{abstract}
  Deciding feasibility of large systems of linear equations and inequalities is one of the most fundamental algorithmic tasks.
  However, due to data inaccuracies or modeling errors, in practical applications one often faces linear systems that are \emph{infeasible}.
  Extensive theoretical and practical methods have been proposed for post-infeasibility analysis of linear systems.
  This generally amounts to detecting a feasibility blocker of small size $k$, which is a set of equations and inequalities whose removal or perturbation from the large system of size $m$ yields a feasible system.
  This motivates a parameterized approach towards post-infeasibility analysis, where we aim to find a feasibility blocker of size at most $k$ in fixed-parameter time $f(k)\cdot m^{\mathcal O(1)}$.
  
  We establish parameterized intractability ($\mathsf{W}[1]$- and $\mathsf{NP}$-hardness) results already in very restricted settings for different choices of the parameters maximum size of a deletion set, number of positive/negative right-hand sides, treewidth, pathwidth and treedepth.
  Additionally, we rule out a polynomial compression for {\sc MinFB} parameterized by the size of a deletion set and the number of negative right-hand sides.
  
  Furthermore, we develop fixed-parameter algorithms parameterized by various combinations of these parameters when every row of the system corresponds to a difference constraint. 
  Our algorithms capture the case of {\sc Directed Feedback Arc Set}, a fundamental parameterized problem whose fixed-parameter tractability was shown by Chen et al. (STOC 2008).
\end{abstract}

%



\section{Introduction}
\label{sec:introduction}
Linear programming is without doubt one of the most powerful tools of optimization theory.
However, the data that is used in these systems may be subject to inaccuracies and uncertainties, and therefore may lead to systems which are infeasible.
Infeasibility may also be the result of modelling errors, or simply incompatibility of constraints.
However, infeasibility itself allows for little conclusions; for a large system of millions of inequalities, infeasibility may stem from a very small subset of data.
A natural question is therefore to detect the smallest number of changes which must be made to a given system in order to make it feasible.
The analysis of infeasible linear systems has been extensively investigated~\cite{AmaldiKann1998,AroraEtAl1997,BermanKarpinski2002,Chinneck1996,Chinneck1997,Chinneck2008,LeithauserEtAl2012,Pfetsch2008}; we refer to the book by Chinneck~\cite{Chinneck2008} for an overview. 

Formally, the {\sc Minimum Feasibility Blocker (MinFB)} problem takes as input a system~$\mathcal S$ of linear inequalities $Ax\leq b$ and asks for a smallest subset $\mathcal I$ such that $\mathcal S\setminus \mathcal I$ is \emph{feasible}.
As~$\mathcal I$ ``blocks'' the feasibility of $\mathcal S$, we refer to $\mathcal I$ as a \emph{feasibility blocker}; we avoid calling $\mathcal I$ a ``solution'', to avoid confusion with the solution of the linear system $\mathcal S\setminus\mathcal I$.
Instead of removing the set $\mathcal I$ of inequalities, we can equivalently perturb the right-hand sides~$b_{\mathcal I}$ of the inequalities in~$\mathcal I$; that is, we can increase the $b$-values of inequalities in $\mathcal I$ to a value such that the perturbed system becomes feasible.

When talking about feasibility, we have to specify over which field, and our choice here is the field $\mathbb Q$.
Over this field, {\sc MinFB} is $\mathsf{NP}$-hard~\cite{Sankaran1993}; as the feasibility of a linear system can be tested in polynomial time (e.g., by the ellipsoid method~\cite{GrotschelEtAl1993}) the {\sc MinFB} problem is $\mathsf{NP}$-complete.
Due to its importance, the {\sc MinFB} problem has been thoroughly investigated from several different viewpoints, including approximation algorithms~\cite{AmaldiKann1998,LeithauserEtAl2012}, polyhedral combinatorics~\cite{Pfetsch2008}, heuristics~\cite{Chinneck1996}, mixed-integer programming~\cite{codato2006}, and hardness of approximation~\cite{AroraEtAl1997}.

Here we take a new perspective on the {\sc MinFB} problem, based on parameterized complexity.
In parameterized complexity, the problem input of size $n$ is additionally equipped with one or more integer parameters $k$ and one measures the problem complexity in terms of both $n$ and~$k$.
The goal is to solve such instances by \emph{fixed-parameter algorithms}, which run in time $f(k)\cdot n^{\mathcal O(1)}$ for some computable function~$f$.
The motivation is that fixed-parameter algorithms can be practical for small parameter values $k$ even for inputs of large size $n$, provided that the function~$f$ exhibits moderate growth.
This is in contrast with algorithms that require time~$n^{f(k)}$, which cannot presumed to be practical for large input sizes $n$.
To show that such impractical run times are best possible, a common approach is to show the problem to be $\mathsf{W}[1]$-hard; a standard hypothesis in parameterized complexity is that no $\mathsf{W}[1]$-hard problem admits a fixed-parameter algorithm.
For background on parameterized complexity, we refer to the monograph by Cygan et al.~\cite{CyganEtAl2015}.

For the {\sc MinFB} problem, arguably the most natural parameter is the minimum size $k$ of a feasibility blocker $\mathcal I$.
The motivation for this choice is that in applications, we are interested in \emph{small} feasibility blockers $\mathcal I$; e.g., Chakravarti~\cite{Chakravarti1994} argues that a feasibility blocker ``with too large cardinality may be hard to comprehend and may not be very useful for post-infeasibility analysis.''
Guillemot~\cite{Guillemot2011} explicitly posed the question of resolving the parameterized complexity of {\sc MinFB}; he conjectured that the problem is fixed-parameter tractable parameterized by the size of a minimum feasibility blocker for matrices with at most 2 non-zero entries per row.

A motivation for our approach comes from the fact that {\sc MinFB} captures one of the most important problems in parameterized complexity, namely {\sc Directed Feedback Arc Set (DFAS)}: given a directed graph $G$, decide if $G$ admits an arc set $F$ of size at most $k$ such that $G - F$ is an acyclic directed graph (DAG). 
It was a long-standing open question whether DFAS admits a fixed-parameter algorithm parameterized by the size~$k$ of the smallest directed feedback arc set, until Chen et al.~\cite{ChenEtAl2008} gave an algorithm with run time $4^kk!n^{\mathcal O(1)}$.
The currently known fastest algorithm for DFAS runs in time $4^kk!k\cdot \mathcal O(n + m)$, and is due to Lokshtanov et al.~\cite{LokshtanovEtAl2018}.
It is not difficult to give a parameter-preserving reduction from DFAS to {\sc MinFB}: for every arc $(u,v)$ of the digraph $G$ that serves as an input to DFAS, we add the inequality $x_u - x_v \leq -1$ to the linear system.
Directed feedback arc sets~$F$ of $G$ are then mapped to feasibility blockers of the same size by removing the constraints corresponding to arcs in $F$, and vice-versa; see \Cref{sec:preliminaries} for details.
Note that the constraint matrix $A$ arising this way is totally unimodular, and each row has exactly two non-zero entries, one $+1$ and one~$-1$.
Such constraints are called \emph{difference constraints}; testing feasibility of systems of difference constraints has been investigated extensively~\cite{RamalingamEtAl1999,SubramaniWojciechowski2017} due to their practical relevance most notably in temporal reasoning.
It is therefore interesting to know whether the more general {\sc MinFB} problem also admits a fixed-parameter algorithm for parameter~$k$, even for special cases like totally unimodular matrices $A$ (where testing feasibility is easy).

Another case of interest for {\sc MinFB} is when the constraint matrix $A$ has bounded treewidth, where the treewidth of~$A$ is defined as the treewidth of the bipartite graph that originates from assigning one vertex to every row and every column of~$A$ and connecting any two vertices by an edge whose corresponding entry in~$A$ is non-zero.
Fomin et al.~\cite{FominEtAl2018} gave a fast algorithm for {\sc MinFB} with constraint matrices of bounded treewidth for the setting $k = 0$.
At the same time, Bonamy et al.~\cite{BonamyEtAl2018} showed that DFAS---a special case of {\sc MinFB}---is fixed-parameter tractable parameterized by the treewidth of the underlying undirected graph of the input digraph\footnote{The algorithm of Bonamy et al.~\cite{BonamyEtAl2018} is stated for the vertex deletion problem, but it can be modified to work for DFAS as well. Note that the standard reduction from DFAS to the vertex deletion problem which preserves the solution size does not necessarily result in a digraph whose underlying undirected graph has bounded treewidth even if the DFAS instance has this property.}.
So the questions arise whether Fomin et al.'s algorithm can be extended to arbitrary values of~$k$, or whether Bonamy et al.'s algorithm can be extended from DFAS to {\sc MinFB}.

Besides treewidth, another parameter in the context of the structure of the constraint matrix~$A$ is the treedepth.
This parameter measures how many recursive vertex deletions are necessary to delete the whole graph. Chan et al.~\cite{ChanEtAl2020}  used the treedepth to solve (integer) linear program formulations.
Also, with the help of this parameter, Iwata et al.~\cite{IwataOO2018} obtained an algorithm to find a negative cycle.
Recently, Ganian et al. \cite{ganian2020fixed} made use of the treedepth to obtain a fixed-parameter algorithm that evaluates dependency-quantified-Boolean formulas.

One of the main unresolved questions around DFAS is whether it admits a \emph{polynomial compression}.
That is, one seeks an algorithm that, given any directed graph~$G$ and integer $k$, in polynomial time computes an instance $I$ of a decision problem $\Pi$ whose size is bounded by some polynomial~$p(k)$, such that $G$ admits a feedback arc set of size at most $k$ if and only if $I$ is a ``yes''-instance of $\Pi$.
The question for a polynomial compression has been stated numerous times as an open problem~\cite{Worker2010,Worker2013,Bedlewo2014,MnichvanLeeuwen2017}; from the algorithms by Chen et al.~\cite{ChenEtAl2008} and Lokshtanov et al.~\cite{LokshtanovEtAl2018} only an \emph{exponential} bound on the size of $I$ follows.
On the other hand, parameterized complexity provides tools such as cross-composition to rule out the existence of such polynomials $p(k)$ modulo the non-collapse of the polynomial hierarchy; we refer to Bodlaender et al.~\cite{BodlaenderEtAl2014} for background.
Given the elusiveness of this problem, we approach the (non-)existence of polynomial compression for DFAS from the angle of the more general {\sc MinFB} problem.

\subsection{Related Work}
\label{sec:relatedwork}
In their fundamental work, Arora, Babai, Stern and Sweedyk~\cite{AroraEtAl1997} considered the problem of removing a smallest set of \emph{equations} to make a given system of linear equations feasible over~$\mathbb Q$.
They gave strong inapproximability results, showing that finding \emph{any} constant-factor approximation is $\mathsf{NP}$-hard.
Berman and Karpinski~\cite{BermanKarpinski2002} gave the first (randomized) polynomial-time algorithm with sublinear approximation ratio for this problem.
The parameterized complexity of the setting with equations was recently investigated by Dabrowski et al.~\cite{DabrowskiJOOW2022}.
Note that for equations, e.g., $x_1+\hdots+x_r = b$, once the values of $r-1$ of the variables have been fixed then the value of the $r^{\textnormal{th}}$ variable can be inferred to satisfy the equation; such an inferral is, however, not possible in the setting of inequalities which we study in this paper.
Therefore, the problems that we study in this paper generalize other known combinatorial optimization problems (e.g., {\sc Directed Feedback Arc Set}) compared to the work of Dabrowski et al.~\cite{DabrowskiJOOW2022} (e.g., {\sc Bipartization}, {\sc Multiway Cut}).

Giannopolous, Knauer and Rote~\cite{GiannopoulosEtAl2009} considered the ``dual'' of {\sc MinFB} from a parameterized point of view: namely, in {\sc MaxFS} we ask for a largest subsystem of an $n$-dimensional linear system $\mathcal S$ which is feasible over~$\mathbb Q$.
They showed that deciding whether a feasible subsystem of at least $\ell$ inequalities in $\mathcal S$ exists is $\mathsf{W}[1]$-hard parameterized by $n + \ell$, even when $\mathcal S$ consists of equations only.

For systems of equations over \emph{finite} fields, finding minimum feasibility blockers has been considered from a parameterized perspective.
In particular, over the binary field $\mathbb F_2$, Crowston et al.~\cite{CrowstonEtAl2013} proved $\mathsf{W}[1]$-hardness even if each equation has exactly three variables and every variable appears in exactly three equations; they further gave a fixed-parameter algorithm for the case where each equation has at most two variables.

\begin{table}[t!]
  \renewcommand*{\arraystretch}{1.5}
  \newcommand{\tableentrybase}[2]{\colorbox{#1}{\parbox[c][2.2em][c]{18mm}{\centering\small #2}}}
  \caption{Complexity landscape of \textsc{MinFB} for difference constraints. ``FPT'' denotes the existence of a fixed-parameter algorithm. Columns of $\pw(G)$ and $\tw(G)$ are merged as they share the same algorithmic and hardness results.
  That is, algorithmic results hold already when the parameter~$\tw(G)$ is used, and hardness results hold already when the parameter~$\pw(G)$ is used.
  Complexities of cells with light shade are implied by cells with darker shade. \label{table:summary}}
  \begin{tabular}{clcccccccc}					
    \multicolumn{2}{c}{  }			& \multicolumn{2}{c}{$-$}	&	\multicolumn{2}{c}{$\pw(G)$ or $\tw(G)$}			&	\multicolumn{2}{c}{$\td(G)$}				\\ \cmidrule(lr){3-4} \cmidrule(lr){5-6} \cmidrule(lr){7-8} 
    \multicolumn{2}{c}{ }			&	$-$	&	$w_+$	&	$\pw/\tw$	&	$\pw/\tw+w_+$	&	$\td$	&	$\td + w_+$		\\ \addlinespace[1pt]
    \multirow[c]{3}{*}{\rotatebox[origin=c]{90}{ $b \in \mathbb{Z} $}} &	$-$	&	 \tableentrybase{ACMBlue!20}{$\mathsf{NP}$-hard}  	&	\tableentrybase{ACMBlue!20}{$\mathsf{NP}$-hard } 	&	\tableentrybase{ACMBlue!20}{$\mathsf{NP}$-hard } 	&	\tableentrybase{ACMBlue!50}{$\mathsf{NP}$-hard Thm.~\ref{thm:NP_hardness_for_pathwidth}}	&	\tableentrybase{ACMBlue!20}{\Whard{1}}  	&	\tableentrybase{ACMBlue!20}{\Whard{1}}  		\\	\addlinespace[1pt] 
		&	$k$	&	\tableentrybase{ACMBlue!20}{\Whard{1}} 	&	\tableentrybase{ACMBlue!20}{\Whard{1}} 	&	\tableentrybase{ACMBlue!20}{\Whard{1}}	&	\tableentrybase{ACMBlue!50}{\Whard{1} Thm.~\ref{thm:W1_hardness_for_pathwidth_and_deletion_size}}	&	\tableentrybase{ACMOrange!50}{FPT \linebreak Thm.~\ref{thm:algorithm_for_td_plus_k}}	&	\tableentrybase{ACMOrange!30}{FPT}		\\ \addlinespace[1pt]
		&	$w_-$	&	\tableentrybase{ACMBlue!20}{\Whard{1}} 	&	\tableentrybase{ACMOrange!50}{FPT \linebreak Thm.~\ref{thm:algorithm_for_bounded_number_of_non_zero_arcs} }  	&	\tableentrybase{ACMBlue!50}{\Whard{1} Thm.~\ref{thm:W1_hardness_for_pathwidth_number_negative_arcs_and_deletion_size}}	&	\tableentrybase{ACMOrange!30}{FPT}	&	\tableentrybase{ACMOrange!30}{FPT}	&	\tableentrybase{ACMOrange!30}{FPT}	\\ \addlinespace[1em] 
		\multirow[c]{3}{*}{\rotatebox[origin=c]{90}{$b \in \{ \pm 1 , 0\}$}} &	$-$	&	\tableentrybase{ACMBlue!20}{$\mathsf{NP}$-hard \cite{Sankaran1993} } 	&	\tableentrybase{ACMBlue!50}{$\mathsf{NP}$-hard \linebreak Thm.~\ref{thm:NP_hardness_for_positive_arcs}} 	&	\tableentrybase{ACMBlue!20}{\Whard{1}}	&	\tableentrybase{ACMOrange!50}{FPT Thm.~\ref{thm:algorithm_for_treewidth_normalized_weights_and_bounded_number_of_positive_arcs}}	&	\tableentrybase{ACMOrange!50}{FPT \hfil Thm.~\ref{thm:algorithm_for_treedepth_and_normalized_weights}}	&	\tableentrybase{ACMOrange!30}{FPT}		\\ \addlinespace[1pt]		
		&	$k$	&	\tableentrybase{ACMBlue!20}{\Whard{1}} 	&	Open	&	\tableentrybase{ACMBlue!50}{\Whard{1} Thm.~\ref{thm:W1_hardness_for_pathwidth_deletion_size_and_normalized_weights}}	&	\tableentrybase{ACMOrange!30}{FPT}	&	\tableentrybase{ACMOrange!30}{FPT}	&	\tableentrybase{ACMOrange!30}{FPT}		\\ \addlinespace[1pt]
		&	$w_-$	&	Open	&	\tableentrybase{ACMOrange!30}{FPT }  	&	\tableentrybase{ACMOrange!50}{FPT Thm.~\ref{thm:algorithm_for_treewidth_normalized_weights_and_bounded_number_of_negative_arcs}}  	&	\tableentrybase{ACMOrange!30}{FPT}	&	\tableentrybase{ACMOrange!30}{FPT}	&	\tableentrybase{ACMOrange!30}{FPT}
  \end{tabular}
\end{table}

\subsection{Our Results}
\label{sec:results}
The main contribution of the paper is settling the complexity of {\sc MinFB} in systems \mbox{$Ax\leq b$} consisting of difference constraints for various choices and combinations of parameters, see \Cref{table:summary} for an overview.
Throughout the paper, we denote the number of positive and negative coordinates of $b$ by $w_+$ and $w_-$, respectively.
For non-trivial instances, $k<w_-$ holds as otherwise we can delete all constraints with negative right hand side to get a feasible system.
A constraint matrix $A$ with difference constraints can be naturally identified with a directed graph $G$, whose treewidth, pathwidth and treedepth is denoted by $\tw(G)$, $\pw(G)$ and $\td(G)$, respectively. 

The results in \Cref{table:summary} are split horizontally by the type of the right hand side ($b\in\mathbb{Z}$ vs. $b\in\{-1,0,1\}$).
In both cases, the rows are indexed by parameters $\{-,k,w_-\}$, where $-$ means that no parameter of this set is taken.
Meanwhile, the columns are doubly indexed by parameters in $\lbrace -, \tw(G), \pw(G), \td(G)\rbrace$ and $\lbrace -, w_+\rbrace$.
As $\tw(G) \leq \pw(G) \leq \td(G)$ holds and $k\leq w_-$ for non-trivial instances, it is sufficient to obtain algorithmic or hardness results for certain combinations of the parameters.
For example, a hardness result for $\pw(G)+w_++k$ immediately translates into a hardness result for $w_-$.
\bigskip

The paper is organized as follows.
In \Cref{sec:preliminaries}, we introduce basic definitions and discuss the connection between the {\sc MinFB} and DFAS problems.
Algorithms for the different combinations of parameters are presented in \Cref{sec:alg}. While two of the problems remain open when the right-hand side of the system has the form $b \in \lbrace -1, 0, 1\rbrace$, we obtain fixed-parameter algorithms for these cases under the stronger assumption $b \in \lbrace -1, 1\rbrace$.
Hardness results are discussed in \Cref{sec:hardness}, where we also rule out a polynomial compression for {\sc MinFB} parameterized by~$k + w_-$.

\section{Preliminaries}
\label{sec:preliminaries}
\paragraph{Directed graphs and feasible potentials.}
Throughout, we work with finite and loopless directed graphs $G$ with vertex set $V(G)$ and arc set $A(G)$.
The sets of edges \emph{leaving} and \emph{entering} a vertex $v\in V(G)$ are denoted by $\delta^+(v)$ and $\delta^-(v)$, respectively.
A \emph{walk} $W$ in $G$ is a sequence of vertices $W = (v_0, v_1, \dots, v_\ell)$ such that $(v_i,v_{i+1})\in A(G)$ for $i = 0,\dots,\ell - 1$.
We call $\ell$ the \emph{length} of the walk.
A walk is \emph{closed} if $v_0 = v_\ell$.
We call the walk a \emph{path} if all of its vertices are distinct.
If all of its vertices are distinct except $v_0$ and $v_\ell$, we call it a \emph{cycle}.
For two walks $W, R$, where the last vertex of $W$ equals the first vertex of $R$, we denote by $W \circ R$ the \emph{concatenation} of~$W$ and~$R$, which is the sequence of all vertices in $W$ followed by all vertices in $R$ except the first.
Our directed graphs $G$ are often \emph{arc-weighted}, which means that they  come with arc weights $w: A(G)\rightarrow \mathbb{Q}$.
The sets of arcs with negative, zero, and positive weights are denoted $A_-$, $A_0$ and $A_+$, respectively. 
Then the set of non-zero arcs is $A_{\neq 0}=A_-\cup A_+$.
The numbers of negative and positive arcs are denoted by~$w_-$ and $w_+$, respectively.
The \emph{weight} of a cycle $C$ in $G$ is then equal to the sum of its arc weights.
In that spirit, we call a cycle \emph{negative} (\emph{non-negative}, \emph{positive}) if its weight is \emph{negative} (\emph{non-negative}, \emph{positive}).
A \emph{shortest} path or cycle is a path or cycle of minimum length.
Note that ``shortest'' does not refer to the weight of a cycle.

We define the arc-weighted variant of the DFAS problem.

\problemdef{Negative Directed Feedback Arc Set (Negative DFAS)}
  {A directed graph $G$, a weight function $w:A(G) \rightarrow \mathbb{Q}$ and an integer $k$.}
  {Find a set $X \subseteq A(G)$ of size at most $k$ such that $G-X$ has no negative cycles.}

A solution $X$ to the \NDFAS{} problem is called a \emph{negative directed feedback arc set} of $(G,w)$.
We will often use the notion of feasible potentials.
\begin{definition}
  Let $(G, w)$ be an arc-weighted directed graph.
  A vertex function $\pi: V(G) \to \mathbb{R}$ is called \emph{feasible potential} if for all arcs $a = (u,v) \in A(G)$ we have $\pi(u) - \pi(v) + w(a) \geq 0$.
\end{definition}

The fundamental result of Gallai~\cite{gallai1958maximum} provides a connection between feasible potentials and negative cycles.
\begin{proposition}[Gallai~\cite{gallai1958maximum}]
\label{thm:gallai}
  An arc-weighted directed graph $(G,w)$ contains no cycle of negative total weight if and only if it has a feasible potential.
\end{proposition}

By \Cref{thm:gallai}, for every solution $S$ of a \NDFAS{} instance $(G, w, k)$ there is a feasible potential $\pi_S$ on $G - S$.
We often call a feasible potential of $G - S$ a \emph{feasible potential for $S$}.
Such a feasible potential also certifies that $S$ is indeed a solution.

Let us now formally define the {\sc MinFB} problem.

\problemdef{Minimum Feasibility Blocker (MinFB)}
  {A coefficient matrix $A \in \mathbb{R}^{m \times n}$, a right-hand side vector $b \in \mathbb{R}^m$ and an integer~$k$.}
  {Find a set $\mathcal{I} \subseteq \{1,\dots,m\}$ of size at most $k$ such that $\left(a_{i,\bullet} \cdot x \leq b_i\right)_{i \in \{1,\dots,m\} \setminus \mathcal{I}}$ is feasible.}

By multiplying rows of $A$ and the corresponding entries of $b$ with $(-1)$, the {\sc MinFB} problem also covers, in a parameter-equivalent way, those cases where some inequalities are of the type $a_{i,\bullet} x \geq b_i$. Furthermore, equations of the form $a_{i,\bullet} \cdot x = b_i$ can be written as two inequalities $a_{i,\bullet} \cdot x \le b_i$ and $-a_{i,\bullet} \cdot x \le -b_i$ (equivalent to $a_{i,\bullet} \cdot x \ge b_i$).
In a feasible solution $x^*$ violating such an equation, at most one of the above inequalities is violated.
Thus, the {\sc MinFB} problem with equations can be reduced to the formulation of the {\sc MinFB} as formulated above without changing the parameter.

Throughout the paper, we concentrate on instances of {\sc MinFB} consisting of difference constraints without explicitly mentioning it, that is, each row of the matrix $A$ contains exactly two non-zero entries, one $+1$ and one $-1$. The next theorem shows that in such cases {\sc MinFB} can be reformulated as a {\sc Negative DFAS} instance.

\begin{theorem}
\label{thm:NDFASisequivalenttospecialMIS}
  The {\sc Negative DFAS} problem and the {\sc MinFB} problem for difference constraints are parameter-equivalent.
\end{theorem}
\begin{proof}
  Let $(G, w, k)$ be a {\sc Negative DFAS} instance, and let $n = |V(G)|$ and $m = |A(G)|$.
  Fix an arbitrary order $v_1, \hdots, v_n$ of the vertices and $e_1,\dots,e_m$ of the arcs.
  Let $A = (a_{i,j})\in\mathbb R^{m\times n}$ denote the incidence matrix of $G$ with entries $a_{i,j} = +1$ for $e_i\in\delta^-(v_j)$, $a_{i,j} = -1$ for $e_i\in\delta^+(v_j)$, and $a_{i,j} = 0$ otherwise.
  Furthermore, let $b_i = w(\alpha_i)$.
  The resulting tuple $(A, b, k)$ is an instance of the {\sc MinFB} problem with $A$ being a matrix of difference constraints.
	
  The construction is bijective in the following sense.
  Given an instance $(A,b,k)$ of {\sc MinFB} with $A = (a_{i,j})\in\mathbb R^{m\times n}$ consisting of difference constraints,
  define a directed graph on~$n$ vertices $v_1, \hdots, v_n$ as follows.
  For every constraint $a_{i,\bullet} \cdot x \leq b_i$, we add an arc $e = (v_{j^+}, v_{j^-})$, where $j^-$ is the unique index with $a_{i,j^-} = -1$ and $j^+$ is the unique index with $a_{i,j^+} = +1$.
  Furthermore, set the weight of $e$ to be $w(e) = b_i$.
  Let $G$ denote the resulting digraph.
  Then $(G, w, k)$ is a {\sc Negative DFAS} instance.

  Our goal is to compare the solutions of the two problems.
  Intuitively, deleted constraints and arcs are in a one to one correspondence. 
  Formally, for each $X \subseteq A(G)$ denote by $I_X$ the corresponding indices of the constraints.
  Then we get the following chain of equivalent statements:
  \begin{itemize}
    \item[] $(G - X, w)$ contains no negative cycles with respect to $w$.
    \item[$\Leftrightarrow$] $(G - X, w)$ has a feasible potential $\pi: V(G) \rightarrow \mathbb{R}$.
    \item[$\Leftrightarrow$] There is some $\pi: V(G) \rightarrow \mathbb{R}$ such that $\pi(u) \leq \pi(v) + w(e)$ for all $e=(u,v) \in A(G)\setminus X$.
    \item[$\Leftrightarrow$] There is some $x \in \mathbb{R}^{V(G)}$ such that $x_u - x_v \leq w(e)$ for all $e = (u,v) \in A(G) \setminus X$.
    \item[$\Leftrightarrow$] There is some $x \in \mathbb{R}^n$ such that $a_{i,\bullet} \cdot x \leq b_i$ for all $i \in \{1,\dots,m\} \setminus I_X$.
  \end{itemize}
  Furthermore, as $X$ and $X_\mathcal{I}$ have the same cardinality, we get that~$X$ is a solution to $(G, w, k)$ if and only if $X_\mathcal{I}$ is a solution to $(A, b, k)$.
\end{proof}

Observe that all instances of \NDFAS{} with number of negative arcs smaller or equal to $k$ are solvable in polynomial time.
\begin{theorem}
\label{thm:solve_w_minus_bounded_by_k}
  Let $(G, w, k)$ be a {\sc Negative DFAS} instance with $w_- \leq k$.
  Then $(G, w, k)$ has a solution.
  Moreover, we can check for this condition and compute a solution in linear time.
\end{theorem}
\begin{proof}
  Iterate over $A(G)$ and collect the set $A_-$ of all arcs with negative weight.
  Then $G - A_-$ contains no negative cycles, as it contains no arcs of negative weight.
  Moreover, we can check whether $|A_-| = w_- \leq k$ and possibly return~$A_-$ as a solution.
\end{proof}

\paragraph{Treewidth, pathwidth and treedepth.}
For an undirected graph $H$, a \emph{tree decomposition} (\emph{path decompostion}) is a pair~$(T,\mathcal B)$ where $T$ is a tree (path) and $\mathcal B$ a collection of bags $B_v\subseteq V(H)$, each bag corresponding to some node $v\in V(T)$.
The bags have the property that for any edge of $H$, there is a bag in $\mathcal B$ that contains both of its endpoints. 
Furthermore, the bags containing $v$ form a subtree (subpath) of $T$ for each vertex $v\in V(H)$.
The \emph{width} of $(T,\mathcal B)$ is defined as the largest bag size of $\mathcal B$ minus one.
The \emph{treewidth} (\emph{pathwidth}) of $H$ is the minimum width $\tw(G)$ over all tree (path) decompositions of~$H$. 
For a directed graph $G$, its treewidth (pathwidth) is defined as the treewidth (pathwidth) of the underlying undirected graph $\langle G\rangle$ of~$G$.

For our usage, it is enough to have a tree decomposition that is within a constant factor of minimum width.
This allows us to use the following algorithm by Korhonen~\cite{Korhonen2022}.
\begin{proposition}[\cite{Korhonen2022}]
	\label{thm:compute_approximate_tree_decomposition}
	There is an algorithm that, given a graph $G$ on $n$ vertices, constructs a tree decomposition of $G$ of width at most $2\tw(G)$ in time~$2^{\mathcal{O}(\tw(G))}n$.
\end{proposition}

Nice tree decompositions are useful structures for dynamic programs in graphs of bounded treewidth. 
\begin{definition}
\label{def:nice_tree_decomposition}
  A \emph{nice tree decomposition} is a rooted tree decomposition $(T, \{B_i\}_{i\in V(T)})$ of a graph $H$ with the following properties:
  \begin{itemize}
    \item For the root $r$ and for every leaf $\ell$, $B_r = \emptyset = B_\ell$.
    \item Each non-leaf node $x$ of $T$ belongs to one of the following categories:
    \begin{itemize}
      \item \emph{Introduce node.} The node $x$ has exactly one child $y$ so that $B_x=B_y \cup \{ v \}$ for some vertex $v \in V(H)$.
      \item \emph{Forget node.} The node $x$ has exactly one child $y$ so that  $B_x=B_y \setminus \{ v \}$ for some vertex $v \in V(H)$.
      \item \emph{Join node.} The node $x$ has exactly two children $y$ and $y'$ so that $B_x=B_y=B_{y'}$.
    \end{itemize}
  \end{itemize}
  For $x \in V(G)$, we denote by $T_x$ the subtree of $T$ rooted at $x$.
  Moreover, we define~$G_x$ to be the graph $G[\bigcup_{y \in T_x} B_y]$.
\end{definition}

Given a tree decomposition of width $w$, one can obtain a nice tree decomposition with width~$w$ in time $\mathcal{O}(n)$, see \cite[Lemma 13.1.3]{kloks1994treewidth}.
Thus, from \Cref{thm:compute_approximate_tree_decomposition} we obtain
\begin{corollary}
\label{thm:compute_nice_tree_decomposition}
  There is an algorithm that, given a graph $G$ on $n$ vertices, constructs a nice tree decomposition of $G$ of width at most $2\tw(G)$ and on $\mathcal O(n)$ nodes in time~$2^{\mathcal{O}(\tw(G))}n$.
\end{corollary}

The \textit{treedepth} of an undirected graph can be defined in different ways, usually taking a certain decomposition of the graph.
We follow the definition given by Ne{\v{s}}et{\v{r}}il and Ossona de Mendez~\cite{NevsetvrilDeMendez2006}.
\begin{definition}
\label{def:treedepth}
  Let $G$ be an undirected graph with connected components $G_1,\hdots, G_q$.
  The \emph{treedepth} $\td(G)$ of $G$ is
  \begin{align*}
    \td(G) = \begin{cases} 
               1, &\mbox{if } |V(G)|=1; \\ 
               1 + \min\limits_{v \in V(G)} \td (G-v), &\mbox{if } q = 1 \mbox{ and } |V(G)| > 1; \\
               \max \{ \td(G_i) : G_i \mbox{ is a connected component of } G\},  & \mbox{otherwise.} 
             \end{cases}
  \end{align*}
  The \emph{treedepth} of a directed graph is the treedepth of its underlying undirected graph.
\end{definition}

Bodlaender et al. have shown relations between treewidth, pathwidth and treedepth.
\begin{proposition}[Bodlaender et al.~\cite{BodlaenderEtAl1995}]
  It holds $\tw(G) \leq \pw(G) \leq \td(G) - 1$ for any graph $G$.
\end{proposition}

Thus, the class of bounded treewidth graphs contains the class of bounded pathwidth graphs, which in turn contains the class of bounded treedepth graphs.

\paragraph{Parameterized complexity.}
For our hardness results we will use two different kinds of hardness.
The first one is $\mathsf{W}[1]$-hardness which, under the standard assumption $\mathsf{W}[1] \not = \mathsf{FPT}$, implies that there is no fixed-parameter  algorithm for problems of this type.
The other hardness considers compressions.
A \emph{polynomial compression} of a language $L$ into a language $Q$ is a polynomial-time computable mapping $\Phi: \Sigma^* \times \mathbb{N} \rightarrow \Sigma^*,$ $\Phi((x,k)) \mapsto y$ such that $((x,k) \in L \Leftrightarrow y \in Q)$ and  $|y| \leq k^{\mathcal O(1)}$ for all $(x,k) \in  \Sigma^* \times \mathbb{N}$.
Many natural parameterized problems do not admit polynomial compressions under the hypothesis $\mathsf{NP}\subsetneq \mathsf{coNP}/poly$.

Both types of hardness can be transferred to other problems by ``polynomial parameter transformations'', which were first proposed by Bodlaender et al.~\cite{BodlaenderEtAl2011}.
\begin{definition}
  Let $\Sigma$ be an alphabet.
  A \emph{polynomial parameter transformation (PPT)} from a parameterized problem $\Pi\subseteq \Sigma^*\times\mathbb N$ to a parameterized problem $\Pi'\subseteq \Sigma^*\times \mathbb N$ is a polynomial-time computable mapping $\Phi: \Sigma^* \times \mathbb{N} \rightarrow \Sigma^* \times \mathbb{N}, (x,k) \mapsto (x', k')$, such that $k' = k^{\mathcal O(1)}$, and $(x,k) \in \Pi \Leftrightarrow (x',k') \in \Pi'$ for all $(x,k)\in \Sigma^*\times\mathbb N$.
  Two parameterized problems are \emph{parameter-equivalent} if there are PPTs in both directions and the transformations additionally fulfill that $k' = k$.
\end{definition}

Note that polynomial parameter transformations are transitive.
Further, a PPT from $\Pi$ to~$\Pi'$ together with a polynomial compression for $\Pi'$ yields a polynomial compression for $\Pi$.
This can be used to rule out polynomial compressions.
\begin{proposition}[Kratsch and Wahlstr{\"o}m~\cite{KratschWahlstrom2013}]
\label{thm:PPTsShowW1andPolComp}
  Let $\Pi,\Pi'$ be parameterized problems.
  If there is a polynomial parameter transformation from $\Pi$ to $\Pi'$ and $\Pi$ admits no polynomial compression, then neither does $\Pi'$.
\end{proposition}

As mentioned in the introduction, the complexity class containing all fixed-parameter tractable problems is called $\mathsf{FPT}$.
Thus, one the one hand, if a parameterized problem $(P,k)$ is in  $\mathsf{FPT}$, then the parameterized problem $(P,k+k')$ is also in $\mathsf{FPT}$ for any parameter $k'$.
Conversely, on the other hand, if a parameterized problem $(P,k+k')$ is $\mathsf{W}[1]$-hard, then the parameterized problem $(P,k)$ is also $\mathsf{W}[1]$-hard. 

\section{Fixed-Parameter Algorithms for Finding Minimum Feasibility Blockers}
\label{sec:alg}
In this section, we develop fixed-parameter algorithms for different combinations of the parameters maximum size $k$ of a deletion set~($k$), number of positive/negative right-hand sides \mbox{($w_+/w_-$)}, treewidth/patwidth ($\tw(G)$, $\pw(G)$) and treedepth ($\td(G)$).
Recall that we always assume that the constraint matrix is of difference constraints.
By \Cref{thm:NDFASisequivalenttospecialMIS}, it suffices to consider the corresponding \NDFAS{} problem instead, therefore we use the directed graph terminology.

Before stating the algorithms, we give an overview how the feasibility of a set $S\subseteq A(G)$ can be checked, that is, how to decide if $G-S$ has indeed no negative cycles.
For this, we use the Floyd–Warshall algorithm \cite{floyd1962algorithm,warshall1962theorem,roy1959transitivite} that finds shortest paths between all pairs of vertices in a weighted digraph with no negative cycles.  
The algorithm can additionally detect whether there is a negative cycle in the directed graph and outputs a shortest one if exists.

\begin{proposition}[Floyd~\cite{floyd1962algorithm}, Roy~\cite{roy1959transitivite}, Warshall~\cite{warshall1962theorem}]
\label{thm:Moore_Bellman_Ford}
  Given an arc-weighted directed graph $(G, w)$ we can detect in time $\mathcal{O}(n^3)$ whether~$G$ contains a negative cycle.
  Moreover, we can recover in the same time a shortest (i.e. length-minimal) negative cycle if one exists.
\end{proposition}

Running this algorithm on $(G - S, w|_{A(G) \setminus S})$ and additionally checking whether $|S| \leq k$, allows us to verify a solution to some \NDFAS{} instance $(G, w, k)$.

\subsection{Algorithm for Bounded Treedepth and Solution Size}
The following result, showed by Ne\v{s}et\v{r}il and Ossona de Mendez \cite{NevsetvrilDeMendez2006}, exhibits that all cycles of a graph with bounded treedepth have bounded length. 
\begin{lemma}
\label{lem:path_and_cycle_length_bounded_by_treedepth}
  Let $G$ be a graph, $P$ be a path in $G$, and $C$ be a cycle of $G$.
  Then $|P| \leq 2^{\td(G)} - 1$ and $|C| \leq 2^{\td(G) - 1}$ holds.
\end{lemma}

The high-level idea is then to detect negative cycles iteratively in the directed graph, and to branch on the arcs of the cycle.
We now present the main algorithm of this section.
The run-time depends on the maximum length of a negative cycle in the given graph.
We will revisit it later on, when we derive other bounds on the length of negative cycles in $G$ depending on the parameter choices.
\begin{lemma}
\label{lem:algorithm_for_bounded_negative_cycle_length_plus_k}
  There is an algorithm solving any {\sc Negative Directed Feedback Arc Set} instance $(G, w, k)$ in time $\mathcal{O}( L^k n^3)$, where $L$ is an upper bound on the length of any negative cycle in $G$.
\end{lemma}
\begin{proof}
  We recursively call the following procedure with some potential partial solution $S \subseteq V(G)$ with $|S| \leq k$.
  The initial call is done with $S = \emptyset$.
  First, we check in time $\mathcal{O}(n^3)$ whether there is a negative cycle $C$ in $G-S$ and recover it via \Cref{thm:Moore_Bellman_Ford}.
  If there is no negative cycle and $|S| \leq k$, we return $S$ as a solution.
  If there is a negative cycle~$C$ and $|S| = k$, we give up on this branch.
  Otherwise, there is a negative cycle $C$ and $|S| < k$.
  Then for every $v \in V(C)$ we make a subroutine call with $S_v = S \cup \lbrace v \rbrace$.
  If our initial call with $S = \emptyset$ did not return a solution for any branch, we report that the instance has no solution.
  This finishes the description of our algorithm.
	
  As we only make subroutine calls in the case of $|S| < k$, we add only one vertex to~$S$, and we start with $|\emptyset| = 0 \leq k$, one can show by induction that all our calls indeed fulfill $|S| \leq k$.
  Now we argue for the correctness.
  To this end, we have to show that if there is a solution, we do indeed return a solution.
  As we only return sets $S$ such that $|S| \leq k$ and $G - S$ has no negative cycles, we just have to make sure that we return a set.
  We consider the variant of our algorithm where we do not return a solution early, but rather save it and return it at the end of the algorithm.
  Let $S^\star$ be an inclusionwise minimal solution to $(G, w, k)$.
  We are going to reconstruct a possible subroutine call sequence $\emptyset = S_0 \subsetneq S_1 \subsetneq \hdots \subsetneq S_{|S^\star|} = S^\star$ that is called by our algorithm.

  We prove by induction on $i$ that $S_i$ appears in one branch of our algorithm.
  We start with $S_0 = \emptyset$ which is our initial subroutine call.
  As long as $i < |S^\star|$, we have that $S_i \subsetneq S^\star$ and thus, by inclusion-wise minimality of $S^\star$, $G - S_i$ contains at least one negative cycle $C_i$.
  Now $S^\star$ is a solution, so we know that there is a $v_i \in S^\star \cap V(C_i)$.
  By $C_i$ being disjoint from~$S_i$, we have that $S_{i+1} = S_i \cup \lbrace v_i \rbrace$ is a strict superset of $S_i$.
  Thus we make a subroutine call with this $S_{i+1}$.
  This shows by induction that our modified algorithm considers the set $S^\star$ and our original algorithm returns a solution.
	
  For the run-time, note that at each sub-routine call we branch into $|C| \leq L$ many branches.
  We start with $|S| = 0$ and in each recursive subroutine call $|S|$ increases by one.
  As we stop once $|S| = k$, this means our recursion nests at most $k$ levels deep.
  This results in $L^k$ subroutine calls.
  In each call, we need $\mathcal{O}(n^3)$ time to check for negative cycles.
\end{proof}

The above observations enables us to give a fixed-parameter algorithm for combined parameter the treedepth and the solution size.
\begin{theorem}
\label{thm:algorithm_for_td_plus_k}
  There is an algorithm solving {\sc MinFB} in time $\mathcal{O}( 2^{k \td(G)} n^2m)$.
\end{theorem}
\begin{proof}
  By \Cref{thm:NDFASisequivalenttospecialMIS}, we can construct an instance of \NDFAS{} in polynomial time that has the same parameters $\td + k$ such that the original instances is a ``yes''-instance if and only if the constructed instance is.
  The theorem then follows by combining \Cref{lem:path_and_cycle_length_bounded_by_treedepth} and \Cref{lem:algorithm_for_bounded_negative_cycle_length_plus_k}. 
\end{proof}

\subsection{Algorithm for Bounded Number of Positive and Negative Arcs}
The goal of this section is to derive an algorithm for \NDFAS{} parameterized in $w_+$ and $w_-$.
In order to do so, we build on skew cuts in directed graphs.
\begin{definition}
  Let $G$ be a directed graph, let $p \in \mathbb{Z}_{\geq 0}$ be a non-negative integer, and let\linebreak $X_1, \hdots, X_p, Y_1, \hdots Y_p \subseteq V(G)$ be pairwise disjoint vertex sets of $G$.
  An \emph{$(X_1, \hdots, X_p) \to (Y_1, \hdots, Y_p)$-skew cut} is an arc set $S \subseteq A(G)$ such that there is no $X_i \to Y_j$-path in $G - S$ for any $(i,j)$ with $1 \leq j \leq i \leq p$.
\end{definition}

Then the \textsc{Skew Separator} problem is defined as follows.

\problemdef{Skew Separator}
  {A graph $G$, vertex sets $X_1, \hdots, X_p, Y_1, \hdots Y_p \subseteq V(G)$ and an integer $k \in \mathbb{Z}_{\geq 0}$.}
  {Find an $(X_1, \hdots, X_p) \to (Y_1, \hdots, Y_p)$-skew cut of size at most~$k$ or decide that no such skew cut exists.}

The \textsc{Skew Separator} problem was introduced by Chen et al.~\cite{ChenEtAl2008}, who devised an algorithm for it as a subroutine for their \textsc{Directed Feedback Vertex Set} algorithm.
In fact, they solve the vertex deletion variant, but it can be applied to the arc deletion variant by subdividing arcs.

\begin{proposition}[Chen et al.~\cite{ChenEtAl2008}]
\label{thm:solving_skew_cut}
  Any instance $(G, (X_1, \hdots, X_p), (Y_1, \hdots, Y_p), k)$ of {\sc Skew Separator} can be solved in time $\mathcal{O}(4^kkn^3)$.
\end{proposition}

The high-level description of our algorithm is the following.
We first guess the intersection between a solution and the set of non-zero arcs.
Afterwards, we focus on the directed graph induced by zero-weight arcs only.
For any solution~$S$, the graph $G-S$ has a feasible potential, that in turn defines an ordering of the end-points of the non-zero arcs.
We use this observation to show that solutions can be obtained by solving a \textsc{Skew Separator} problem in the zero-arc graph.
As we do not know the feasible potential or the exact  \textsc{Skew Separator} instance, our algorithm guesses among all potential orderings of the endpoints of the non-zero arcs.

The main graph we work with is constructed from $G$ by removing the arcs of $A_{\neq 0}$ and splitting their endpoints.
\begin{definition}
  Let $(G, w)$ be an arc-weighted directed graph and denote by $Z$ the set of vertices $z \in V(G)$ with $(\delta^+(z) \cup \delta^-(z)) \cap A_{\neq 0}(G) \neq \emptyset$.
  The \emph{zero-weight propagation graph} of $G$ is the graph $\overrightarrow{G}_0$ obtained from $G$ by deleting~$A_{\neq 0}$ and splitting every $z \in Z$ into two vertices $z^+$ and~$z^-$, where~$z^+$ inherits the outgoing arcs and $z^-$ inherits the incoming arcs of $z$.
  In this context we denote for every subset $Y \subseteq Z$ by $Y^+$ and $Y^-$ the set of all $z^+$'s and $z^-$'s with $z \in Y$, respectively.
\end{definition}

By abuse of notation, we will identify the arcs in $\overrightarrow{G}_0$ with those of $A_0$.
We are now ready state the main observation of this chapter.
\begin{lemma}
\label{lem:zero_arc_NDFAS_to_Skew_Separator}
  Let $(G, w)$ be an arc-weighted directed graph and let $S \subseteq A_0(G)$ be a negative directed feedback arc set of $(G,w)$.
  Let $Z$ be the endpoints of arcs in $A_{\neq 0}(G)$ and let $\overrightarrow{G}_0$ be the zero-weight propagation graph of $G$.
  Then there is an ordered partition $(Z_1, \hdots, Z_p)$ of $Z$ such that
  \begin{enumerate}
    \item $S$ is a $(Z_1^+, \hdots, Z_p^+, \emptyset) \to (\emptyset, Z_1^-, \hdots, Z_p^-)$-skew cut in $\overrightarrow{G}_0$, and
    \item any $(Z_1^+, \hdots, Z_p^+, \emptyset) \to (\emptyset, Z_1^-, \hdots, Z_p^-)$-skew cut in $\overrightarrow{G}_0$ is a negative directed feedback arc set of~$(G,w)$.
  \end{enumerate}
\end{lemma}
\begin{proof}
  As $S$ is a negative directed feedback arc set of $(G,w)$, the graph $G' = G - S$ has a feasible potential $\pi$ with respect to $w|_{A(G')}$.
  We define an ordered partition $(Z_1, \hdots, Z_p)$ of $Z$ by using the potential $\pi$.
  Two vertices $z$ and $z'$ belong to the same partition class if and only if $\pi(z) = \pi(z')$.
  This way each partition class admits a unique value that is possessed by all of its elements as a potential. 
  The partition classes are then ordered according to these values in a decreasing order.
	
  We now verify the statements of the theorem.
  First we check, whether $S$ is a $(Z_1^+, \hdots, Z_p^+, \emptyset) \to (\emptyset, Z_1^-, \hdots, Z_p^-)$-skew cut in $\overrightarrow{G}_0$.
  Suppose, for sake of contradiction, that this does not hold.
  Then there is an $x^+ \to y^-$-path in $\overrightarrow{G}_0$ with $x^+ \in Z_i^+$ and $y^- \in Z_j^-$ such that $j < i$.
  By the ordering of the partition classes, we get $\pi(y) > \pi(x)$.
  Now consider the corresponding $x \to y$-path $P$ in $G - S$, where $x$ and $y$ are the original vertices of $x^+$ and $y^-$ before they got split.
  By construction of~$\overrightarrow{G}_0$, $P$ contains arcs of weight zero only.
  As $\pi$ is a feasible potential, $\pi(u) - \pi(v) = \pi(u) - \pi(v) + w(a) \geq 0$ for each arc $a = (u,v)$ of $P$.
  Summing this up over all arcs in $P$, we get $\pi(x) - \pi(y) = \sum_{a = (u,v) \in A(P)} \pi(u) - \pi(v) \geq 0$.
  In other words $\pi(y) \leq \pi(x)$, a contradiction.
	
  It remains to verify that any $(Z_1^+, \hdots, Z_p^+, \emptyset) \to (\emptyset, Z_1^-, \hdots, Z_p^-)$-skew cut in~$\overrightarrow{G}_0$ is a negative directed feedback arc set of $(G, w)$.
  Let $S^\star$ be such a skew cut.
  We try to prove that~$S^\star$ is a negative directed feedback arc set of $(G,w)$ by constructing a feasible potential $\pi^\star$ of $G^\star = G - S^\star$ with respect to $w|_{A(G^\star)}$.
  We may assume that $\max_{v \in V(G)} \pi(v) \leq 0$ as the potentials can be shifted by the same value on every vertex.
  We define $\pi^\star$ as
  \begin{equation*}
    \pi^\star(v) = \begin{cases}
                     \pi(v) &\text{, if $v \in Z$}\\
					        \min \lbrace \pi(z) \mid \text{ there is a $z^+ \to v$-path in $\overrightarrow{G}_0 - S^\star$} \rbrace	&\text{, otherwise}
                   \end{cases}
  \end{equation*}
  where we define the minimum to be $0$ if there is no $z^+ \to v$-path in $\overrightarrow{G}_0 - S^\star$ for any $z^+ \in Z^+$.
	
  We show that $\pi^\star$ is a feasible potential in $G^\star$.
  As $A_{\neq 0}(G^\star) = A_{\neq 0}(G)$, $u,v\in Z$ holds for each arc $a = (u,v) \in A_{\neq 0}(G^\star)$.
  This implies $\pi^\star(u) - \pi^\star(v) + w(a) = \pi(u) - \pi(v) + w(a) \geq 0$.
  Hence, it remains to check arcs $a = (u,v) \in A_0(G^\star) = A_0(G) \setminus S^\star$.
  If $u \in Z$, then the arc $(u^+, v^-)$ or $(u^+, v)$ exists in $\overrightarrow{G}_0 - S^\star$.
  In the former case we again have $u,v \in Z$ and the potential is feasible for this arc.
  Otherwise, we have that there is an $u^+ \to v$-path in $\overrightarrow{G}_0 - S^\star$ and thus $\pi^\star(v) \leq \pi(u) = \pi^\star(u)$.
  As $w(a) = 0$, this implies $\pi^\star(u) - \pi^\star(v) + w(a) \geq 0$.
  So we can assume that $u \notin Z$.
  If additionally $u$ does not have a $z^+ \to u$-path in $\overrightarrow{G}_0 - S^\star$, then $\pi^\star(u) = 0$ and by $\max_{v \in V(G)} \pi(v) \leq 0$, we have that $\pi^\star(v) \leq 0 = \pi^\star(u)$, implying $\pi^\star(u) - \pi^\star(v) + w(a) \geq 0$.
  Now we know that $u \notin Z$ and there is a $z^+ \to u$-path in $\overrightarrow{G}_0 - S^\star$ for some $z \in Z$.
  We distinguish between the two cases $v \notin Z$ and $v \in Z$.
  If $v \notin Z$, then the arc $(u, v)$ exists in $\overrightarrow{G}_0 - S^\star$.
  Hence, any $z^+ \to u$-path can be prolonged to a $z^+ \to v$-path by adding $(u,v)$ at the end.
  Thus $\pi^\star(v) \leq \pi^\star(u)$, and $\pi^\star(u) - \pi^\star(v) + w(a) \geq 0$.
  If $v \in Z$, we have that the arc $(u, v^-)$ exists in $\overrightarrow{G}_0 - S^\star$.
  Any $z^+ \to u$-path can be prolonged to a $z^+ \to v^-$-path.
  By definition of our \textsc{Skew Separator} instance, we have that for all such paths in $\overrightarrow{G}_0 - S^\star$, we have $\pi(z) \geq \pi(v)$.
  This implies $\pi^\star(v) \leq \pi^\star(u)$, and so $\pi^\star(u) - \pi^\star(v) + w(a) \geq 0$.
  Hence, $\pi^\star$ is a feasible potential for $G - S^\star$ and thus $S^\star$ is a negative directed feedback arc set of $(G,w)$.
\end{proof}

We are now going to use this result to give an algorithm solving the NDFAS problem parameterized by $w_+$ and $w_-$, and thus we solve our original problem.
\begin{theorem}
\label{thm:algorithm_for_bounded_number_of_non_zero_arcs}
  There is an algorithm that solves any {\sc MinFB} instance $(G, w, k)$ in time\linebreak $2^{\mathcal{O}\left( (w_+ + w_-) \log (w_+ + w_-)\right)}n^{\mathcal O(1)}$.
\end{theorem}
\begin{proof}
  By \Cref{thm:NDFASisequivalenttospecialMIS}, we can construct an instance of \NDFAS{} in polynomial time that has the same parameters $w_+ + w_-$ such that the original instances is a ``yes''-instance if and only if the constructed instance is. 

  If $w_- \leq k$, return the set $A_-(G)$ of negative arcs.
  Otherwise, for every subset $A' \subseteq A_{\neq 0 }(G)$ of size at most $k$ do the following.
  For the graph $G' = G - A'$ with $Z'$ being endpoints of non-zero arcs $A_{\neq 0}(G')$ and for every ordered partition $(Z_1, \hdots, Z_p)$ of $Z'$, we call the algorithm for  \textsc{Skew Separator} (see \Cref{thm:solving_skew_cut}) on
  \begin{equation*}
    (\overrightarrow{G'}_0, (Z_1^+, \hdots, Z_p^+, \emptyset), (\emptyset, Z_1^-, \hdots, Z_p^-), k - |A'|).
  \end{equation*}
  If the algorithm returns a set $S'$, we check whether $A' \cup S'$ is a negative directed feedback arc set for $G$ and if it is, we return it.
  If the algorithm does not find a negative directed feedback arc set for any $A'$ and ordered partition $(Z_1, \hdots, Z_p)$ of $Z'$, we return that the problem has no negative directed feedback arc set of size at most $k$.
	
  As we check any solution we return for correctness, we just have to verify that we always return a negative directed feedback arc set, if there is any.
  Let $S$ be a negative directed feedback arc set for $(G, w)$.
  For $A' = S \cap A_{\neq 0}(G)$ we know that $S' = S \setminus A'$ is a negative directed feedback arc set for $G' = G - A'$, as $(G - A') - S' = G - S$ which contains no negative cycles.
  Now $S' \cap A_{\neq 0}(G') = \emptyset$, and thus by \Cref{lem:zero_arc_NDFAS_to_Skew_Separator} there is an ordered partition $(Z_1, \hdots, Z_p)$ of~$Z'$ such that $S'$ is a $(Z_1^+, \hdots, Z_p^+, \emptyset) \to (\emptyset, Z_1^-, \hdots, Z_p^-)$-skew cut in $\overrightarrow{G'}_0$.
  Thus, we know that $(\overrightarrow{G'}_0, (Z_1^+, \hdots, Z_p^+, \emptyset), (\emptyset, Z_1^-, \hdots, Z_p^-), k - |A'|)$ is a ``yes''-instance.
  Moreover, \Cref{lem:zero_arc_NDFAS_to_Skew_Separator} tells us that any solution $S^\star$ to this \textsc{Skew Separator} instance our algorithm call finds, is a negative directed feedback arc set for $G'$.
  Thus, $A' \cup S^\star$ is a negative directed feedback arc set of $G$, as $G - (A' \cup S^\star) = G' - S^\star$.
  Also, $A' \cup S^\star$ has size at most $k$.
  Hence, our algorithm finds a negative directed feedback arc set of the given size if one exists.
	
  For the run-time note that by returning the trivial solution $A_-$ if $w_- \leq k$, we have either a linear run-time for constructing $A_-$ or $k < w_-$.
  Now we have at most $2^{w_+ + w_-}$ possible subsets~$A'$ of $A_{\neq 0}$.
  The set $Z'$ contains at most $2 w_+ + 2w_-$ vertices and the number of its ordered partitions can be bounded by $(2w_+ + 2w_-)! \cdot 4^{w_+ + w_-}$ (orders of $Z'$ times the choice for every element whether to start a new subset there).
  Thus, the overall run-time is
  \begin{equation*}
    \mathcal{O}(2^{w_+ + w_-} \cdot (2w_+ + 2w_-)! \cdot 4^{w_+ + w_-} \cdot (4^kkn^3 + n^{\mathcal O(1)})),
  \end{equation*}
  which, using $k < w_-$, can be rewritten as $2^{\mathcal{O}\left( (w_+ + w_-) \log (w_+ + w_-)\right)}n^{\mathcal O(1)}$.
\end{proof}

\subsection{Number of Positive or Negative Arcs and Feasible Potentials}
This section is dedicated to \NDFAS{} with arc weights in $\lbrace -1, 0, 1\rbrace$ and how the number of positive or negative arcs influences the feasible potential of the solution. 
We use this results for the algorithms in the following section.
We show that for both parameters~$w_-$ and~$w_+$, the solution must have an integral feasible potential in $[0, w_-]$ and $[0, w_+]$, respectively.
Unfortunately, the latter only holds for graphs that are strongly connected after removing the solution.
Moreover, we also prove that for the parameter treedepth $\td(G)$, the solution has a feasible potential in $[0, 2^{\td(G)}]$

We first state a folklore technique for constructing a feasible potential of a graph.
\begin{lemma}
\label{lem:constructing_a_feasible_potential}
  Let $(G, w)$ be an arc-weighted directed graph that contains no negative cycles.
  Let~$G^\star$ be the directed graph obtained from $G$ by introducing a new vertex $s^\star$ that is connected to every original vertex $v \in V(G)$ by an arc $(s^\star, v)$ of weight~$0$.
  For every $v \in V(G^\star)$, let $\pi^\star(v)$ be the weight of a minimum-weight $s^\star \to v$-path $P_v$.
  Then $\pi = \pi^\star|_{V(G)}$ is a feasible potential for~$G$.
\end{lemma}
\begin{proof}
  We claim that $\pi^\star$ is a feasible potential for $G^\star$.
  Suppose, for sake of contradiction, that this is not the case.
  Then there is an arc $a = (u, v)$ for which $\pi^\star(u) - \pi^\star(v) + w(a) < 0$.
  In other words, we have that $w(P_u \circ a) < w(P_v)$.
  As $P_v$ is a minimum weight $s^\star \to v$-path, we have that $P_u \circ a$ is only an $s^\star \to v$-walk.
  Moreover, $P_u \circ a$ has to contain a negative closed walk, as otherwise it would contain an $s^\star \to v$-path of weight less than $P_v$.
  As $P_u$ is a path, this closed walk is indeed a cycle $O$ containing the arc $a$.
  Now $s^\star$ has only outgoing arcs, thus the cycle lies in $G^\star - s^\star = G$, a contradiction to $G$ having no negative cycles.
  So~$\pi^\star$ is a feasible potential for~$G^\star$.
  By $G = G^\star - s^\star$ being a subgraph of~$G^\star$, we get that the function $\pi = \pi^\star|_{V(G)}$ is a feasible potential for $G$.
\end{proof}

With this construction in place, we can make our observations about integral feasible potentials with few distinct values.
\begin{lemma}
\label{lem:few_normalized_negative_arc_weights_have_bounded_potential}
  Let $G$ be a directed graph, and $w: A(G) \to \lbrace -1, 0, 1\rbrace$ be a weight function.
  If $G$ contains no negative cycles, then there is a feasible potential $\pi: V(G) \to [0, w_-]$.
\end{lemma}
\begin{proof}
  First, we use \Cref{lem:constructing_a_feasible_potential} to construct the auxiliary graph $G^\star$ with its function~$\pi^\star$.
  Now note that any $s^\star \to v$-path can contain at most $w_-$-many negative arcs and each of those has weight $-1$.
  Thus, we have that $\pi^\star(v) = w(P_v) \geq -w_-$ for all $v \in V(G^\star)$.
  Moreover, the arc $(s^\star, v)$ always forms an $s^\star \to v$-path of weight $0$.
  By all arc-weights being integral, we have that~$\pi^\star$ is an integer function with values in $[-w_-, 0]$.
  Now, by \Cref{lem:constructing_a_feasible_potential}, $\pi = \pi^\star|_{V(G)}$ is a feasible potential for $G$, which has only integer values in $[-w_-, 0]$.
  Shifting all values by $w_-$ shows the statement.
\end{proof}

\begin{lemma}
\label{lem:few_normalized_positive_arc_weights_in_strong_graphs_have_bounded_potential}
  Let $G$ be a strongly connected directed graph, and $w: A(G) \to \lbrace -1, 0, 1\rbrace$ be a weight function.
  If $G$ contains no negative cycles, then there is a feasible potential $\pi: V(G) \to [0, w_+]$.
\end{lemma}
\begin{proof}
  First, we use \Cref{lem:constructing_a_feasible_potential} to construct the auxiliary graph $G^\star$ with its function $\pi^\star$.
  Suppose, for sake of contradiction, that an $s^\star \to v$-path $P_v$ has weight less than $-w_+$.
  Let $P$ be the $u \to v$-subpath of $P_v$ that contains everything but the first arc $a = (s^\star, u)$ of $P_v$.
  As $w(a) = 0$ we know that $w(P) < -w_+$ and moreover $P$ exists in $G$.
  Now,~$G$ is strongly connected, and therefore there is a $v \to u$-path $Q$ in $G$.
  As $G$ contains only $w_+$-many positive arcs and all of them have weight~$1$, we get $w(Q) \leq w_+$.
  Hence, $P \circ Q$ is a closed walk of weight $w(P) + w(Q) < w_+ - w_+ = 0$.
  This negative closed walk contains a negative cycle, in contradiction to $G$ having none of those.
  Thus, we have that $\pi^\star(v) = w(P_v) \geq -w_+$ for all $v \in V(G^\star)$.
  Moreover, the arc $(s^\star, v)$ always forms an $s^\star \to v$-path of weight $0$.
  By all arcs weights being integral, we have that $\pi^\star$ is an integer function with values in $[-w_+, 0]$.
  Now, by \Cref{lem:constructing_a_feasible_potential}, $\pi = \pi^\star|_{V(G)}$ is a feasible potential for $G$, which has only integer values in $[-w_+, 0]$.
  Shifting all values by $w_+$ shows the theorem.
\end{proof}

\begin{lemma}
\label{lem:bounded_treedepth_graphs_have_bounded_potential}
  Let $G$ be a directed graph, and $w: A(G) \to \lbrace -1, 0, 1\rbrace$ be a weight function.
  If~$G$ contains no negative cycles, then there is a feasible potential $\pi: V(G) \to [0, 2^{\td(G)}]$.
\end{lemma}
\begin{proof}
  First, we use \Cref{lem:constructing_a_feasible_potential} to construct the auxiliary graph $G^\star$ with its function~$\pi^\star$.
  By \Cref{lem:path_and_cycle_length_bounded_by_treedepth} we have that any path in $G$ contains at most $2^{\td(G)}$ many arcs.
  Moreover, any $s^\star \to v$-path consists of an arc $(s^\star, u)$ of weight $0$, followed by an $u \to v$-path in~$G$.
  Thus, any $s^\star \to v$-path can contain at most $2^{\td(G)}$ many negative arcs and each of those has weight~$-1$.
  Thus, we have that $\pi^\star(v) = w(P_v) \geq -2^{\td(G)}$ for all $v \in V(G^\star)$.
  Moreover, the arc $(s^\star, v)$ always forms an $s^\star \to v$-path of weight $0$.
  By all arcs weights being integral, we have that $\pi^\star$ is an integer function with values in $[-2^{\td(G)}, 0]$.
  Now, by \Cref{lem:constructing_a_feasible_potential}, $\pi = \pi^\star|_{V(G)}$ is a feasible potential for $G$, which has only integer values in $[-2^{\td(G)}, 0]$.
  Shifting all values by $2^{\td(G)}$ shows the theorem.
\end{proof}

\subsection{Dynamic Program for Treewidth and Bounded Feasible Potentials}
The aim of this section is to apply our findings on feasible potentials with few different values to a dynamic program utilizing the treewidth.
The overall approach is computing a nice tree decomposition of our graph (see \Cref{def:nice_tree_decomposition} and \Cref{thm:compute_nice_tree_decomposition}) and then guessing via a dynamic program the feasible potential of some solution on each bag of the tree decomposition.
The deleted arcs are then exactly those that violate the guessed potential.
However, when we parameterize in the number of positive arcs, we are only guaranteed a feasible potential with few different values for every strong component of $G - S$.
Therefore we also have to guess a topological order of the strong components of $G - S$ when restricted to the bags of the tree decomposition.
To handle both cases simultaneously, we introduce a set $\mathcal{C}$ which contains ordered partitions of $V(G)$ that represent components for which the guessed potential could be feasible.
For the parameter $w_-$ the single partition consisting of all vertices suffices as set $\mathcal{C}$.
In the $w_+$ case, $\mathcal{C}$ has to contain a topological order of the strong components of $G - S$.
In this case we choose~$\mathcal{C}$ as all ordered partitions of $V(G)$.
We generalize the properties we need for $\mathcal{C}$ and our feasible potential, to unify both choices.

\begin{definition}
  Let $G$ be a directed graph and $w:A(G) \to \mathbb{Z}$ be a weight function.
  We call an arc set $S \subseteq A(G)$ \emph{$((C_1, \hdots, C_t), \pi)$-feasible} for some ordered partition $(C_1, \hdots, C_t)$ of $V(G)$ and some $\pi: V(G) \to \mathbb{Z}$, if for all arcs $a = (p,q) \in A(G) \setminus S$ with $p \in C_i$ and $q \in C_j$ we have either $i < j$ or  $i = j$ and $\pi(p) - \pi(q) + w(a) \geq  0$.
	
  For an ordered partition $(C_1, \hdots, C_t)$ of $V(G)$ and some $U \subseteq V(G)$, we call an ordered partition $(C'_1, \hdots, C'_{t'})$ of $U$ the \emph{projection of $(C_1, \hdots, C_t)$ on $U$}, denoted by $(C_1, \hdots, C_t)|_U$, if for all $u \in C_i \cap C'_p$ and $v \in C_j \cap C'_q$ we have $i < j$ if and only if $p < q$.
  For a set $\mathcal{C}$ of ordered partitions of $V(G)$, we denote by $\mathcal{C}|_U$ the set~$\lbrace C\mid_U \mid C \in \mathcal{C}\rbrace$.
	
  Similarly, for some $U \subseteq V(G)$ and a pair $(C, \pi)$ consisting of an ordered partition~$C$ of $V(G)$ and $\pi: V(G) \to \mathbb{Z}$, we call $(C', \pi')$ the \emph{projection of $(C, \pi)$ on $U$} for $C'$ being the projection of~$C$ on $U$ and $\pi' = \pi|_U$.
	
  Let $(T,\mathcal B)$ be a tree decomposition of $\langle G\rangle $ with vertex bags $(B_x)_{x \in V(T)}$.
  We say that a set $\mathcal{C}$ of ordered partitions of $V(G)$ is \emph{$(T,\mathcal B)$-compatible} if for every $C, C' \in \mathcal{C}$ with $C|_{B_x}= C'|_{B_x}$, we have that there is a $C^\star \in \mathcal{C}$, with $C^\star|_{V(G_x)} = C|_{V(G_x)}$ and $C^\star|_{V(G) \setminus (V(G_x) \setminus B_x)} = C'|_{V(G) \setminus (V(G_x) \setminus B_x)}$.
\end{definition}

We are now able to state our general algorithm.
\begin{lemma}
\label{lem:dynamic_program_for_treewidth_normalized_weights_and_bounded_number_of_positive_or_negative_arcs}
  Let $(G, w, k)$ be a {\sc Negative DFAS} instance with weights $w: A(G) \to \mathbb{Z}$.
  Given a nice tree decomposition~$(T,\mathcal B)$ of $G$ with vertex bags $(B_x)_{x \in V(T)}$ of width $\mathcal{O}(\tw(G))$, a set $\mathcal{C}$ of ordered partition of $V(G)$ compatible with $(T,\mathcal B)$, and two integers $a \leq b$, there is an algorithm that in time $f(\mathcal{C}, T) 2^{\mathcal{O}(\tw(G) \log (b- a))}\cdot (n + m)$  computes a minimum-size negative directed feedback arc set that is $(C, \pi)$-feasible for some $C \in \mathcal{C}$ and $\pi: V(G) \to \mathbb{Z} \cap [a,b]$.
  Here, $f(\mathcal{C}, T)$ is the maximum time needed to enumerate~$\mathcal{C}|_{B_x}$ for any $x \in V(T)$.
\end{lemma}
\begin{proof}
  We compute a dynamic programming table $D$ via dynamic program on $V(T)$ from the leaves upwards.
  The table has an entry $D[x, \kappa]$ for every $x \in V(T)$ and $\kappa$, where $\kappa$ consists of an ordered partition $C \in \mathcal{C}|_{B_x}$ and an integer-valued function $\pi: B_x \to \mathbb{Z} \cap [a, b]$.
  The entry $D[x, \kappa]$ contains an arc set $S_{x, \kappa} \subseteq A(G_x)$ with $G_x$ being the graph induced by the subtree decomposition of $T$ rooted at $x$.
  Our dynamic program is motivated by the following observation:	
  \begin{claim}
  \label{claim:minimal_kappa_feasible_solutions_inside_bags}
	Let $S \subseteq A(G_x)$ be a minimum-size set that is $(C, \pi)$-feasible for some \mbox{$C \in \mathcal{C}|_{V(G_x)}$} and $\pi: V(G_x) \to \mathbb{Z} \cap [a,b]$.
	Let $( (C^x_1, \hdots, C^x_{t^x}), \pi^x)$ be the projection of $(C, \pi)$ to~$B_x$.
	Then $S$  contains exactly those arcs $a = (p,q) \in A(G[B_x])$ with $p \in C^x_i$ and $q \in C^x_j$ such that either $j < i$ or  $i = j$ and $\pi^x(p) - \pi^x(q) + w(a) < 0$.
  \end{claim}
  \begin{claimproof}
	The set $S$ necessarily contains these arcs as otherwise it would not be $(C, \pi)$-feasible.
	Suppose, for sake of contradiction, that $S$ contains further arcs. Removing such an arc from $S$ does not violate $(C, \pi)$-feasibility and thus results in a smaller $S$, contradicting the minimality of $S$.
	This completes the proof of \Cref{claim:minimal_kappa_feasible_solutions_inside_bags}.
  \end{claimproof}
	
  The entries are computed for each $x \in V(T)$ depending on its node type.
  \begin{description}
	\item[Leaf nodes.] As $B_x = \emptyset$, the only choice for $\kappa$ are the empty partition $()$ and $\pi$ being the empty function $\varepsilon$ and we set $D[x, ((), \varepsilon)] = \emptyset$.
	\item[Introduce nodes.] Let $v$ be the newly introduced vertex to the bag $B_x$ and let $y$ be the child of $x$ in $T$.
	  Let $\kappa$ consist of an ordered partition $(C_1, \hdots C_t) \in \mathcal{C}|_{B_x}$ and an integer-valued function $\pi: B_x \to \mathbb{Z} \cap [a, b]$.
    Define $S_v$ to be the arcs $a = (p,q) \in \delta(v)$ with $p \in C_i$ and $q \in C_j$ such that either $j < i$ or  $i = j$ and $\pi(p) - \pi(q) + w(a) < 0$.
    Then we set
    \begin{equation*}
      D[x, \kappa] = S_v \cup D[y, \kappa'],
    \end{equation*}
    where $\kappa'$ is the projection of $\kappa$ to $B_y$.
    \item[Forget nodes.] Let $y$ be the child of $x$ in~$T$.
      Let $\kappa$ consist of an ordered partition $C \in \mathcal{C}|_{B_x}$ and an integer-valued function $\pi: B_x \to \mathbb{Z} \cap [a, b]$.
      Denote by $K(y, \kappa)$ the set of all $\kappa^\star$ consisting of ordered partitions $C^\star \in \mathcal{C}|_{B_y}$ and $\pi^\star: B_y \to \mathbb{Z} \cap [a, b]$, whose projection to $B_x$ is~$\kappa$.
      Then we set
      \begin{equation*}
        D[x, \pi] = \operatorname*{arg min}\limits_{\substack{S = D[y, \kappa^\star] \\ \kappa^\star \in K(y, \kappa)}} |S| \enspace .
      \end{equation*}
			
    \item[Merge nodes.] Let $y_1$ and $y_2$ the two children of $x$.
      Then for every possible choice of~$\kappa$, we set
      \begin{equation*}
        D[x, \kappa] = D[y_1, \kappa] \cup D[y_2, \kappa] \enspace .
      \end{equation*}
	\end{description}
	
	We claim that at the root $r$ of $T$, the unique entry $D[r, ((), \varepsilon)]$, where $\varepsilon$ denotes the empty function, contains a negative directed feedback arc set for $G$ of minimum size.
	First note that this is indeed the unique entry as $B_r = \emptyset$.	
	We are going to prove the stronger statement that the entry $D[x, \kappa]$ contains a minimum-size set $S \subseteq A(G_x)$ that is $(C, \pi)$-feasible for some $C \in \mathcal{C}|_{V(G_x)}$, and $\pi: V(G_x) \to \mathbb{Z} \cap [a,b]$ with the projection of $(C, \pi)$ to $B_x$ being $\kappa$.
	That is, we have to prove, that the $D[x, \kappa]$'s indeed contain a set that is $(C, \pi)$-feasible as above and that it has  minimum size among those.
	We first focus on the $(C, \pi)$-feasibility, which we prove by induction from leaf vertices to the root.
	Let $x \in V(T)$ be a node, such that the feasibility property holds for all vertices $y$ that are below $x$ in $T$.
    \begin{description}
      \item[Leaf nodes.] As $G_x$ is the empty graph, the statement holds trivially.
	
      \item[Introduce nodes.] For the keys $\kappa$, we enumerate all ordered partitions $C \in \mathcal{C}|_{B_x}$.
      Let $C^\uparrow \in \mathcal{C}$ be an ordered partition that projects to $C$ in $B_x$.
      The set $D[y, \kappa']$ we chose is, by induction, $(C'|_{V(G_y)}, \pi')$-feasible on $G_y$.
      By $T$-compatibility of $\mathcal{C}$ we have that there is a $C^\star \in \mathcal{C}$, with $C^\star|_{V(G_y)} = C'|_{V(G_y)}$ and $C^\star|_{V(G) \setminus (V(G_y) \setminus B_y)} = C^\uparrow|_{V(G) \setminus (V(G_y) \setminus B_y)}$.
      Thus, by construction, our set is $(C^\star|_{V(G_x)}, \pi^\star|_{V(G_x)})$-feasible, where $\pi^\star(v) = \pi(v)$ and $\pi^\star(u) = \pi'(u)$ for all $u \in V(G_x) \setminus \lbrace v \rbrace$.
			
      \item[Forget nodes.] As $G_x = G_y$ the statement holds trivially.
		 
      \item[Merge nodes.] Let $y_1$ and $y_2$ be the two children of $x$. By induction, let $D[y_i, \kappa]$ be $(C_i|_{V(G_{y_i})}, \pi_i)$-feasible.
      By $T$-compatibility of $\mathcal{C}$ we have that there is a $C^\star \in \mathcal{C}$ such that $C^\star|_{V(G_x)} = C_1|_{V(G_x)} \textsf{ and } C^\star|_{V(G) \setminus (V(G_x) \setminus B_x)} = C_2|_{V(G) \setminus (V(G_x) \setminus B_x)}$.
      With $\pi^\star(u) = \pi_1(u)$ for all $u \in V(G_{y_1})$ and $\pi^\star(u) = \pi_2(u)$ otherwise, we get that our chosen set is $(C^\star|_{V(G_x)}, \pi^\star|_{V(G_x)})$-feasible.
	\end{description}
	
	Now assume for contradiction that our set is not the minimum choice among the $(C, \pi)$-feasible ones.
	Then there is a node $x \in V(T)$ such that our statement holds for all nodes~$y$ in the subtree of $T$ rooted at $x$ but not for $x$ (with $x$ possibly being a leaf node and the set of other nodes being empty).
	In particular, we made the minimum choice for all these nodes~$y$ and all $\kappa$'s.
	We make a case distinction based on the type of $x$.
    \begin{description}
      \item[Leaf nodes.] $G_x$ is empty and thus the empty set is the right choice of~$D[x, ( (), \varepsilon)]$.
		
      \item[Introduce nodes.] Assume there is a set of smaller size $S^\star$ that is $(C^\star, \pi^\star)$-feasible for some $C^\star = (C^\star_1, \hdots C^\star_t) \in \mathcal{C}|_{V(G_x)}$ and $\pi^\star: V(G_x) \to \mathbb{Z} \cap [a,b]$ with $(C^\star, \pi^\star)$'s projection to $B_x$ being $\kappa$.
      Choose $S^\star$ smallest possible among all such choices.
      By \Cref{claim:minimal_kappa_feasible_solutions_inside_bags}, we have that $S^\star \cap A(G[B_x])$ contains exactly those arcs $a = (p,q) \in A(G[B_x])$ with $p \in C^\star_i$ and $q \in C^\star_j$ such that either $j < i$ or  $i = j$ and $\pi^\star(p) - \pi^\star(q) + w(a) < 0$.
      In especially, $S^\star \cap \delta(v)$ and $D[x, \kappa] \cap \delta(v)$ are identical by $(C^\star, \pi^\star)$ projecting to $\kappa$.
      Thus, $S^\star \setminus \delta(v)$ is $(C^\star|_{V(G_y)}, \pi^\star|_{V(G_y)})$-feasible by $G_y - (S^\star \setminus \delta(v)) \subseteq G_x - S^\star$ and projects down to the same $\kappa'$ as $\kappa$ in $B_y$.
      But $S^\star \setminus \delta(v)$ is of size smaller than $D[y, \kappa|_{B_y}]$, a contradiction to $D[y, \kappa|_{B_y}]$ containing the smallest such arc set.
			
      \item[Forget nodes.] We have that $G_x = G_y$ and thus any candidate for $D[x, \kappa]$ is also a candidate for $D[y, \kappa^\star]$ by extending the $\kappa^\star$ in a way that matches the candidate.
      By taking the minimum over the possible choices of $\kappa^\star$ and the claim holding for~$y$, we get that $D[x, \pi]$ is indeed such a set of minimum size.
			
      \item[Merge nodes.]	Assume there is a set of smaller size $S^\star$ that is $(C^\star, \pi^\star)$-feasible for some $C^\star = (C^\star_1, \hdots C^\star_{t^\star}) \in \mathcal{C}|_{V(G_x)}$ and $\pi^\star: V(G_x) \to \mathbb{Z} \cap [a,b]$ with $(C^\star, \pi^\star)$'s projection to $B_x$ being $\kappa$.
      Choose $S^\star$ smallest possible among all such choices.
      By \Cref{claim:minimal_kappa_feasible_solutions_inside_bags}, we have that $S^\star \cap A(G[B_x])$ contains exactly those arcs $a = (p,q) \in A(G[B_x])$ with $p \in C^\star_i$ and $q \in C^\star_j$ such that either $j < i$ or  $i = j$ and $\pi^\star(p) - \pi^\star(q) + w(a) < 0$.
      In especially, we have that $S^\star \cap A(G[B_x])$, $D[y_1, \kappa]  \cap A(G[B_x])$ and $D[y_2, \kappa]  \cap A(G[B_x])$ are identical (as $B_x = B_{y_1} = B_{y_2}$).
      Thus, from $|S^\star| < |D[x, \kappa]|$, we have that $|S^\star \cap A(G_{y_i})| < |D[y_i, \kappa]|$ for some $i \in \lbrace 1, 2 \rbrace$.
      As $S^\star \cap A(G_{y_i})$ is a candidate for $D[y_i, \kappa]$, this is a contradiction to the minimality of $D[y_i, \kappa]$.
  \end{description}
	
  Thus, our algorithm computes a minimum-size negative directed feedback arc set of~$(G,w)$.
  For the run-time, note that for every of the $\mathcal{O}(n)$ nodes of the tree decomposition $T$, we make a number of computation that is bounded by the number of $\kappa$'s for the node itself and it's up to two children.
  Any of those computations is polynomial in the size of the arc sets of the bag.
  Now, the number of $\kappa$'s is bounded by the number of possible functions $\pi$ and the time needed to enumerate $\mathcal{C}|_{B_x}$ for any $x \in V(T)$.
  The number of functions $\pi$ is bounded by $(b - a + 1)^{\mathcal{O}(\tw(G))} = 2^{\mathcal{O}(\tw(G) \log (b - a))}$, showing the overall run-time.
\end{proof}

We can deduce three of our main results from this general algorithm.
\begin{theorem}
\label{thm:algorithm_for_treewidth_normalized_weights_and_bounded_number_of_negative_arcs}
  Let $(G, w, k)$ be a {\sc MinFB} instance with weights $w: A(G) \to \lbrace -1, 0, 1\rbrace$.
  Then in time $2^{\mathcal{O}(\tw(G) \log w_-)}(n + m)$, we can decide the existence of a solution to $(G, w, k)$.
\end{theorem}
\begin{proof}
  By \Cref{thm:NDFASisequivalenttospecialMIS}, we can construct an instance of \NDFAS{} in polynomial time that has the same parameters $\tw + w_-$ such that the original instances is a ``yes''-instance if and only if the constructed instance is. 

  Let $\mathcal{C}$ consist of the single partition $(V(G))$ of $V(G)$.
  Compute a nice tree decomposition $(T,\mathcal B)$ of $G$.
  Note that $\mathcal{C}$ is $(T,\mathcal B)$-compatible.
  We know that any $(C, \pi)$-feasible set $S$ for $C \in \mathcal{C}$ is a negative directed feedback arc set with $\pi$ being a feasible potential of $G - S$.
  By \Cref{lem:few_normalized_negative_arc_weights_have_bounded_potential}, we know that for any negative directed feedback arc set $S$ of~$(G, k)$, the graph $G - S$ has a feasible potential $\pi: V(G) \to \mathbb{Z} \cap [0, w_-]$.
  Thus, we call \Cref{lem:dynamic_program_for_treewidth_normalized_weights_and_bounded_number_of_positive_or_negative_arcs} with $(T,\mathcal B)$, $\mathcal{C}$, $a = 0$ and $b = w_-$ and get a negative directed feedback arc set $S'$ of minimum size.
  We then check whether $|S'| \leq k$ or not, and return the corresponding answer.
	
  For the run-time, we use \Cref{lem:dynamic_program_for_treewidth_normalized_weights_and_bounded_number_of_positive_or_negative_arcs} and \Cref{thm:compute_nice_tree_decomposition}.
  Regarding \Cref{lem:dynamic_program_for_treewidth_normalized_weights_and_bounded_number_of_positive_or_negative_arcs}, observe that the sets of the form $C|_{B_x}$ for any $x \in V(T)$ are exactly the sets $B_x$.
\end{proof}

\begin{theorem}
\label{thm:algorithm_for_treewidth_normalized_weights_and_bounded_number_of_positive_arcs}
  Given any {\sc MinFB} instance $(G, w, k)$ with weights $w: A(G) \to \lbrace -1, 0, 1\rbrace$, in time\linebreak $2^{\mathcal{O}(\tw(G) (\log \tw(G) + \log w_+))}(n + m)$ we can decide the existence of a solution to $(G, w, k)$.
\end{theorem}
\begin{proof}
  By \Cref{thm:NDFASisequivalenttospecialMIS}, we can construct an instance of \NDFAS{} in polynomial time that has the same parameters $\tw + w_+$ such that the original instances is a ``yes''-instance if and only if the constructed instance is.

  Let $\mathcal{C}$ consist of all ordered partitions of $V(G)$.
  Compute a nice tree decomposition $(T,\mathcal B)$ of~$G$.
  By properties of tree decompositions, we have that for any $x \in V(T)$ the sets $V(G_x)$ and $V(G) \setminus (V(G_x) \setminus B_x)$ only intersect in $B_x$.
  Thus, $\mathcal{C}$ is $(T,\mathcal B)$-compatible as we can recombine any two ordered partitions that match on a subset $B_x$.
	
  We know that any $(C, \pi)$-feasible set $S$ for $C = (C_1, \hdots, C_t) \in \mathcal{C}$ is a negative directed feedback arc set with $\pi$ being a feasible potential for every $G[C_i] - S$ and $(C_1, \hdots, C_t)$ is a \mbox{(super-)partition} of the strong components of $G - S$.
  By \Cref{lem:few_normalized_positive_arc_weights_in_strong_graphs_have_bounded_potential}, we know that for any negative directed feedback arc set  $S$ of $(G, k)$, the strong components of $G - S$ have a feasible potential $\pi: V(G) \to \mathbb{Z} \cap [0, w_+]$.
  Thus, we call \Cref{lem:dynamic_program_for_treewidth_normalized_weights_and_bounded_number_of_positive_or_negative_arcs} with $T$, $\mathcal{C}$, $a = 0$ and $b = w_+$ and get a negative directed feedback arc set $S'$ of minimum size.
  We then check whether $|S'| \leq k$ or not and return the corresponding answer.
	
  For the run-time, we use \Cref{lem:dynamic_program_for_treewidth_normalized_weights_and_bounded_number_of_positive_or_negative_arcs} and \Cref{thm:compute_nice_tree_decomposition}. In \Cref{lem:dynamic_program_for_treewidth_normalized_weights_and_bounded_number_of_positive_or_negative_arcs}, the sets of the form $C|_{B_x}$ for any $x \in V(T)$ can be enumerated in time $2^{\mathcal{O}(\tw(G) \log \tw(G))}$ by taking every ordered partition of~$B_x$.
\end{proof}

\begin{theorem}
\label{thm:algorithm_for_treedepth_and_normalized_weights}
  Let $(G, w, k)$ be a {\sc MinFB}  instance with weights $w: A(G) \to \lbrace -1, 0, 1\rbrace$.
  Then, in time $2^{\mathcal{O}((\td(G))^2)}(n + m)$, we can decide the existence of a solution to $(G, w, k)$.
\end{theorem}
\begin{proof}
  By \Cref{thm:NDFASisequivalenttospecialMIS}, we can construct, in polynomial time, an instance of \NDFAS{} with the same parameter $\td$, such that the original instances is a ``yes''-instance if and only if the constructed instance is.

  Let $\mathcal{C}$ consist of the single partition $(V(G))$ of $V(G)$.
  Compute a nice tree decomposition $(T,\mathcal B)$ of~$G$.
  Note that $\mathcal{C}$ is $(T,\mathcal B)$-compatible.
  We know that any $(C, \pi)$-feasible set $S$ for $C \in \mathcal{C}$ is a negative directed feedback arc set of $(G,w)$ with $\pi$ being a feasible potential of $G - S$.
  By \Cref{lem:bounded_treedepth_graphs_have_bounded_potential}, we know that for any negative directed feedback arc set $S$ of $(G,w)$, $G - S$ has a feasible potential $\pi: V(G) \to \mathbb{Z} \cap [0, 2^{\td(G)}]$.
  Thus, we call \Cref{lem:dynamic_program_for_treewidth_normalized_weights_and_bounded_number_of_positive_or_negative_arcs} with $T$, $\mathcal{C}$, $a = 0$ and $b = 2^{\td(G)}$ and get a negative directed feedback arc set $S'$ of minimum size.
  We then check whether $|S'| \leq k$ or not and return the corresponding answer.
	
  For the run-time, we use \Cref{lem:dynamic_program_for_treewidth_normalized_weights_and_bounded_number_of_positive_or_negative_arcs} and \Cref{thm:compute_approximate_nice_tree_decomposition}. In \Cref{lem:dynamic_program_for_treewidth_normalized_weights_and_bounded_number_of_positive_or_negative_arcs}, the sets of the form $C|_{B_x}$ for any $x \in V(T)$ can be enumerated in constant time as they are exactly the set $(B_x)$.
  Last but not least, we have \mbox{$\tw(G) \leq \td(G)$}, and thus an overall run-time of $2^{\mathcal{O}(\td(G) \log (2^{\td(G)}))}(n + m) = 2^{\mathcal{O}((\td(G))^2)}(n + m)$.
\end{proof}

\subsection{Algorithm for \texorpdfstring{$\lbrace -1, 1 \rbrace$}{\{-1,1\}} entries with Few Negative Arcs}
In this section we give an algorithm when parameterized by $w_-$.
However, this algorithm works only for weights of the form $w: A(G) \to \lbrace -1, 1\rbrace$.
As stated in the results overview (see \Cref{sec:results}), the parameterized complexity of {\sc MinFB} with respect to parameter is open for right-hand side $b \in \lbrace -1, 0, 1\rbrace$, and $\mathsf{W}[1]$-hard for general integral weights.
In the case of $\lbrace -1, +1 \rbrace$ values, finding a negative directed feedback arc set however becomes easy by the following observation.

\begin{lemma}
\label{lem:negative_cycle_length_bounded_by_w_minus_for_special_arc_weights}
  Let $G$ be a directed graph with arc weights $w: A(G) \to \lbrace -1, 1\rbrace$.
  Then any negative cycle of $G$ has length at most~$2w_-$.
\end{lemma}
\begin{proof}
  Any cycle $C$ of length more than $2w_-$ contains at least $(w_- + 1)$-many arcs of weight $+1$ and at most $w_-$-many arcs of weight $-1$, therefore its total weight is non-negative.
\end{proof}

Combining \Cref{lem:negative_cycle_length_bounded_by_w_minus_for_special_arc_weights} and \Cref{lem:algorithm_for_bounded_negative_cycle_length_plus_k} leads to the following theorem.
\begin{theorem}
\label{thm:algorithm_for_w_minus_and_special_arc_weights}
  There is an algorithm solving \textsc{MinFB} with right-hand sides in $\lbrace -1, 1\rbrace$ in time $\mathcal{O}( (2w_-)^{k} n^2m)$.
\end{theorem}
\begin{proof}
  By \Cref{thm:NDFASisequivalenttospecialMIS}, we can construct an instance of \NDFAS{} in polynomial time that has the same parameters $k+w_-$ such that the original instances is a ``yes''-instance if and only if the constructed instance is.
  The theorem follows by combining \Cref{lem:algorithm_for_bounded_negative_cycle_length_plus_k} and \Cref{lem:negative_cycle_length_bounded_by_w_minus_for_special_arc_weights}.
\end{proof}

Note that the run-time can be improved to $\mathcal{O}((k+1)w_-^k n^2m)$ by guessing the number of $-1$ arcs in our solution first and enumerating all sets of that size.
Then the algorithm of \Cref{lem:algorithm_for_bounded_negative_cycle_length_plus_k} can be modified to only make subroutine calls on arcs of weight $+1$, leading to an improved run-time.
Moreover, note that we can assume $k \leq w_-$ by \Cref{thm:solve_w_minus_bounded_by_k}, which shows that \NDFAS{} with arc weights $w: A(G) \to \lbrace -1, 1\rbrace$ is fixed-parameter tractable in~$w_-$.

\subsection{Algorithm for \texorpdfstring{$\lbrace -1, 1 \rbrace$}{\{-1,1\}} Weights with Few Positive Arcs}
\label{sec:fewpositiverighthandsides}
We will now study weight functions $w: A(G) \rightarrow \{-1, +1\}$ when parameterized by $k + w_+$.
As stated in the results overview in \Cref{sec:results} this parameter is open for weights of the form $w: A(G) \to \lbrace -1, 0, 1\rbrace$ and $\mathsf{W}[1]$-hard for general integral weights.

The main observation for our algorithm is made in the following lemma:
\begin{lemma}
  \label{thm:shortcycleorrepresentativeedge}
  Let $G$ be a directed graph with arc weights $w:A(G) \rightarrow \{\pm 1\}$.
  Then either~$G$ has a negative cycle of length at most $2(w_+)^2 + 2w_+$, or every negative cycle $C$ has some arc $a \in A(C)$ that lies only on negative cycles of $G$.
\end{lemma}
\begin{proof}
  First, if $G$ has no negative cycles then we are done.
  We are also done if there is a negative cycle of length at most $2(w_+)^2 + 2w_+$.
  So, assume that $G$ has a negative cycle and all negative cycles have length at least $2(w_+)^2 + 2w_+ + 1$.
    
  Suppose, for sake of contradiction, that each arc $a$ in a negative cycle lies on a cycle $C_a$ of non-negative weight.
  Let~$C$ be a shortest negative cycle in $G$.
  By the argumentation above,~$C$ has length at least $2(w_+)^2 + 2w_+ + 1$.
  In particular, $C$ has length at least $w_+(w_+ + 1) + 1$ and contains at most $w_+$ arcs of weight~$+1$; all other arcs of $C$ have weight $-1$.
  By pigeonhole principle, there must be a segment $P$ of~$C$ consisting of $w_+ + 1$ arcs of weight $-1$.
    
  By assumption, for every arc $a = (v,w) \in A(P)$ we have a $w \to v$-path $R_a$ of length at most $|C_a| - 1$.
  As $w(C_a) \geq 0$ and the $C_a$'s contain at most $w_+$ arcs of weight $+1$ and all other arcs have weight $-1$ the cycles have length at most~$2w_+$.
  Thus, the reverse paths $R_a$ have length $|R_a| \leq |C_a| - 1 \leq 2w_+ - 1$.
  Say $P$ is an $s \to t$-path, then by concatenating the $R_a$'s with $a \in P$ we get an $t \to s$-walk $R'$ which contains a $t \to s$-path $R$.
  The path $R$ contains at most $w_+$ arcs of weight~$+1$, thus $w(R) \leq w_+$.
    
  Consider the closed walk $O = P \circ R$.
  Then $w(O) = w(P) + w(R) = -(w_++1) + w(R) \leq -1$.
  Thus,~$O$ contains a negative cycle $C'$.
  Eventually,
  \begin{align*}
    |\ell(C')| & \leq |O|\\
               & = |P| + |R|\\
               & \leq |P| + \sum_{e \in E(P)} |R_e|\\
               & \leq 2w_+(w_++1)\\
               & < 2(w_+)^2 + 2w_+ + 1 \leq |\ell(C)|
  \end{align*}
  yields a contradiction to the fact that $C$ was a shortest negative cycle.
\end{proof}

This lemma forms the basis of our algorithm, \Cref{Alg:PositiveParameter}.
First, the algorithm checks for negative cycles with up to $2(w_+)^2+2w_+$ arcs.
It then guesses the arc contained in a solution like the algorithm alluded to in \Cref{lem:algorithm_for_bounded_negative_cycle_length_plus_k}.
Afterwards, we are left with a digraph without short negative cycles.
We now identify the set $U$ of arcs which are not part of a non-negative cycle.
Then, we know that $G - S$ may not contain a cycle on which an arc of $U$ lies, as this cycle would be negative by definition of $U$.
Likewise, any negative cycle in $G$ has some arc in $U$ by the previous lemma.
For general sets $U$, this is exactly the \textsc{Directed Subset Feedback Arc Set} problem.

\problemdef{Directed Subset Feedback Arc Set}
  {A graph $G$, an arc set $U \subseteq A(G)$ and an integer $k \in \mathbb{Z}_{\geq 0}$.}
  {Find an arc set $X \subseteq A(G)$ of size at most~$k$ such that every cycle of $G - X$ is disjoint from $U$ or decide that no such set exists.}

The \textsc{Directed Subset Feedback Arc Set} problem was shown to be fixed-parameter tractable for parameter $k$ by Chitnis et al.~\cite{ChitnisEtAl2015}.
\begin{proposition}[\cite{ChitnisEtAl2015}]
\label{thm:subsetdfasfpt}
  {\sc Directed Subset Feedback Arc Set} is solvable in time $2^{\mathcal{O}(k^3)}\cdot n^{\mathcal O(1)}$.
\end{proposition}

\begin{algorithm}
  \caption{NegativeCycleDeletion\label{Alg:PositiveParameter}}
  \SetKwFunction{SDFAS}{DirectedSubsetFeedbackArcSet}
  \DontPrintSemicolon
  \SetKwInOut{Input}{Input}\SetKwInOut{Output}{Output}
  \SetKwFunction{False}{false}
	
  \Input{A directed graph $G$ with arc weights $w:A(G) \rightarrow \mathbb{Z}$ and $k \in \mathbb{Z}_{\geq 0}$.}
  \Output{A set $S \subseteq A(G)$ of at most $k$ arcs such that $G - S$ has no negative cycle, or \False if no such set exists.}
	
  \If{$k < 0$}{
    \Return \False.\;
  }
  \eIf{there is some negative cycle $C$ of length at most $2(w_+)^2 + 2w_+$ in $G$}{
    Branch on deleting an arc of $C$ and try to solve with $k-1$.\;
  }
  {
  Identify the set $U$ of all arcs which do not lie on a non-negative cycle.\;
    \Return \SDFAS $(G, U, k)$.\;
  }
\end{algorithm}

Before we prove correctness and run-time, we show how to detect the set~$U$ of all arcs which lie only on negative cycles.
We first argue that this problem is $\mathsf{NP}$-hard even for weights $w:A(G) \rightarrow \{ -1, +1\}$.
To this end, we provide a reduction from the \textsc{Hamiltonian $s$-$t$-Path} problem, which for a directed graph $H$ and vertices $s,t\in V(H)$ asks for an $s \to t$-path in~$H$ visiting each vertex of~$H$ exactly once.
Its $\mathsf{NP}$-hardness was shown by Karp~\cite{Karp1972}.
The reduction works as follows:
Take the original directed graph $H$ and two vertices $s,t \in V(H)$ which we want to test for the existence of a Hamiltonian path starting in $s$ and ending in $t$.
Add a path~$P$ of length $n-1$ from $t$ to $s$ to the graph.
Assign weight $+1$ to each arc of $H$, and weight $-1$ to each arc of $P$.
Then an arc of $P$ lies on a cycle of non-negative length if and only if there is a Hamiltonian $s \to t$-path in $H$.

However, for this construction of weights $w$ we have $w_+ \in \Omega(n)$.
We will now show that the task is indeed fixed-parameter tractable when parameterized by $w_+$.
For that, the main observation is that every non-negative cycle has length at most $2w_+$.
We now consider the {\sc Weighted Longest Path} problem: given a directed graph $G$ with arc weights $w: A(G) \rightarrow \mathbb{R}$ and numbers $W \in \mathbb{R}, \ell \in \mathbb{Z}_{\geq 0}$, the task is to find a path of length exactly $\ell$ and weight at least $W$ in $G$.
Zehavi~\cite{Zehavi2015} gave a fast algorithm for {\sc Weighted Longest Path}, based on color coding-related techniques and representative sets.
\begin{proposition}[\cite{Zehavi2015}]
  {\sc Weighted Longest Path} can be solved in time $2^{\mathcal{O}(\ell)} \cdot \mathcal{O}(m \log n)$.
\end{proposition}

So given an arc $a = (s, t)$, one can enumerate all path sizes $\ell$ from $1$ to $2w_+ - 1$ and ask whether there is a $t \to s$-path of length $\ell$ of weight at least $-w(a)$.
This way one can detect a non-negative cycle containing $a$.

\begin{corollary}
\label{thm:partofposcycledetection}
  Let $G$ be a directed graph, let $w: A(G) \rightarrow \{\pm 1\}$ be a function and $(s,t) \in A(G)$ be an arc.
  Then one can detect in time $2^{\mathcal{O}(w_+)} \cdot m\log n$ if $a = (s,t)$ is part of some non-negative cycle~$C$.
\end{corollary}

Finally, we argue the correctness and run-time of \Cref{Alg:PositiveParameter}, proving the following:
\begin{theorem}
\label{thm:positiveparameter}
  \Cref{Alg:PositiveParameter} is correct and solves an instance of {\sc MinFB} with right-hand sides in $ \lbrace -1, 1\rbrace$  in time $2^{\mathcal{O}(k^3+w_+ + k\log w_+)}\cdot n^{\mathcal O(1)}$.
\end{theorem}
\begin{proof}
  By \Cref{thm:NDFASisequivalenttospecialMIS}, we can construct an instance of \NDFAS{} in polynomial time that has the same parameters $k+w_+$ such that the original instances is a ``yes''-instance if and only if the constructed instance is.

  We detect with help of \Cref{thm:Moore_Bellman_Ford} in time $\mathcal{O}\left( n^2m\right)$ whether there is a negative cycle in~$G$, and if the minimum-length negative cycle $C$ has length at most $2(w_+)^2+2(w_+)$.
  If such a cycle~$C$ exists, we split into $|\ell(C)|$ instances, one for each $a \in A(C)$, calling our algorithm recursively with parameter~$k$ decreased by one on $G - a$.
  This is correct, as one of the arcs of $C$ needs to be deleted, and we just search for every arc if there is a solution containing this arc.
    
  Otherwise, all negative cycles have length at least $2(w_-)^2+2(w_-)+1$, and by \Cref{thm:shortcycleorrepresentativeedge} every negative cycle must have an arc not contained in a non-negative cycle.
  Choose one arc of each such cycle and gather them in the set $U_{<0}$.
  By definition of the set $U$, in our algorithm we have $U_{<0} \subseteq U$.
  {\sc Subset DFAS} for $(G,W,k)$ now asks for a set $S$ of size at most $k$ such that $G - S$ has no cycle containing an arc of $W$.
  As $U_{<0} \subseteq U$, we get that a solution for $(G, U, k)$ is also a solution for $(G, U_{<0},k)$.
  We now show that the reverse direction also holds, thus solving the original problem.
  For this, let $S_{<0}$ be a solution for $(G, U_{<0},k)$ and $C$ be a cycle in $G - S_{<0}$.
  If there is no such cycle, we are done.
  Otherwise, we know that $C$ cannot be a cycle of negative weight as every cycle of negative weight has an arc in~$U_{<0}$.
  But as $w(C) \geq 0$, our cycle~$C$ cannot contain an arc of $U$ as those are not contained in non-negative cycles.
  Thus, $S_{<0}$ is a solution for $(G, U, k)$.
  Also, all cycles in $G - S_{<0}$ are non-negative, and therefore $S_{<0}$ is also a solution to our {\sc Negative DFAS} instance $(G, k)$.
    
  This shows the correctness of \Cref{Alg:PositiveParameter}.
  The run-time can be bounded as follows.
  Detecting cycles in line 4 can be done in time $\mathcal{O}((2(w_+)^2 + 2w_+)nm)$.
  The branching step then creates up to $2(w_+)^2 + 2w_+$ instances with the parameter~$k$ decreased by one.
  As there is no other recursive call to this algorithm, we have at most $(2(w_+)^2+2w_++1)^k$ instances for which this algorithm is called.
  In the end, we call the algorithm alluded to in \Cref{thm:partofposcycledetection} to compute the set~$U$ which takes time $2^{\mathcal{O}(w_+)} \cdot m^2\log n$, as it is called for every arc.
  By \Cref{thm:subsetdfasfpt}, the final call to the \textsc{Directed Subset Feedback Arc Set} oracle takes time $2^{\mathcal{O}(k^3)}\cdot n^{\mathcal O(1)}$.
    
  Thus, we obtain an overall run-time of
  \begin{displaymath}
    (2w_+^2+2w_++1)^k \cdot \left[(w_+^2 + 2w_+)\cdot \mathcal{O}(nm) + 2^{\mathcal{O}(w_+)} \cdot \mathcal{O}(m^2\log n) + 2^{\mathcal{O}(k^3)}\cdot n^{\mathcal O(1)}\right],
  \end{displaymath}
  which simplifies to $(w_+)^{\mathcal{O}(k)} \cdot 2^{\mathcal{O}(k^3 + w_+)} \cdot n^{\mathcal O(1)} = 2^{\mathcal{O}(k^3+w_+ + k\log w_+)} \cdot n^{\mathcal O(1)}$.
\end{proof}

\section{Parameterized Intractability Results}
\label{sec:hardness}

\subsection{\texorpdfstring{$\mathsf{NP}$-Hardness}{NP-Hardness} for Number of Positive Arcs}
For completeness of the hardness results, this section contains a short observation about the equivalence of DFAS and \NDFAS{} instances where all arc weights are $-1$.
This implies that \NDFAS{} is $\mathsf{NP}$-hard even in the case where all arc weights are~$-1$.
This implies the hardness result for \textsc{MinFB} with parameter $w_+$ and arc weights $w: A(G) \to \lbrace -1, 0, 1\rbrace$.

\begin{theorem}
\label{thm:NP_hardness_for_positive_arcs}
  {\sc Negative DFAS} is $\mathsf{NP}$-hard, even if all arc weights are $-1$.
\end{theorem}
\begin{proof}
  We show the theorem by a reduction from DFAS, which is $\mathsf{NP}$-hard.
  Let $(G, k)$ be an instance of DFAS.
  We claim that $(G, w, k)$ with $w \equiv -1$ is an equivalent instance of \NDFAS{}.
  Indeed, for every $S \subseteq A(G)$ there are no cycles in $G - S$ if and only if $G - S$ with weights~$-1$ contains no negative cycles.
\end{proof}



\subsection{\texorpdfstring{$\mathsf{NP}$-Hardness}{NP-Hardness} for Constant Pathwidth \label{sec:Patwidth_NPHard}}

In this section we show that {\sc MinFB} is $\mathsf{NP}$-hard even for constraint matrices $A$ whose pathwidth is bounded by 6.
To this end, we reduce \textsc{Partition} to {\sc Negative DFAS} in directed graphs whose underlying undirected graph has pathwidth at most 6.

\problemdef{Partition}
  {A set $\mathcal{A} = \lbrace a_1,\hdots,a_n\rbrace$ of positive integers.}
  {Find a subset~$\mathcal{A}'$ such that	$\sum_{a_i \in \mathcal{A}'} a_i = \sum_{a_i \in \mathcal{A} \setminus \mathcal{A}'}a_i$ or decide that no such subset exists.}

Using $A = \sum_{i = 1}^n a_i$, we can reformulate {\sc Partition} as the problem of finding a subset $\mathcal{A}'$ such that $\mathcal{A}'$ and $\mathcal{A} \setminus \mathcal{A}'$ both sum up to $\frac{A}{2}$ (or no such subset exists).
Karp~\cite{Karp1972} showed that {\sc Partition} is $\mathsf{NP}$-complete.

\begin{proposition}[Karp~\cite{Karp1972}]
  {\sc Partition} is $\mathsf{NP}$-complete.
\end{proposition}

Now we are ready to state our hardness result.
\begin{theorem}
\label{thm:NP_hardness_for_pathwidth}
  {\sc Negative DFAS} is $\mathsf{NP}$-hard, even for graphs of pathwidth $6$ and one arc of positive weight.
\end{theorem}
\begin{proof}
  Let $\mathcal{A}$ be an instance of \textsc{Partition} and $A = \sum_{a_i \in \mathcal{A}} a_i$.
  For every number $a_i \in \mathcal{A}$ we construct a gadget $G_i$ as follows (see \Cref{fig:GadgetPartition} for an illustration).
  Let $V^{(i)} = \{ s_i^{(j)}, t_i^{(j)}, x_ i^{(j)}, y_ i^{(j)} \mid j=1,2 \}$ be the set of vertices in $G_i$.
  There are three different types of arcs.
  The first arc set $A_1^{i}=\{ (x_i^{(j)},y_i^{(j)}) \mid j=1,2 \}$ with arc weight $-a_i$, contains the arcs we will consider for deletion later.
  The second arc set $A_2^{i}=\{ (y_i^{(j)},x_i^{(j+1)}) , (y_i^{(j)},x_i^{(j-1)}) \mid j=1,2 \}$ with arc weight $0$, contains arcs which enforce the deletion of arcs from the first arc set by inducing negative cycles.
  The last arc set $A_3^{i}=\{ (s_i^{(j)},t_i^{(j)}), (s_i^{(j)},x_i^{(j)}) , (y_i^{(j)},t_i^{(j)}) \mid j=1,2 \}$ with arc weight $0$, contains arcs that connect the vertices $s_i^{(j)}$ and~$t_i^{(j)}$ to the rest of the gadget.

  \begin{figure}[h]
	\centering
	\quad
		\begin{tikzpicture}[scale=0.8, state/.style={circle,  minimum size=.7cm}]
		\tikzset{>={Latex[width=2mm, length=2mm]}}	
		\node (1) at (2,1) [state, draw, inner sep=-.05mm] {$x_i^{(2)}$};
		\node (2) at (4,1) [state, draw, inner sep=-.05mm] {$y_i^{(2)}$};
		\node (3) at (0,0) [state, draw, inner sep=-.05mm] {$s_i^{(2)}$};
		\node (4) at (6,0) [state, draw, inner sep=-.05mm] {$t_i^{(2)}$};
		\node (5) at (2,3) [state, draw, inner sep=-.05mm] {$x_i^{(1)}$};
		\node (6) at (4,3) [state, draw, inner sep=-.05mm] {$y_i^{(1)}$};
		\node (7) at (0,4) [state, draw, inner sep=-.05mm] {$s_i^{(1)}$};
		\node (8) at (6,4) [state, draw, inner sep=-.05mm] {$t_i^{(1)}$};
		\foreach \a / \b / \c / \pos  in  
		{7/8/$0$  /above,
			7/5/$0$/below,
			3/4/$0$  /below,
			3/1/$0$/above,
			5/6/$-a_i$/above,	
			1/2/$-a_i$/below,	
			6/8/$0$/below,	
			2/4/$0$/above%
		}
		{ \draw (\a) edge[->] node[\pos]{ \scriptsize \c}  (\b);
		}; 
		\draw (2) edge[->] node[above,pos=0.1]{ \scriptsize $0$ }  (5);
		\draw (6) edge[->] node[below,pos=0.1]{ \scriptsize $0$ }  (1);
		;
		\end{tikzpicture}
	\caption{The gadget graph $G_i$. \label{fig:GadgetPartition}}
  \end{figure}
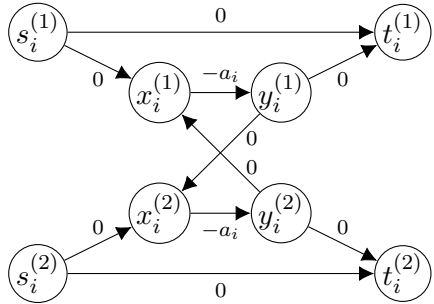

  The whole gadget $G_i$ is then defined as $(V^{(i)} , A_1^{(i)} \cup A_2^{(i)} \cup A_3^{(i)} )$.
  Out of these gadgets we construct the graph $(G, w)$ of our \NDFAS{} instance by taking the union of all gadgets~$G_i$ for $i = 1,\hdots,n$, where we identify $t_i^{(j)}=s_{i+1}^{(j)}$ for $j=1,2$ and $i \in \lbrace 1, \hdots, n-1\rbrace$.
  Additionally, we add two vertices $s$ and~$t$ with the arcs $(s,s_1^{(j)})$ and $(t_{n}^{(j)},t)$ for $j=1,2$ of weight $0$.
  Finally, we add the arc $(t,s)$ of weight $\frac{A}{2}$.
  We then choose $(G, w, k)$ with $k = n$ as our \NDFAS{} instance.
	
  Next, we show that $\mathcal{A}$ has a solution if and only if $(G, w, n)$ has one.
  \begin{claim}
    If $\mathcal{A}$ has a solution $\mathcal{A}'$, then there is a solution $S^\star$ to $(G, w, n)$.
  \end{claim} 
  \begin{claimproof}
    For every $i \in \lbrace 1, \hdots, n\rbrace$, we define an $s_i \in S^\star$ by
    \begin{equation*}
      s_i = \begin{cases}
              (x_i^{(1)}, y_i^{(1)}),	&\text{if $a_i \in \mathcal{A}'$,}\\
			  (x_i^{(2)}, y_i^{(2)}),	&\text{if $a_i \notin \mathcal{A}'$ \enspace .}
			\end{cases}
    \end{equation*}

    Suppose, for sake of contradiction, that $G - S^\star$ contains a negative cycle $C$.
    Observe that for each $i \in \lbrace 1, \hdots, n\rbrace$ the graph $G_i - s_i$ contains no negative cycle.
    Moreover, any gadget~$G_i$ can only be entered through one $s_i^{(j)}$ and left through $t_i^{(j)}=s_{i+1}^{(j)}$ with the same $j$.
    Thus, $C$ has to start in $s$, pass through the gadgets $G_i$ in order either on the $j=1$ or the $j=2$ side, then go to~$t$ and eventually use the backward arc $(t,s)$.
    As the backward arc has weight $\frac{A}{2}$, $C$ contains an $s \to t$-path $P$ of weight less than $\frac{A}{2}$.
    Furthermore, $P$ uses either only vertices with index $j=1$ or $j=2$.
    The only arcs of negative weight with index $j = 1$ in $G - S^\star$ are those $(x_i^{(1)}, y_i^{(1)})$ with $a_i \notin \mathcal{A}'$.
    For these we have $\sum_{a_i \notin \mathcal{A}'} w((x_i^{(1)}, y_i^{(1)})) = \sum_{a_i \notin \mathcal{A}'} -a_i = -\frac{A}{2}$, as $\mathcal{A}'$ is a solution to $\mathcal{A}$.
    Likewise, the only arcs of negative weight with index $j = 2$ in $G - S^\star$ are those $(x_i^{(2)}, y_i^{(2)})$ with $a_i \in \mathcal{A}'$.
    For these we have $\sum_{a_i \in \mathcal{A}'} w((x_i^{(2)}, y_i^{(2)})) = -\sum_{a_i \in \mathcal{A}'} a_i = -\frac{A}{2}$, as $\mathcal{A}'$ is a solution to $\mathcal{A}$.
    Thus, in any case, $P$ has weight at least $-\frac{A}{2}$, a contradiction.
    \begin{figure}
	  \vspace*{1mm}
      \resizebox{\textwidth}{!}{%
      \begin{tikzpicture}[ state/.style={circle,  minimum size=.7cm}]
			\tikzset{>={Latex[width=2mm, length=2mm]}}	
			\foreach \k in {1,...,4}{
				\node (1*\k) at (2+6*\k,1) [state, draw, inner sep=-.05mm] {$x_{\k}^{(2)}$};
				\node (2*\k) at (4+6*\k,1) [state, draw, inner sep=-.05mm] {$y_{\k}^{(2)}$};
				\node (3*\k) at (0+6*\k,0) [state, draw, inner sep=-.05mm] {$s_{\k}^{(2)}$};
				\node (5*\k) at (2+6*\k,3) [state, draw, inner sep=-.05mm] {$x_{\k}^{(1)}$};
				\node (6*\k) at (4+6*\k,3) [state, draw, inner sep=-.05mm] {$y_{\k}^{(1)}$};
				\node (7*\k) at (0+6*\k,4) [state, draw, inner sep=-.05mm] {$s_{\k}^{(1)}$};
				\draw (2*\k) edge[->,line width=0.5mm] (5*\k);
				\draw (6*\k) edge[->,line width=0.5mm] (1*\k);}%
			\node (3*5) at (0+6*5,0) [state, draw, inner sep=-.05mm] {$t_{4}^{(2)}$};
			\node (7*5) at (0+6*5,4) [state, draw, inner sep=-.05mm] {$t_{4}^{(1)}$};
			\foreach \k in {1,2}{
				\draw (5*\k) edge[->,line width=0.5mm, ACMBlue] node[above]{\scriptsize{$-a_\k$}}  (6*\k);
				\draw (1*\k) edge[->,densely dotted,line width=0.5mm, gray!50]  (2*\k);
				\draw (7*\k) edge[->,line width=0.5mm, ACMBlue]  (5*\k);
				\draw (3*\k) edge[->,line width=0.5mm]  (1*\k);}
			\foreach \k in {3,4}{
				\draw (5*\k) edge[->,densely dotted,line width=0.5mm, gray!50] (6*\k);
				\draw (1*\k) edge[->,line width=0.5mm, ACMOrange] node[above]{\scriptsize{$-a_\k$}}  (2*\k);
				\draw (7*\k) edge[->,line width=0.5mm]  (5*\k);
				\draw (3*\k) edge[->,line width=0.5mm, ACMOrange] (1*\k);	} 
			\foreach \h \k in{6*1/7*2, 6*2/7*3, 7*4/7*5, 7*3/7*4}
			{ \draw (\h) edge[->,line width=0.5mm,ACMBlue]   (\k);}
			\foreach \h \k in{3*1/3*2,3*2/3*3, 2*4/3*5,2*3/3*4}
			{ \draw (\h) edge[->,line width=0.5mm,ACMOrange]   (\k);}
			\foreach \h \k in{7*1/7*2,2*1/3*2,3*4/3*5, 6*4/7*5, 6*3/7*4,7*2/7*3,3*3/3*4,2*2/3*3}
			{ \draw (\h) edge[->,line width=0.5mm] (\k);}
			;
			\end{tikzpicture} }
		\caption{Union of four gadgets after the deletion of negative cycles.
				Deleted arcs are shown with a dotted gray line.
				The blue path shows the shortest path from \mbox{$s_1^{(1)}$ to $t_4^{(1)}$}, and the orange path shows the shortest path from $s_1^{(2)}$ to~$t_4^{(2)}$. \label{fig:GadgetpartitionGraph}}
	\end{figure}
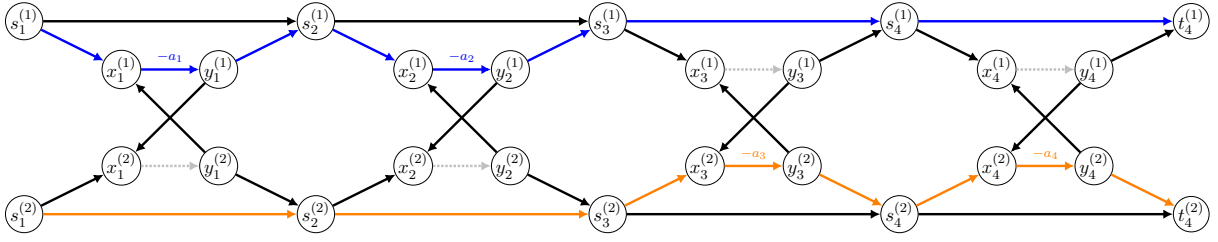
  \end{claimproof}

  \begin{claim}
    If $(G, w, n)$ has a solution $S^\star$ then $\mathcal{A}$ has a solution $\mathcal{A}'$.
  \end{claim}
  \begin{claimproof}
    We choose $\mathcal{A}' = \lbrace a_i \mid (x_i^{(1)}, y_i^{(1)}) \in S^\star\rbrace$.
    Suppose, for sake of contradiction, that $\mathcal{A}'$ is not a solution.
    Then we have that either $\sum_{a_i \in \mathcal{A}'} a_i > \frac{A}{2}$ or $\sum_{a_i \notin \mathcal{A}'} a_i > \frac{A}{2}$.
    We consider now a minimum-weight $s_1^{(j)} \to t_n^{(j)}$-path $P_j$ in $G - S^\star$ for $j = 1, 2$ as depicted in \Cref{fig:GadgetpartitionGraph}.
    Note that for every gadget $G_i$, our solution $S^\star$ has to contain one of the arcs $(x_i^{(1)}, y_i^{(1)})$, $(x_i^{(2)}, y_i^{(2)})$, $(y_i^{(1)}, x_i^{(2)})$ and $(y_i^{(2)}, x_i^{(1)})$.
    As $S^\star$ contains only $n$ elements (one for each gadget), the arcs $(s_i^{(j)}, t_i^{(j)})$ are always available and have weight $0$.
    Moreover, if $(x_i^{(j)}, y_i^{(j)})$ is undeleted, the $s_i^{(j)} \to t_i^{(j)}$-path $(s_i^{(j)}, x_i^{(j)}) \circ (x_i^{(j)}, y_i^{(j)}) \circ (y_i^{(j)}, t_i^{(j)})$ of weight $-a_i$ exists.
    Thus, a minimum-weight $s_1^{(j)} \to t_n^{(j)}$-path~$P_j$ has weight at most $-\sum_{a_i \in \mathcal{A}_j} a_i$ where $\mathcal{A}_j$ is the set of $a_i$'s for which $(x_i^{(j)}, y_i^{(j)})$ is undeleted.
    By choice of $\mathcal{A}'$, we have that $\mathcal{A}_1 = \mathcal{A} \setminus \mathcal{A}'$ and $\mathcal{A}_2 \supseteq \mathcal{A}'$.
    So, if $\sum_{a_i \in \mathcal{A}'} a_i > \frac{A}{2}$, we have that $w(P_2) \leq -\sum_{a_i \in \mathcal{A}_2} a_i \leq -\sum_{a_i \in \mathcal{A}'} a_i < -\frac{A}{2}$.
    If instead $\sum_{a_i \notin \mathcal{A}'} a_i > \frac{A}{2}$, we have that $w(P_1) \leq -\sum_{a_i \in \mathcal{A}_1} a_i \leq -\sum_{a_i \notin \mathcal{A}'} a_i < -\frac{A}{2}$.
    Thus, in any case we have an \mbox{$s_1^{(j)} \to t_n^{(j)}$-path $P_j$} in $G - S^\star$ of weight less than $- \frac{A}{2}$ for some $j \in \lbrace 1, 2 \rbrace$.
	
    Again, our solution $S^\star$ contains exactly one arc of every gadget and thus no arc incident to~$s$ or $t$.
    So, we can complete this path $P_j$ to a cycle $(s, s_1^{(j)}) \circ P_j \circ (t_n^{(j)}, t) \circ (t, s)$ in $G - S^\star$ of weight $w(P_j) + \frac{A}{2} < 0$, a contradiction to $G - S^\star$ containing no negative cycles.
  \end{claimproof}

  This completes the reduction from \textsc{Partition} to \textsc{Negative DFAS} and therefore shows the $\mathsf{NP}$-hardness of {\sc Negative DFAS}.
  It only remains to bound the pathwidth of the generated instances. 

  We show that the underlying graph of $G$ has pathwidth at most 6, by providing a path decomposition $(P, \mathcal{B})$ of $G$ of width 6. Let $P$ be the path on $2n + 1$ vertices, corresponding to bags $B_1 , . . . , B_{2n+1} \in \mathcal{B}$: 
  \begin{itemize}
    \item $B_{2i-1} = \{ s, s_i^{(1)}, s_i^{(2)}, x_i^{(1)}, x_i^{(2)} , y_i^{(1)}, y_i^{(2)} \} $ for $i= 1, \hdots, n$ 
    \item $B_{2i} = \{ s, s_i^{(1)}, s_i^{(2)} , y_i^{(1)}, y_i^{(2)}, t_i^{(1)}, t_i^{(2)} \} $ for $i= 1, \hdots, n$ 
    \item $B_{2n+1} = \{ s, t_n^{(1)}, t_n^{(2)}, t \} $ 
  \end{itemize}
  Also, the arc $(t, s)$ is the only arc of positive weight.
\end{proof}

\subsection{\texorpdfstring{\Whard{1}ness}{W[1]-hardness} for Treedepth and Few Positive Arcs}
In this section we prove \Whard{1}ness for \textsc{MinFB} when parameterized by treedepth and number of positive arcs, by reducing \NDFAS{} to \textsc{Clique}.
\begin{theorem}
  \label{thm:W1_hardness_for_treedepth_and_number_positive_arcs}
  {\sc Negative DFAS} is \Whard{1} when parameterized by the treedepth and number of positive arcs.
\end{theorem}
\begin{proof}
  We prove the theorem by reduction from \textsc{Clique}.
  The input of \textsc{Clique} consists of an undirected graph $G$ and an integer $k$.
  The question is, whether there is a vertex set $X \subseteq V(G)$ of at least $k$ vertices such that $G[X]$ forms a complete graph.
  \textsc{Clique} is $\mathsf{W}[1]$-hard parameterized by $k$, see, e.g., the book by Cygan et al.~\cite[Section 13]{CyganEtAl2015}.
	
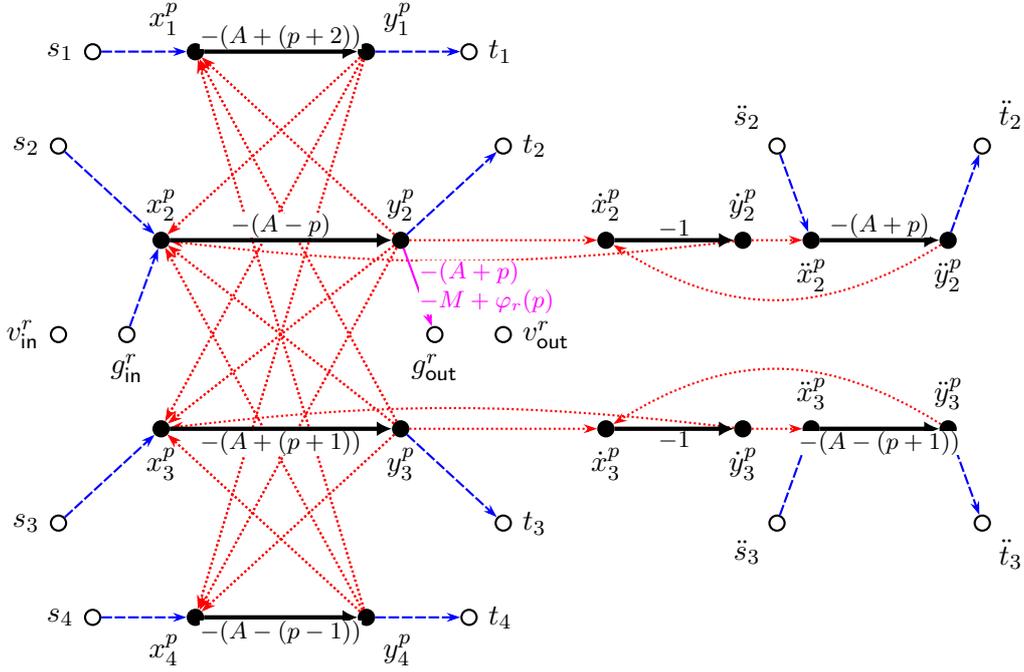
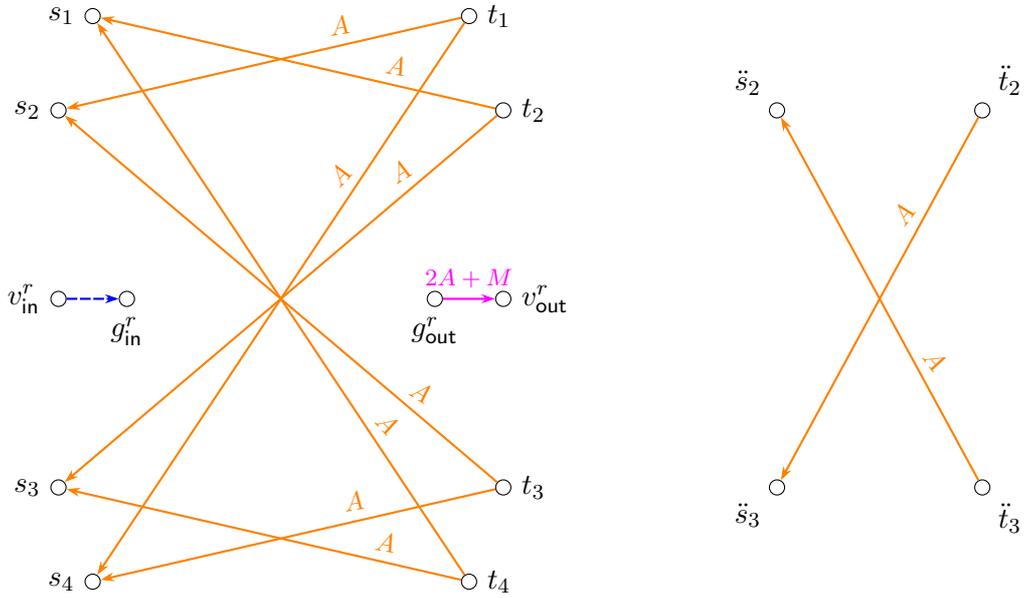
\begin{figure}[p]
		\begin{subfigure}[c]{\textwidth}
		\centering
		\begin{tikzpicture}[xscale=0.9, yscale=1.25, hollow/.style={circle, minimum size=2mm}, filled/.style={fill, hollow}, thick]
			\tikzset{>={Stealth[width=1.2mm, length=1.8mm]}}
			\tikzset{whitelabel/.style={fill=white, on layer=foreground, inner sep=0}}
			\pgfdeclarelayer{foreground}
			\foreach \vertexname/\vertexlabel/\labelpos/\vertexstyle/\x/\y in {%
				s1/$s_1$/180/hollow/0.5/6,
				s2/$s_2$/180/hollow/0/5,
				s3/$s_3$/180/hollow/0/1,
				s4/$s_4$/180/hollow/0.5/0,
				t1/$t_1$/0/hollow/6.0/6,
				t2/$t_2$/0/hollow/6.5/5,
				t3/$t_3$/0/hollow/6.5/1,
				t4/$t_4$/0/hollow/6.0/0,
				dds2/$\ddot{s}_2$/135/hollow/10.5/5,
				dds3/$\ddot{s}_3$/225/hollow/10.5/1,
				ddt2/$\ddot{t}_2$/45/hollow/13.5/5,
				ddt3/$\ddot{t}_3$/315/hollow/13.5/1,
				gin/$g^r_\textsf{in}$/270/hollow/1.0/3,
				gout/$g^r_\textsf{out}$/270/hollow/5.5/3,
				vin/$v^r_\textsf{in}$/180/hollow/0/3,
				vout/$v^r_\textsf{out}$/0/hollow/6.5/3,
				x1/$x^p_1$/135/filled/2/6,
				x2/$x^p_2$/90/filled/1.5/4,
				x3/$x^p_3$/270/filled/1.5/2,
				x4/$x^p_4$/225/filled/2/0,
				y1/$y^p_1$/45/filled/4.5/6,
				y2/$y^p_2$/90/filled/5.0/4,
				y3/$y^p_3$/270/filled/5.0/2,
				y4/$y^p_4$/315/filled/4.5/0,
				dx2/$\dot{x}^p_2$/90/filled/8/4,
				dx3/$\dot{x}^p_3$/270/filled/8/2,
				dy2/$\dot{y}^p_2$/90/filled/10/4,
				dy3/$\dot{y}^p_3$/270/filled/10/2,
				ddx2/$\ddot{x}^p_2$/270/filled/11/4,
				ddx3/$\ddot{x}^p_3$/90/filled/11/2,
				ddy2/$\ddot{y}^p_2$/270/filled/13/4,
				ddy3/$\ddot{y}^p_3$/90/filled/13/2%
			}{
				\node (\vertexname) [\vertexstyle, label={\labelpos:{\vertexlabel}}, draw, inner sep=-.05mm] at (\x,\y) {};
			}

			\begin{scope}[ACMRed, densely dotted]			
			
			\foreach \tail / \head / \label / \orientation /\pos  in {%
				y1/x2/$-1$/above/0.2,
				y1/x3/$-1$/above/0.2,
				y1/x4/$-1$/above/0.2,
				y2/x1/$-1$/above/0.2,
				y2/x3/$-1$/above/0.2,
				y2/x4/$-1$/above/0.2,
				y3/x1/$-1$/above/0.2,
				y3/x2/$-1$/above/0.2,
				y3/x4/$-1$/above/0.2,
				y4/x1/$-1$/above/0.2,
				y4/x2/$-1$/above/0.2,
				y4/x3/$-1$/above/0.2%
				}{
					\draw (\tail) edge[->]  (\head);
				}
				
			\foreach \tail / \head / \out / \in / \label / \orientation /\pos  in {%
				y2/dx2/0/180/$-1$/above/0.5,
				dy2/x2/185/355/$-1$/below/0.5,
				dy2/ddx2/0/180/$-1$/above/0.5,
				ddy2/dx2/205/335/$-1$/below/0.5,
				y3/dx3/0/180/$-1$/above/0.5,
				dy3/x3/175/5/$-1$/above/0.5,
				dy3/ddx3/0/180/$-1$/above/0.5,
				ddy3/dx3/155/25/$-1$/above/0.5%
				}{
					\draw (\tail) edge[->,out=\out,in=\in] (\head);
				}
				
				\foreach \tail / \head / \label / \orientation /\pos  in {%
				y1/x2/$-1$/above/0.2,
				y1/x3/$-1$/above/0.2,
				y1/x4/$-1$/above/0.2,
				y2/x1/$-1$/above/0.2,
				y2/x3/$-1$/above/0.2,
				y2/x4/$-1$/above/0.2,
				y3/x1/$-1$/above/0.2,
				y3/x2/$-1$/above/0.2,
				y3/x4/$-1$/above/0.2,
				y4/x1/$-1$/above/0.2,
				y4/x2/$-1$/above/0.2,
				y4/x3/$-1$/above/0.2%
				}{
					\draw (\tail) edge[->]  (\head);
				}
			\end{scope}
			
			\begin{scope}[thick, ACMBlue, dash pattern= on 4pt off 1 pt]
				\foreach \tail / \head / \label / \orientation /\pos  in {%
				s1/x1/$0$/above/0.2,
				s2/x2/$0$/above/0.2,
				s3/x3/$0$/above/0.2,
				s4/x4/$0$/above/0.2,
				y1/t1/$0$/above/0.2,
				y2/t2/$0$/above/0.2,
				y3/t3/$0$/above/0.2,
				y4/t4/$0$/above/0.2,
				dds2/ddx2/$0$/above/0.2,
				dds3/ddx3/$0$/above/0.2,
				ddy2/ddt2/$0$/above/0.2,
				ddy3/ddt3/$0$/above/0.2,
				gin/x2/$0$/above/0.2%
				}{
					\draw (\tail) edge[->]  (\head);
				}
			\end{scope}
			
			\begin{scope}[ACMPurple]
				\foreach \tail / \head / \label / \orientation /\pos  in {%
				y2/gout/\parbox{1.9cm}{$- (A + p)$\\$- M + \varphi_r(p)$}/right/0.5%
				}{
					\draw (\tail) edge[->, thick] node[\orientation=0.2, pos=\pos, whitelabel]{ \footnotesize \label}  (\head);
				}
			\end{scope}
			
			\foreach \tail / \head / \label / \orientation /\pos  in {%
				x1/y1/$-(A + (p+2))$/above/0.5,
				x2/y2/$-(A - p)$/above/0.5,
				x3/y3/$-(A + (p+1))$/below/0.5,
				x4/y4/$-(A - (p-1))$/below/0.5,
				dx2/dy2/$-1$/above/0.5,
				dx3/dy3/$-1$/below/0.5,
				ddx2/ddy2/$-(A + p)$/above/0.5,
				ddx3/ddy3/$-(A - (p + 1))$/below/0.5%
				}{
					\draw (\tail) edge[-{Latex[width=2mm, length=2mm]}, ultra thick] node[sloped, \orientation =0.25, pos=\pos, whitelabel]{\footnotesize \label}  (\head);
				}
		\end{tikzpicture}
		\subcaption{Meta gadget (sub-level view).}
		\end{subfigure}
		
		\vspace*{3mm}
		
		\begin{subfigure}[c]{\textwidth}
		\centering
		\begin{tikzpicture}[xscale=0.9, yscale=1.25, hollow/.style={circle, minimum size=2mm}, filled/.style={fill, hollow}]
			\tikzset{>={Stealth[width=1.2mm, length=1.8mm]}}
			\foreach \vertexname/\vertexlabel/\labelpos/\vertexstyle\x/\y in {%
				s1/$s_1$/180/hollow/0.5/6,
				s2/$s_2$/180/hollow/0/5,
				s3/$s_3$/180/hollow/0/1,
				s4/$s_4$/180/hollow/0.5/0,
				t1/$t_1$/0/hollow/6.0/6,
				t2/$t_2$/0/hollow/6.5/5,
				t3/$t_3$/0/hollow/6.5/1,
				t4/$t_4$/0/hollow/6.0/0,
				dds2/$\ddot{s}_2$/135/hollow/10.5/5,
				dds3/$\ddot{s}_3$/225/hollow/10.5/1,
				ddt2/$\ddot{t}_2$/45/hollow/13.5/5,
				ddt3/$\ddot{t}_3$/315/hollow/13.5/1,
				gin/$g^r_\textsf{in}$/270/hollow/1.0/3,
				gout/$g^r_\textsf{out}$/270/hollow/5.5/3,
				vin/$v^r_\textsf{in}$/180/hollow/0/3,
				vout/$v^r_\textsf{out}$/0/hollow/6.5/3%
			}{
				\node (\vertexname) [\vertexstyle, label distance=4mm, label=\labelpos:{\vertexlabel}, draw, inner sep=-.05mm] at (\x,\y) {};
			}

			\begin{scope}[thick, ACMBlue, dash pattern= on 4pt off 1 pt]
				\foreach \tail / \head / \label / \orientation /\pos  in {%
				vin/gin/$0$/above/0.2%
				}{
					\draw (\tail) edge[->]  (\head);
				}
			\end{scope}
			\begin{scope}[ACMPurple, thick]
				\foreach \tail / \head / \label / \orientation /\pos  in {%
				gout/vout/$2A + M$/above/0.5%
				}{
					\draw (\tail) edge[->] node[\orientation=0.1, pos=\pos]{ \footnotesize \label}  (\head);
				}
			\end{scope}
			
			\begin{scope}[ACMOrange, thick]
				\foreach \tail / \head / \label / \orientation /\pos  in {%
				t1/s2/$A$/above/0.3,
				t1/s4/$A$/above/0.3,
				t2/s1/$A$/above/0.26,
				t2/s3/$A$/above/0.2,
				t3/s2/$A$/above/0.2,
				t3/s4/$A$/above/0.35,
				t4/s1/$A$/above/0.25,
				t4/s3/$A$/above/0.2,
				ddt2/dds3/$A$/above/0.3,
				ddt3/dds2/$A$/above/0.3%
				}{
					\draw (\tail) edge[->] node[sloped, \orientation, pos=\pos]{ \small \label}  (\head);
				}
			\end{scope}
				
		\end{tikzpicture}
		\subcaption{meta Gadget (top-level view).}
		\end{subfigure}
		
		\vspace*{-4pt}

		\caption{Overview of construction of a meta gadget.
				The dotted red arcs have weight~$-1$, the dashed blue arcs have weight $0$, the solid orange arcs have weight $A$.
				The black arcs model the choice we make for each gadget and the purple arcs are used to extract information from this choice.
				For both of these arc types the weight is individual per arc.\label{fig:GadgetTreedepth}}
	\end{figure}
	
	First we introduce a meta gadget that allows us to model different choices as minimum deletion sets influencing the minimum weight of a path between prescribed pairs of vertices.
	See \Cref{fig:GadgetTreedepth} for an illustration.
	Let $a, b \in \mathbb{Z}_{>0}$ be positive integers.
	For every $r \in \lbrace 1, \hdots, b\rbrace$, let $\varphi_r: \lbrace 1, \hdots, a\rbrace \to \mathbb{Z} \cap [-M , M]$ be some function.
	We define the gadget~$R^{a,b}$ in the following way.
	Let $A = a + 4$.
	For every $p \in \lbrace -1, 0, \hdots, a+1\rbrace$ we have a smaller gadget consisting of the vertices $x^p_i, y^p_i$ for $i \in \lbrace 1, 2,3,4\rbrace$ and $\dot{x}^p_i, \dot{y}^p_i, \ddot{x}^p_i, \ddot{x}^p_i$ for $i \in \lbrace 2,3\rbrace$ that are interconnected by the following arcs:
	\begin{itemize}
      \item arc $(x^p_1, y^p_1)$ of weight $-(A + (p + 2))$,
      \item arc $(x^p_2, y^p_2)$ of weight $-(A - p)$,
      \item arc $(x^p_3, y^p_3)$ of weight $-(A + (p + 1))$,
      \item arc $(x^p_4, y^p_4)$ of weight $-(A - (p - 1))$,
      \item arc $(\dot{x}^p_i, \dot{y}^p_i)$ of weight $-1$ for $i \in \lbrace 2,3\rbrace$,
      \item arc $(\ddot{x}^p_2, \ddot{y}^p_2)$ of weight $-(A + p)$,
      \item arc $(\ddot{x}^p_3, \ddot{y}^p_3)$ of weight $-(A - (p + 1))$,
      \item arcs $(y^p_i, x^p_j)$ for any distinct $i,j \in \lbrace 1,2,3,4\rbrace$ all of weight $-1$,
      \item arcs $(y^p_i, \dot{x}^p_i), (\dot{y}^p_i, x^p_i), (\dot{y}^p_i, \ddot{x}^p_i), (\ddot{y}^p_i, \dot{x}^p_i)$ for $i \in \lbrace 2,3\rbrace$ all of weight $-1$.
	\end{itemize}
	For $p = -1$, we duplicate the arc $(x^{-1}_1, y^{-1}_1)$ and for $p = a + 1$, we duplicate the arc $(x^{a+1}_4, y^{a+1}_4)$.
	Moreover, we duplicate the arcs $(\dot{y}^p_i, \ddot{x}^p_i), (\ddot{y}^p_i, \dot{x}^p_i)$ for $i \in \lbrace 2,3\rbrace$.
	
	Shared by all these mini gadgets, there are vertices $s_i, t_i$ for $i \in \lbrace 1,2,3,4\rbrace$.
	These are connected to the previous vertices by the arcs $(s_i, x^p_i)$ and $(y^p_i, t_i)$ for $i \in \lbrace 1,2,3,4\rbrace$, all of weight~$0$.
	The following arcs run between the shared vertices: $(t_i, s_j)$ for $i,j \in \lbrace 1,2,3,4\rbrace$ with $|i - j|$ odd, all of weight $A$.
	Moreover there are vertices $\ddot{s}_i, \ddot{t}_i$ for $i \in \lbrace 2,3\rbrace$.
	These are adjacent to the arcs $(\ddot{s}_i, \ddot{x}^p_i)$ and $(\ddot{y}^p_i, \ddot{t}_i)$ of weight $0$ for $i \in \lbrace 2,3\rbrace$, as well as to the arcs $(\ddot{t}_2, \ddot{s}_3)$ and $(\ddot{t}_3, \ddot{s}_2)$, both of weight $A$.
	
	This defines the functional part of the gadget.
	To extract the information we want, we add for every $r \in \lbrace 1, \hdots, b\rbrace$ the vertices $g^r_\textsf{in}$ and $g^r_\textsf{out}$ to gather the information of the mini gadgets.
	For every $p \in \lbrace 1, \hdots, a\rbrace$ these vertices are connected to the mini-gadgets by an arc $(g^r_\textsf{in}, x^p_2)$ of weight $0$ and an arc $(y^p_2, g^r_\textsf{out})$ of weight $- (A + p) - M + \varphi_r(p)$.
	
	Finally, for every $r \in \lbrace 1, \hdots, b\rbrace$, the gadget $R^{a,b}$ contains the vertices $v^r_\textsf{in}$ and $v^r_\textsf{out}$ that defines the interface to the outside.
	These are connected to $g^r_\textsf{in}$ and $g^r_\textsf{out}$ by an arc $(v^r_\textsf{in}, g^r_\textsf{in})$ of weight $0$ and an arc $(g^r_\textsf{out}, v^r_\textsf{out})$ of weight $2A + M$.
	
  We prove a series of claims about these meta-gadgets.
  \begin{claim}
  \label{claim:structure_of_mini_gadget_solutions}
    Any solution to $R^{a,b}$ as {\sc Negative DFAS} instance has size at least $5(a + 3)$.
    Moreover, any solution of size at most $5(a + 3)$ deletes
    \begin{itemize}
      \item exactly three of the four arcs $(x^p_i, y^p_i)$ with $i \in \lbrace 1,2,3,4\rbrace$,
      \item exactly one of the two arcs $(\dot{x}^p_2, \dot{y}^p_2)$ and $(\ddot{x}^p_2, \ddot{y}^p_2)$, and
      \item exactly one of the two arcs $(\dot{x}^p_3, \dot{y}^p_3)$ and $(\ddot{x}^p_3, \ddot{y}^p_3)$.
    \end{itemize}
  \end{claim}
  \begin{claimproof}
    We prove that any solution to $R^{a,b}$ as \NDFAS{} instance has to delete at least five arcs from every mini-gadget.
    For any $p \in \lbrace -1, 0, \hdots, a+1\rbrace$ and consider the following cycles.
	For distinct $i,j \in \lbrace 1,2,3,4\rbrace$ there is a cycle
	\begin{equation*}
      (x^p_i, y^p_i) \circ (y^p_i, x^p_j) \circ (x^p_j, y^p_j) \circ (y^p_j, x^p_i),
    \end{equation*}
	which has negative weight as all arcs are negative.
	Moreover, there are the cycles
    \begin{equation*}
      (\ddot{x}^p_i, \ddot{y}^p_i) \circ (\ddot{y}^p_i, \dot{x}^p_i) \circ (\dot{x}^p_i, \dot{y}^p_i) \circ (\dot{y}^p_i, \ddot{x}^p_i)
    \end{equation*}
	for $i \in \lbrace 2,3\rbrace$, both of which are negative as all their arcs are negative.
	Now, any set of four arcs is disjoint from at least one of these cycles.
	So every solution deletes at least five arcs from every mini-gadget.
	As there are $a+3$ arc-disjoint mini-gadgets, every solution has size at least $5(a+3)$.
		
	Note, any solution that deletes exactly five arcs from a mini-gadget must use at least two of these arcs to hit cycles involving the dotted vertices.
	Thus, there are only three arcs to intersect the cycles on non-dotted vertices.
	As the cycles for $(i,j)$ and $(j,i)$ are distinct (but not disjoint), these three arcs have the form $(x^p_i, y^p_i)$ as otherwise one of the non-dotted cycles is not hit.
	To hit the two cycles involving dotted vertices, a solution has to spend exactly one arc on every cycle.
	As the arcs $(\dot{y}^p_i, \ddot{x}^p_i)$ and $(\ddot{y}^p_i, \dot{x}^p_i)$ are doubled, the solution must delete either $(\dot{x}^p_i, \dot{y}^p_i)$ or $(\ddot{x}^p_i, \ddot{y}^p_i)$ for $i \in \lbrace 2,3\rbrace$.
  \end{claimproof}
	
  \begin{claim}
  \label{claim:treedepth_hardness_explicit_solution}
    Consider $R^{a,b}$ as {\sc Negative DFAS} instance.
	Then for every $p^\star \in \lbrace 1, \hdots, a\rbrace$ there is a solution $S_{p^\star}$ of size $5(a+3)$ such any  $v^r_\textsf{in} \to v^r_\textsf{out}$-path in $R^{a,b} - S_{p^\star}$ has weight $\varphi^r(p^\star)$ for any $r \in \lbrace 1, \hdots, b\rbrace$.
  \end{claim}
  \begin{claimproof}
    Let $S_{P^\star}$ consist of the following arcs with non-dotted endpoints:
	\begin{itemize}
	  \item for $-1 \leq p < p^\star - 1$, the arcs $(x^p_2, y^p_2)$, $(x^p_3, y^p_3)$, and $(x^p_4, y^p_4)$,
	  \item for $p^\star - 1$ the arcs $(x^{p^\star - 1}_1, y^{p^\star - 1}_1)$, $(x^{p^\star - 1}_2, y^{p^\star - 1}_2)$ and $(x^{p^\star - 1}_4, y^{p^\star + 1}_4)$,
	  \item for $p^\star$ the arcs $(x^{p^\star}_1, y^{p^\star}_1)$, $(x^{p^\star}_3, y^{p^\star}_3)$ and $(x^{p^\star}_4, y^{p^\star}_4)$,
	  \item for $p^\star < p \leq a + 1$, the arcs $(x^p_1, y^p_1)$, $(x^p_2, y^p_2)$, and $(x^p_3, y^p_3)$.
	\end{itemize}
	Additionally, $S_{p^\star}$ contains the following arcs with dotted endpoints:
	\begin{itemize}
	  \item for $-1 \leq p < p^\star - 1$, the arcs $(\ddot{x}^p_2, \ddot{y}^p_2)$ and $(\ddot{x}^p_3, \ddot{y}^p_3)$,
	  \item for $p^\star - 1$ the arcs $(\ddot{x}^{p^\star - 1}_2, \ddot{y}^{p^\star - 1}_2)$ and $(\dot{x}^{p^\star - 1}_3, \dot{y}^{p^\star - 1}_3)$,
	  \item for $p^\star$ the arcs $(\dot{x}^{p^\star}_2, \dot{y}^{p^\star}_2)$ and $(\ddot{x}^{p^\star}_3, \ddot{y}^{p^\star}_3)$,
	  \item for $p^\star < p \leq a + 1$, the arcs $(\ddot{x}^p_2, \ddot{y}^p_2)$ and $(\ddot{x}^p_3, \ddot{y}^p_3)$.
	\end{itemize}
		
	We have to check that $R^{a,b} - S_{p^\star}$ has no cycle of negative weight.
	Note that the vertices $x^p_i$ where $(x^p_i, y^p_i)$ is deleted have out-degree zero and thus are not part of any cycles.
	The same holds for the vertices $y^p_i$ where $(x^p_i, y^p_i)$ is deleted as they have in-degree zero.
	The argument also holds for $\ddot{x}^p_i$ and $\ddot{y}^p_i$ with $(\ddot{x}^p_i, \ddot{y}^p_i)$ deleted.		
	In addition, $\dot{x}^p_i$ and $\dot{y}^p_i$ are never part of any cycle, as for any $p$ either $(\dot{x}^p_i, \dot{y}^p_i)$ is deleted, and they have either in- or out-degree zero, or $(x^p_i, y^p_i)$ and $(\ddot{x}^p_i, \ddot{y}^p_i)$ are deleted, of which at least one is part of any cycle involving $\dot{x}^p_i$ and $\dot{y}^p_i$.
	Moreover, the vertices~$v^r_\textsf{in},v^r_\textsf{out},g^r_\textsf{in}$ and $g^r_\textsf{out}$ are never part of any cycle in $R^{a,b}$.
	So the only vertices on a cycle are
	\begin{itemize}
	  \item the vertices $s_i, t_i$ (non-dotted, top level),
	  \item the vertices $x^p_i, y^p_i$ where $(x^p_i, y^p_i)$ is not deleted, (non-dotted, mini-gadget),
	  \item the vertices $\ddot{s}_i, \ddot{t}_i$ (dotted, top level),
	  \item the vertices $\ddot{x}^p_i, \ddot{y}^p_i$ where $(\ddot{x}^p_i, \ddot{y}^p_i)$ is not deleted, (dotted, mini-gadget).
	\end{itemize}
	Note that as $\dot{x}^p_i$ and $\dot{y}^p_i$ are never part of any cycle, there is no cycle involving both the non-dotted and the dotted part of our gadget.
	Thus, we can analyze them separately.
		
	First we consider the non-dotted part.
	By leaving out the vertices that do not form a cycle (see above), the mini-gadgets form an $s_i \to t_i$-path for certain $i$'s.
	The only remaining arcs (those that are not part of these paths) are the backwards arcs $(t_i, s_j)$ on the top-level.
	We analyze the $s_i \to t_i$-paths first.
	As we are interested only whether there are negative cycles and not their exact weight, we may restrict ourselves to minimum weight paths of this kind.
	Observe that in $R^{a,b} - S_{p^\star}$,
	\begin{itemize}
	  \item any $s_1 \to t_1$-path has weight at least $-(A + p^\star)$,
	  \item the only $s_2 \to t_2$ path has weight $-(A-p^\star)$,
	  \item the only $s_3 \to t_3$ path has weight $-(A+p^\star)$, and
	  \item any $s_4 \to t_4$-path has weight at least $-(A + p^\star)$.
	\end{itemize}
	The backward arcs $(t_i, s_j)$ have weight $A$ and exist only for $|i-j|$ odd.
	Therefore, on every cycle we have that $s_i \to t_i$-paths of weight at least $-(A-p^\star)$ and of weight at least $-(A+p^\star)$ alternate.
	Also, between them lies an arc of weight~$A$.
	To return to the starting point a walk has to consist of only such pairs (as the backward arcs only exist for $|i-j|$ odd.
	Each such pair has weight at least $-(A-p^\star) + A - (A + p^\star) + A = 0$ and thus every closed walk in the non-dotted part of $R^{a,b} - S_{p^\star}$ has non-negative weight.
			
	For the dotted part of the gadget, note that of the arcs of the form $(\ddot{x}^p_i, \ddot{y}^p_i)$, only $(\ddot{x}^{p^\star}_2, \ddot{y}^{p^\star}_2)$ and $(\ddot{x}^{p^\star-1}_3, \ddot{y}^{p^\star-1}_3)$ are undeleted.
	So the only cycle in the dotted part has vertex sequence $(\ddot{s}_2, \ddot{x}^{p^\star}_2, \ddot{y}^{p^\star}_2, \ddot{t}_2, \ddot{s}_3, \ddot{x}^{p^\star-1}_3, \ddot{y}^{p^\star-1}_3, \ddot{t}_3, \ddot{s}_2)$ and is of weight $-(A + p^\star) + A - (A - (p^\star -1+1)) + A = 0$.
	Thus, there are no negative cycles in $R^{a,b} - S_{p^\star}$ and $S_{p^\star}$ is a solution of size $5(a + 3)$.
		
	It remains to show that any $v^r_\textsf{in} \to v^r_\textsf{out}$-path has weight $\varphi^r(p^\star)$ for any \mbox{$r \in \lbrace 1, \hdots, b\rbrace$}.
	Note that any $v^r_\textsf{in} \to v^r_\textsf{out}$-path in $R^{a,b}$ has to start with a subpath $v^r_\textsf{in}, g^r_\textsf{in}, x^p_2, y^p_2$ and to end with a (potentially overlapping) subpath $x^{p'}_2, y^{p'}_2, g^r_\textsf{out}, v^r_\textsf{out}$ for some (not necessarily distinct) \mbox{$p, p' \in \lbrace -1, 0, \hdots, a+1\rbrace$}.
	As $(x^p_2, y^p_2)$ is contained in $S_{p^\star}$ for any $p \neq p^\star$, we have that the only  $v^r_\textsf{in} \to v^r_\textsf{out}$-path is $v^r_\textsf{in}, g^r_\textsf{in}, x^p_2, y^p_2, g^r_\textsf{out}, v^r_\textsf{out}$.
	Its weight is
	\begin{equation*}
      -(A - p^\star) - (A + p^\star) - M + \varphi_r(p^\star) +2A+M = \varphi_r(p^\star)\enspace .\qedhere
	\end{equation*}
  \end{claimproof}
	
  \begin{claim}
  \label{claim:treedepth_hardness_general_solution_path_length}
    Let $S$ be a solution to $R^{a,b}$ as {\sc Negative DFAS} instance of size $5(a + 3)$.
	Then there is a $p^\star \in \lbrace 1, \hdots, a\rbrace$ such that any $v^r_\textsf{in} \to v^r_\textsf{out}$-path in $R^{a,b} - S$ has weight $\varphi^r(p^\star)$ for any $r \in \lbrace 1, \hdots, b\rbrace$.
  \end{claim}
  \begin{claimproof}
	We want to prove that $S$ has the structure of undeleted arcs that is show in \Cref{fig:undeleted_arc_structure}.

    \begin{figure}
		\centering
		\begin{tikzpicture}[scale=0.9, hollow/.style={circle, minimum size=1.8mm}, filled/.style={fill, hollow}]
			\tikzset{>={Latex[width=2mm, length=2mm]}}
			
			\foreach \i in {1, 2, 3, 4}{
				\node at (-1.2, 4-\i) {$x^p_{\i}$/$y^p_{\i}$};
			}

			\foreach \x in {1.5, 5.4, 11.25}{
				\node at (\x, 1.5) {$\hdots$};
			}
			
			\foreach \num/\p/\x in {%
				1/-1/0,
				2/\hat{p}/2.25,
				3/\hat{p}+1/4.0,
				4/\tilde{p}/6.25,
				5/p^\star/8.0,
				6/\check{p}/9.75,
				7/a+1/12.0%
			}{
				\node at (\x+0.5, 3.7) {$\p$};
				\foreach \i in {1, 2, 3, 4}{
					\node (x\i\num) [filled, draw, inner sep=-.05mm] at (\x,4-\i) {};
					\node (y\i\num) [filled, draw, inner sep=-.05mm] at (\x+0.75,4-\i) {};
				}
			}
			
			\foreach \num/\i in {1/1, 2/1, 3/3, 4/3, 5/2, 6/4, 7/4}{
				\draw (x\i\num) edge[->, ultra thick]  (y\i\num);
			}
				
		\end{tikzpicture}
		\caption{Structure of undeleted arcs $(x_i^p, y_i^p)$ as proven in \Cref{claim:treedepth_hardness_general_solution_path_length}.\label{fig:undeleted_arc_structure}}
	\end{figure}
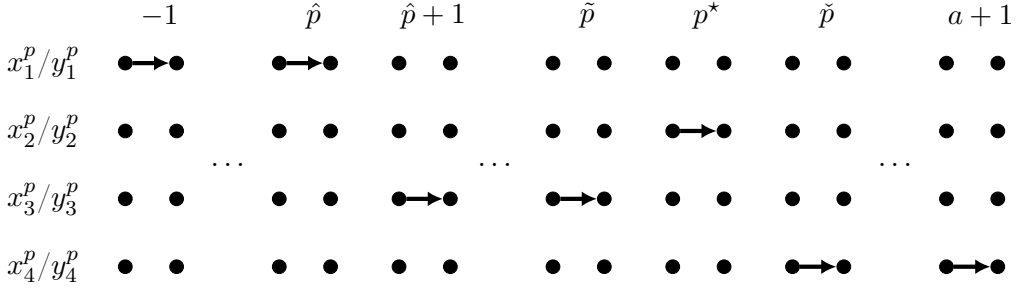	
	
    By \Cref{claim:structure_of_mini_gadget_solutions}, our solution $S$ has to delete exactly three arcs of type $(s^p_i, t^p_i)$ in every mini-gadget, which is equivalent to leaving exactly one arc $(x^p_{i_p}, y^p_{i_p})$ undeleted for every $p \in \lbrace 0, 1, \hdots, a+1\rbrace$ with $i_p \in \lbrace 1,2,3,4\rbrace$.
    Note that by doubling $(x^{-1}_1, y^{-1}_1)$ and $(x^{a+1}_4, y^{a+1}_4)$, these are the undeleted arcs for $p \in \lbrace -1, a+1\rbrace$.
		
    For ease of notation, let $w_i(p)$ denote the weight of the arc $(x^p_i, y^p_i)$ for any $i \in \lbrace 1, 2, 3, 4\rbrace$ and any $p \in \lbrace -1, 0, \hdots, a+1\rbrace$.
    Then for any two distinct indices \mbox{$p, p' \in \lbrace -1, 0, \hdots, a+1\rbrace$} with $|i_p - i_{p'}|$ odd, we have that the cycle
    \begin{equation*}
      (s_{i_p}, x^{p}_{i_p}) \circ (x^{p}_{i_p}, y^{p}_{i_p}) \circ (y^{p}_{i_p}, t_{i_p}) \circ (t_{i_p}, s_{i_{p'}}) \circ (s_{i_{p'}}, x^{p'}_{i_{p'}}) \circ (x^{p'}_{i_{p'}}, y^{p'}_{i_{p'}}) \circ (y^{p'}_{i_{p'}}, t_{i_{p'}}) \circ (t_{i_{p'}}, s_{i_p})
    \end{equation*}
	exists in $R^{a,b} - S$.
	This cycle we denote by $C_{p,p'}$.
	The weight of this cycle is $2A + w_{i_p}(p) + w_{i_{p'}}(p')$.
	As $S$ is a solution to~$R^{a,b}$ as \NDFAS{} instance, these cycles always have non-negative weight.
		
	Let $\hat{p}$ be the maximum $p$ such that $(x^p_1, y^p_1)$ is undeleted, i.e. $\hat{p} = \max \lbrace p \mid i_p = 1\rbrace$.
	Analogously, let $\check{p}$ be the minimum~$p$ such that $(x^p_4, y^p_4)$ is undeleted, i.e. $\check{p} = \min \lbrace p \mid i_p = 4\rbrace$.
	Note that both exist, since we doubled the arcs $(x^{-1}_1, y^{-1}_1)$ and $(x^{p+1}_4, y^{p+1}_4)$.
	So the cycle $C_{\check{p}, \hat{p}}$ exists and has non-negative weight $0 \leq 2A -(A + (\hat{p}+2)) - (A - (\check{p} -1)) = \check{p} - \hat{p} - 3$.
	This is equivalent to $\hat{p} + 3 \leq \check{p}$.
	So for every index $\hat{p} < p < \check{p}$, both arcs $(x^p_1, y^p_1)$ and $(x^p_4, y^p_4)$ are deleted and there are at least two such indices.

  Next we show that $(x^{\hat{p}+1}_2, y^{\hat{p} +1}_2)$ is also deleted.
  Otherwise, we have that $i_{\hat{p} + 1} = 2$ and thus the cycle $C_{\hat{p}, \hat{p} + 1}$ has weight $2A -(A + (\hat{p}+2)) - (A - (\hat{p} +1)) = -1 $, which is a contradiction to $C_{\hat{p}, \hat{p} + 1}$ having non-negative weight.
  So we have $i_{\hat{p} +1} = 3$.
  Now define $\tilde{p}$ to be the maximum~$p$ such that $(x^p_3, y^p_3)$ is undeleted, i.e. $\tilde{p} = \max \lbrace p \mid i_p = 3\rbrace$.
  By the previous argument we know that this maximum exists and $\hat{p} +1 \leq \tilde{p}$.
		
  Consider the cycle $C_{\check{p}, \tilde{p}}$.
	Since this cycle again has non-negative weight, we have $0 \leq 2A -(A + (\check{p} - 1)) - (A - (\tilde{p} +1)) = \check{p} - \tilde{p} - 2$.
  This is equivalent to $\tilde{p} + 2 \leq \check{p}$.
  Note that for the index $\tilde{p} + 1$, we have that $\hat{p}, \tilde{p} < \tilde{p} + 1 < \check{p}$.
  That means the only arc that can be undeleted for $\tilde{p} + 1$ is $(x^{\tilde{p}+1}_2, y^{\tilde{p}+1}_2)$, i.e. $i_{\tilde{p} + 1} = 2$.
  We want to show that $p^\star = \tilde{p} + 1$ is indeed the only index $p$ with $i_p = 2$.
  For this let $\overleftarrow{p}^\star = \min \lbrace p \mid i_p = 2\rbrace$ and $\overrightarrow{p}^\star = \max \lbrace p \mid i_p = 2\rbrace$.
  By the weight of $C_{\overleftarrow{p}^\star, \tilde{p}}$ being non-negative, we get $0 \leq 2A -(A - \overleftarrow{p}^\star) - (A + (\tilde{p} +1)) = \overleftarrow{p}^\star - \tilde{p} - 1$, i.e. $\tilde{p} + 1 \leq \overleftarrow{p}^\star$.
		
  Now consider any $i \in \lbrace 2,3\rbrace$ and any $p \in \lbrace -1, 0, \hdots, a+1 \rbrace$.
  By \Cref{claim:structure_of_mini_gadget_solutions}, we have that $S$ contains either $(\dot{x}^p_i, \dot{y}^p_i)$ or $(\ddot{x}^p_i, \ddot{y}^p_i)$, but none of the arcs $(y^p_i, \dot{x}^p_i)$ and $(\dot{y}^p_i, x^p_i)$.
  Thus, if $(x^p_i, y^p_i)$ is undeleted, by the negative cycle $(x^p_i, y^p_i) \circ (y^p_i, \dot{x}^p_i) \circ (\dot{x}^p_i, \dot{y}^p_i) \circ (\dot{y}^p_i, x^p_i)$, we have that $(\dot{x}^p_i, \dot{y}^p_i)$ is deleted and thus $(\ddot{x}^p_i, \ddot{y}^p_i)$ is undeleted.
  So for $\tilde{p}$, the arc $(\ddot{x}^{\tilde{p}}_3, \ddot{y}^{\tilde{p}}_3)$ is undeleted, and for $\overrightarrow{p}^\star$ the arc $(\ddot{x}^{\overrightarrow{p}^\star}_2, \ddot{y}^{\overrightarrow{p}^\star}_2)$ is undeleted.
  By \Cref{claim:structure_of_mini_gadget_solutions}, we also get that none of the other arcs of the following cycle $\ddot{C}_{\overrightarrow{p}^\star, \tilde{p}}$ are deleted:
  \begin{equation*}
    (\ddot{s}_2,  \ddot{x}^{\overrightarrow{p}^\star}_2) \circ (\ddot{x}^{\overrightarrow{p}^\star}_2, \ddot{y}^{\overrightarrow{p}^\star}_2) \circ (\ddot{y}^{\overrightarrow{p}^\star}_2, \ddot{t}_2) \circ (\ddot{t}_2, \ddot{s}_3) \circ (\ddot{s}_3, \ddot{x}^{\tilde{p}}_3) \circ (\ddot{x}^{\tilde{p}}_3, \ddot{y}^{\tilde{p}}_3) \circ (\ddot{y}^{\tilde{p}}_3, \ddot{t}_3) \circ (\ddot{t}_3, \ddot{s}_2)\enspace .
  \end{equation*}
  As this cycle exists in $R^{a,b} - S$, it has non-negative weight, which implies
  \begin{equation*}
	  0 \leq -(A + \overrightarrow{p}^\star) + A - (A - (\tilde{p} + 1)) + A = \tilde{p} - \overrightarrow{p}^\star + 1\enspace .
  \end{equation*}
  This is equivalent to $\overrightarrow{p}^\star \leq \tilde{p} + 1$.
	Overall we have $\tilde{p} + 1 \leq \overleftarrow{p}^\star \leq \overrightarrow{p}^\star \leq \tilde{p} + 1$.
	So indeed the only index $p$ with $(x^p_2, y^p_2)$ undeleted is $p^\star = \tilde{p} + 1$.
		
	Finally, consider the $v^r_\textsf{in} \to v^r_\textsf{out}$-paths in $R^{a,b} - S$ for $r \in \lbrace 1, \hdots, b\rbrace$.
	Note that each such path starts with the subpath $v^r_\textsf{in}, g^r_\textsf{in}, x^p_2, y^p_2$ and ends with the (potentially overlapping) subpath $x^{p'}_2, y^{p'}_2, g^r_\textsf{out}, v^r_\textsf{out}$ for $p, p' \in \lbrace -1, 0, \hdots, a+1\rbrace$.
	As the only index with $(x^p_2, y^p_2)$ undeleted is $p^\star = \tilde{p} + 1$, we have that the only $v^r_\textsf{in} \to v^r_\textsf{out}$-path in $R^{a,b} - S$ has the form $(v^r_\textsf{in}, g^r_\textsf{in}, x^{p^\star}_2, y^{p^\star}_2, g^r_\textsf{out}, v^r_\textsf{out})$.
	This path has weight
	\begin{equation*}
      -(A - p^\star) - (A + p^\star) - M + \varphi_r(p^\star) +2A+M = \varphi_r(p^\star)\enspace .
    \end{equation*}
    It only remains to show that $p^\star \in \lbrace 1, \hdots, a\rbrace$.
	For this note that $-1 \leq \hat{p} \leq \tilde{p} - 1$ which is equivalent to $1 \leq \tilde{p} + 1 = p^\star$, and that $p^\star = \tilde{p} + 1 \leq \check{p} - 1 \leq a$.
  \end{claimproof}
	
  After proving what minimum solutions look like for our meta gadget, we can now give the overall reduction from \textsc{Clique}.
  Let $v_1, \hdots, v_n$ be an arbitrary ordering of $V(G)$ and $e_1, \hdots, e_m$ be an arbitrary ordering of $E(G)$.
  For every $i \in \lbrace 1,\hdots, k\rbrace$, we introduce a vertex gadget $H^i$ which is a copy of the meta gadget $R^{a,b}$ with $a = n$, $b=2$, $M = n$ and the functions $\varphi_1(p) = p$ and $\varphi_2(p) = -p$.
  Moreover, we rename the vertices $v^1_\textsf{in}$ and $v^1_\textsf{out}$ of $H_i$ to $s$ and $z^+_i$, and the vertices~$v^2_\textsf{in}$ and $v^2_\textsf{out}$ of $H_i$ to $s$ and $z^-_i$.
  For every $i,j \in \lbrace 1,\hdots, k\rbrace$ with $i < j$, we introduce an edge gadget $H^{i,j}$ which is a copy of the meta gadget $R^{a,b}$ with $a = m$, $b=4$, $M = n$ and the functions defined as follows.
  For every $e_\ell = \lbrace v_p, v_q\rbrace \in E(G)$ with $p \leq q$ we let $\varphi_1(\ell) = -p$, $\varphi_2(\ell) = p$, $\varphi_3(\ell) = -q$ and $\varphi_4(\ell) = q$.
  Moreover, we rename the vertices in the following way:
  \begin{itemize}
    \item $v^1_\textsf{in}$ and $v^1_\textsf{out}$ of $H_{i,j}$ become $z^+_i$ and $t$,
	\item $v^2_\textsf{in}$ and $v^2_\textsf{out}$ of $H_{i,j}$ become $z^-_i$ and $t$,
	\item $v^3_\textsf{in}$ and $v^3_\textsf{out}$ of $H_{i,j}$ become $z^+_j$ and $t$,
	\item $v^4_\textsf{in}$ and $v^4_\textsf{out}$ of $H_{i,j}$ become $z^+_j$ and $t$,
  \end{itemize}
	
  Any renamed vertices sharing the same name are identified with each other.
  Finally, we add an arc $(t,s)$ of weight $0$.
  Call the resulting graph $H$.
  We claim that $H$ as \NDFAS{} instance has a solution of size at most $d = k \cdot 5(n+3) + \frac{1}{2}k(k+1) \cdot 5(m+3)$ if and only if $G$ has a clique of size $k$.
	
  For the forward direction let $S$ be any solution to $H$ as \NDFAS{} instance of size at most $d$.
  By \Cref{claim:structure_of_mini_gadget_solutions}, we know that $S$ restricted to any $H^i$ has size at least $5(n+3)$ and restricted to any $H^{i,j}$ has size at least $5(m+3)$.
  By $d$ matching exactly the total of these numbers, we know that $S$ has exactly these sizes in the gadgets.
  Moreover, $(t,s)$ is not contained in $S$.
	
  Using \Cref{claim:treedepth_hardness_general_solution_path_length}, we get that for every $H^i$ there is a $p_i \in \lbrace 1, \hdots, n\rbrace$ such that any $s \to z^+_i$-path has weight $p_i$ and any $s \to z^-_i$-path has weight $-p_i$.
  Also, we get that for every $H^{i,j}$ there is an~$\ell_{i,j}$ such that for $e_{\ell_{i,j}} = \lbrace v_p, v_q\rbrace$ we have that any $z^+_i \to t$-path has weight $-p$, any $z^-_i \to t$-path has weight $p$, any $z^+_j \to t$-path has weight $-q$ and any $z^-_j \to t$-path has weight $q$.	
  Let now
  \begin{equation*}
    V' = \lbrace v_{p_i} \mid i \in \lbrace 1, \hdots, k\rbrace \rbrace \subseteq V(G) \textsf{ and } E' = \lbrace e_{\ell_{i,j}} \mid i,j \in \lbrace 1, \hdots, k\rbrace, i < j \rbrace \subseteq E(G)\enspace .
  \end{equation*}
  We claim that $(V', E')$ is a clique in $G$ of size $k$.
  As $V'$ and $E'$ have exactly the right cardinality, we only have to proof that for any unordered pair $v_{p_i}, v_{p_j} \in V'$ there is an edge between them.
  Without loss on generality we can assume $i < j$, otherwise swap the indices.
  Moreover, let $e_{\ell_{i,j}} = \lbrace v_p, v_q\rbrace$.
	
  Consider the cycle consisting of the arc $(t,s)$, the $s \to z^+_i$-path in $H^i - S$ and the $z^+_i \to t$-path in $H^{i,j} - S$.
  This cycle is non-negative and has weight $p_i - p$, thus $p \leq p_i$.
  By considering the cycle consisting of the arc $(t,s)$, the $s \to z^-_i$-path in $H^i - S$ and the $z^-_i \to t$-path in $H^{i,j} - S$, we get a non-negative cycle of weight $-p_i + p$.
  Thus, $p\geq p_i$ and combined with the former inequality, $p = p_i$.
  Analogously, by using the paths starting and ending in $z^+_j$ and $z^-_j$ instead, we get that $p_j = q$.
  Thus, we get that indeed the edge $\lbrace v_{p_i}, v_{p_j}\rbrace  = \lbrace v_p, v_q\rbrace = e_{\ell_{i,j}} \in E(G)$.
  So~$G$ contains the clique $(V', E')$ of size $k$.
  \smallskip
	
  For the backward direction consider any clique on vertices $v_{p_1}, \hdots, v_{p_k}$ in $G$.
  Choose by \Cref{claim:treedepth_hardness_explicit_solution} a solution $S^i$ of size $5(n+3)$ for every gadget $H^i$ such that any $s \to z^+_i$-path has weight~$p_i$ and any $s \to z^-_i$-path has weight $-p_i$.
  Analogously, by \Cref{claim:treedepth_hardness_explicit_solution} choose a solution $S^{i,j}$ of size $5(n+3)$ for every gadget $H^{i,j}$ such that any $z^+_i \to t$-path has weight $-p_i$, any $z^-_i \to t$-path has weight $p_i$, any $z^+_j \to t$-path has weight $-p_j$ and any $z^-_j \to t$-path has weight $p_j$.
  Note that the latter solution can be chosen that way, because $\lbrace v_{p_i}, v_{p_j} \rbrace \in E(G)$ and $i < j$.
	
  We claim that $S = \bigcup_{i} S^i \cup \bigcup_{i < j} S^{i,j}$ is a solution to $H$ of size $d$.
  The size bound follows from adding up the sizes of the individual solutions.
  It remains to check that $H - S$ contains no negative cycles.
  As $S$ constraint to a single gadget was chosen as a solution, the only negative cycles that can exists use some high-level vertices.
  There are two different types of such cycles.
  The first type uses the arc $(t,s)$, followed by an $s \to z^+_i$-path in $H^i$ for some $i$ and an $z^+_i \to t$-path in some $H^{i,j}$ or $H^{j,i}$.
  Our partial solutions were chosen such that these cycles have weight $0 + p_i - p_i$.
  The other type uses the arc $(t,s)$, followed by an $s \to z^-_i$-path in $H^i$ for some $i$ and an $z^-_i \to t$-path in some $H^{i,j}$ or $H^{j,i}$.
  But by choice of $S$ these paths have weight $0 - p_i + p_i$.
  So $H - S$ contains no negative cycles.
	
  It remains to show that the treedepth and the number of arcs with positive weight is bounded.
  First we will derive bounds for these parameters on the meta-gadgets $R^{a,b}$.
  For the treedepth we will use the recursive formula of \Cref{def:treedepth}.
  Consider the set
  \begin{equation*}
		D = \lbrace s^1, s^2, s^3, s^4, t^1, t^2, t^3, t^4, \ddot{s}_2, \ddot{s}_3, \ddot{t}_2, \ddot{t}_3 \rbrace \cup \lbrace g^r_\textsf{in}, g^r_\textsf{out}, v^r_\textsf{in}, v^r_\textsf{out}  \mid r \in \lbrace 1, \hdots, b\rbrace \rbrace \enspace .
  \end{equation*}
  Note that $R^{a,b} - D$ consists only of the mini-gadgets, which are connected components of size $16$ each.
  Thus, $\td(R^{a,b}) \leq |D| + 16 = 28 + 4 b$.
	
  For the number of arcs with positive weight note that $A$ and $M$ are chosen such that the only positive arc weights are $A$ and $2A + M$.
  Of the former there are ten and of the later there are $b$ many in every meta gadget.
  Thus, the number of arcs with positive weight in $R^{a,b}$ is bounded by $10 + b$.
	
  For the graph $H$, we can combine the above bounds to bound its treedepth and its number of positive arcs.
  Regarding the treedepth of $H$, note that for $D' = \lbrace s, t \rbrace \cup \lbrace z^+_i, z^-_i \mid i \in \lbrace 1, \hdots,k\rbrace \rbrace$, the graph $H - D'$ consists of connected components that only form subgraphs of the $H^i$ and~$H^{i,j}$.
  By \Cref{def:treedepth}, it holds
  \begin{align*}
	\td(H)	&\leq |D'| + \td(H - D') \leq |D'| + \sum_{i=1}^k \left(\td(H^i) + \sum_{j=i+1}^k \td(H^{i,j})\right)\\
				&\leq |D'| + 36 k + 22k(k+1) \leq 22k^2 + 60k + 2 \enspace .
  \end{align*}
	
  Regarding the number of positive arcs in $H$, note that we did not add any new arc of positive weight.
  Thus, the number of positive arcs in $H$ is bounded by $12k + 7k(k+1)$.
\end{proof}

\subsection{\texorpdfstring{\Whard{1}ness}{W[1]-hardness} for Pathwidth, Deletion Size and Few Positive Arcs}
In this section we prove \Whard{1}ness for \textsc{MinFB} with general integral weights when parameterized in $\pw(G) + k + w_+$.
To this end, we make a reduction from \MulticoloredClique{} to \NDFAS{}.
\MulticoloredClique{} is the problem of deciding whether a graph $G$ with a proper vertex-coloring $\psi:V(G) \rightarrow \{1,\hdots,k\}$ contains a clique of size $k$; the problem is $\mathsf{W}[1]$-hard parameterized by $k$~\cite{FellowsEtAl2009}.
In addition, we will see how this implies \Whard{1}ness for \textsc{MinFB} parameterized in $\pw(G) + k$ for instances with weights of the form $w: A(G) \to \lbrace -1, 0, 1\rbrace$. 

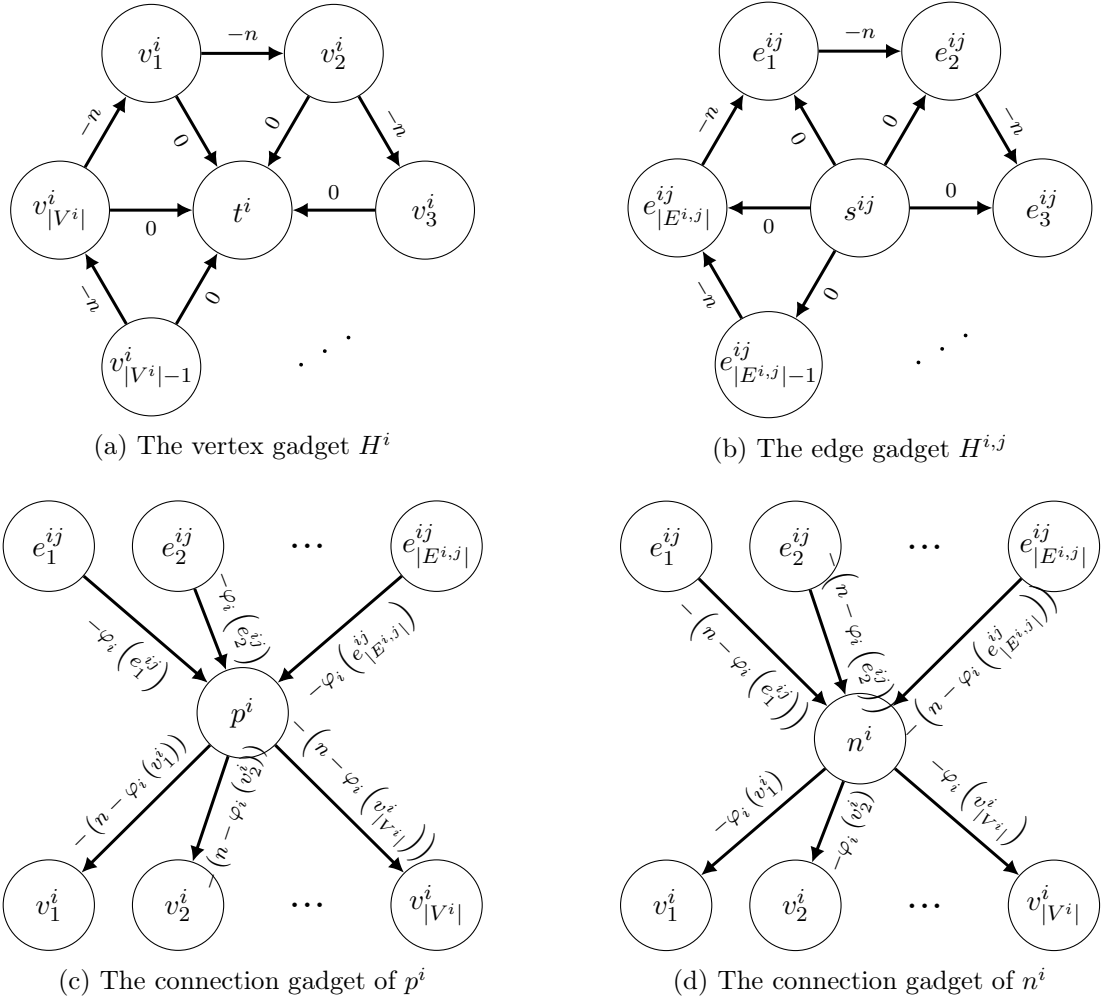
\begin{figure}
  \captionsetup[subfigure]{justification=centering}
\begin{center}

\def \radius {3cm}
\begin{subfigure}[c]{0.49\textwidth}
\centering
\begin{tikzpicture}[scale=0.8, state/.style={circle,  minimum size=1.3cm}]
	\tikzset{>={Latex[width=2mm, length=2mm]}}	
	\node (1)[state, draw, inner sep=-.05mm] at ({360/6 * (6)+120}:\radius) {$v_ 1^i$};
	\node (2)[state, draw, inner sep=-.05mm] at ({360/6 * (5 )+120}:\radius) {$v_ 2^i$};
	\node (3)[state, draw, inner sep=-.05mm] at ({360/6 * (4 )+120}:\radius) {$v_ 3^i$};
	\node (4)[state, inner sep=-.05mm] at ({360/6 * (3 )+120}:.9*\radius) {\begin{turn}{30} \LARGE .\hspace*{.2cm}.\hspace*{.2cm}. \end{turn}};
	\node (5)[state, draw, inner sep=-.05mm] at ({360/6 * (2 )+120}:\radius) {$v_ {|V^i|-1}^i$};
	\node (6)[state, draw, inner sep=-.05mm] at ({360/6 * (1 )+120}:\radius) {$v_ {|V^i|}^i$};
	\node (7)[state, draw, inner sep=-.05mm] at ({0}:0) {$t^i$};
	
	\foreach \a / \b / \label / \orientation /\pos  in {
		1/2/$-n$/above/0.5,
		2/3/$-n$/above/0.5,
		5/6/$-n$/below/0.5,
		6/1/$-n$/above/0.5,
		1/7/$0$/below/0.5,
		2/7/$0$/above/0.5,
		3/7/$0$/above/0.5,
		5/7/$0$/below/0.5,
		6/7/$0$/below/0.5%
		}{
			\draw (\a) edge[->, line width=0.4mm] node[sloped, \orientation, pos=\pos]{ \scriptsize \label}  (\b);
		}
\end{tikzpicture}
\subcaption{The vertex gadget $H^i$\label{fig:pathwidth_vertex_gadget}}
\end{subfigure}
\hfill
\begin{subfigure}[c]{0.49\textwidth}
\centering

\begin{tikzpicture}[scale=0.8, state/.style={circle,  minimum size=1.3cm}]
	\tikzset{>={Latex[width=2mm, length=2mm]}}	
	\node (1)[state, draw, inner sep=-.05mm] at ({360/6 * (6)+120}:\radius) {$e_ 1^{ij}$};
	\node (2)[state, draw, inner sep=-.05mm] at ({360/6 * (5 )+120}:\radius) {$e_ 2^{ij}$};
	\node (3)[state, draw, inner sep=-.05mm] at ({360/6 * (4 )+120}:\radius) {$e_ 3^{ij}$};
	\node (4)[state, inner sep=-.05mm] at ({360/6 * (3 )+120}:.9*\radius) {\begin{turn}{30} \LARGE .\hspace*{.2cm}.\hspace*{.2cm}. \end{turn}};
	\node (5)[state, draw, inner sep=-.05mm] at ({360/6 * (2 )+120}:\radius) {$e_ {|E^{i,j}|-1}^{ij}$};
	\node (6)[state, draw, inner sep=-.05mm] at ({360/6 * (1 )+120}:\radius) {$e_ {|E^{i,j}|}^{ij}$};
	\node (7)[state, draw, inner sep=-.05mm] at ({0}:0) {$s^{ij}$};
	
		\foreach \a / \b / \label / \orientation /\pos  in {
		1/2/$-n$/above/0.5,
		2/3/$-n$/above/0.5,
		5/6/$-n$/below/0.5,
		6/1/$-n$/above/0.5,
		7/1/$0$/below/0.5,
		7/2/$0$/above/0.5,
		7/3/$0$/above/0.5,
		7/5/$0$/below/0.5,
		7/6/$0$/below/0.5%
		}{
			\draw (\a) edge[->, line width=0.4mm] node[sloped, \orientation, pos=\pos]{ \scriptsize \label}  (\b);
		}
\end{tikzpicture}
\subcaption{The edge gadget $H^{i,j}$\label{fig:pathwidth_edge_gadget}}
\end{subfigure}
\par\bigskip
\begin{subfigure}[c]{0.49\textwidth}
\centering

\begin{tikzpicture}[scale=1.7, state/.style={circle,  minimum size=1.2cm}]
	\tikzset{>={Latex[width=2mm, length=2mm]}}	
	\node (1) at (0,0) [state, draw, inner sep=-.05mm] {$v_1^i$};
	\node (2) at (1,0) [state, draw, inner sep=-.05mm] {$v_2^i$};
	\node (3) at (2,0) [state, inner sep=-.05mm] {\LARGE $...$ };
	\node (4) at (3,0) [state, draw, inner sep=-.05mm] {$v_{|V^i|}^i$};
	\node (5) at (1.5,1.5) [state, draw, inner sep=-.05mm] {$p^i$};
	\node (6) at (0,2.8) [state, draw,inner sep=-.05mm] {$e_{1}^{ij}$};
	\node (7) at (1,2.8) [state, draw, inner sep=-.05mm] {$e_{2}^{ij}$};
	\node (8) at (2,2.8) [state, inner sep=-.05mm] {\LARGE $...$ };
	\node (9) at (3,2.8) [state, draw, inner sep=-.05mm] {$e_{|E^{i,j}|}^{ij}$};
	
	\foreach \a / \b / \label / \orientation / \pos  in {
		5/1/$-\left(n -{\varphi }_{i} \left( v_1^i \right)\right)$/above/0.5,
		5/2/$ -\left(n -{\varphi }_{i} \left( v_2^i \right)\right)$/below/0.5,
		5/4/$ -\left(n -{\varphi }_{i} \left( v_{|V^i|}^i \right)\right)$/above/0.5,
		6/5/$ {-\varphi }_{i} \left( e_{1}^{ij} \right)$/below/0.5,
		7/5/$ {-\varphi }_{i} \left( e_{2}^{ij} \right)$/above/0.5,
		9/5/$ {-\varphi }_{i} \left( e_{|E^{i,j}|}^{ij} \right)$/below/0.5%
		}{
			\draw (\a) edge[->, line width=0.4mm] node[sloped, \orientation, pos=\pos]{ \scriptsize \label}  (\b);
		}
\end{tikzpicture}
\subcaption{The connection gadget of $p^i$\label{fig:pathwidth_pos_gadget}}
\end{subfigure}
\hfill
\begin{subfigure}[c]{0.49\textwidth}
\centering

\begin{tikzpicture}[scale=1.7, state/.style={circle,  minimum size=1.2cm}]
	\tikzset{>={Latex[width=2mm, length=2mm]}}	
	\node (1) at (0,0.2) [state, draw, inner sep=-.05mm] {$v_1^i$};
	\node (2) at (1,0.2) [state, draw, inner sep=-.05mm] {$v_2^i$};
	\node (3) at (2,0.2) [state, inner sep=-.05mm] {\LARGE $...$ };
	\node (4) at (3,0.2) [state, draw, inner sep=-.05mm] {$v_{|V^i|}^i$};
	\node (5) at (1.5,1.5) [state, draw, inner sep=-.05mm] {$n^i$};
	\node (6) at (0,3) [state, draw,inner sep=-.05mm] {$e_{1}^{ij}$};
	\node (7) at (1,3) [state, draw, inner sep=-.05mm] {$e_{2}^{ij}$};
	\node (8) at (2,3) [state, inner sep=-.05mm] {\LARGE $...$ };
	\node (9) at (3,3) [state, draw, inner sep=-.05mm] {$e_{|E^{i,j}|}^{ij}$};
	
	\foreach \a / \b / \label / \orientation / \pos  in {
		5/1/${-\varphi }_{i} \left( v_1^i \right)$/above/0.5,
		5/2/${-\varphi }_{i} \left( v_2^i \right)$/below/0.5,
		5/4/${-\varphi }_{i} \left( v_{|V^i|}^i \right)$/above/0.5,
		6/5/$-\left(n -{\varphi }_{i} \left( e_{1}^{ij} \right)\right)$/below/0.5,
		7/5/$-\left(n -{\varphi }_{i} \left( e_{2}^{ij} \right)\right)$/above/0.5,
		9/5/$-\left(n -{\varphi }_{i} \left( e_{|E^{i,j}|}^{ij} \right)\right)$/below/0.5%
		}{
			\draw (\a) edge[->, line width=0.4mm] node[sloped, \orientation, pos=\pos]{ \scriptsize \label}  (\b);
		}
\end{tikzpicture}
\subcaption{The connection gadget of $n^i$\label{fig:pathwidth_neg_gadget}}
\end{subfigure}

\end{center}
\caption{Gadgets used in the construction of \Cref{thm:W1_hardness_for_pathwidth_and_deletion_size}\label{fig:pathwidth_hardness}}
\end{figure}

\begin{theorem}
\label{thm:W1_hardness_for_pathwidth_and_deletion_size}
  {\sc Negative DFAS} is \Whard{1} when parameterized by the pathwidth, deletion size and number of positive arcs.
  This holds already for instances $(G, w, k)$, where $w: A(G) \to \mathbb{Z} \cap [-|V(G)|^2, |V(G)|^2]$.
\end{theorem}
\begin{proof}
  We prove \Whard{1}ness by showing a parameterized reduction from \MulticoloredClique{}.
  For this purpose, let $(G, k)$ an instance of \MulticoloredClique{} with $V(G) = \biguplus_{i=1}^k V^i$, where each $V^i$ induces an independent set of $G$.
  For ease of notation, define sets $E^{i,j}$ as the set of edges of type $\lbrace v_i, v_j\rbrace$ with $v_i \in V^i$ and $v_j \in V^j$, for $i,j = 1,\hdots,k$ with $i< j$.
  Moreover, we abuse this notation slightly by defining $E^{j,i} = E^{i,j}$ but with the order of vertices switched, i.e. $\lbrace v_i, v_j\rbrace \in E^{i,j}$ but $\lbrace v_j, v_i\rbrace \in E^{j,i}$.
	
  We construct an equivalent instance $(H, w, d)$ of \NDFAS{}, see \Cref{fig:pathwidth_hardness} for an illustration.
  The graph $H$ will have three types of gadgets:
  \begin{enumerate}
    \item a vertex gadget $H^i$ representing $V^i$ for every $i \in \lbrace 1, \hdots, k\rbrace$,
    \item an edge gadget $H^{i,j}$ representing $E^{i,j}$ for every $1 \leq i < j \leq k$, and
    \item a consistency gadget $C^i$ for every $i \in \lbrace 1, \hdots, k\rbrace$.
  \end{enumerate}
	
  The gadgets $H^i$ representing $V^i = \lbrace v^i_1, \hdots, v^i_{|V^i|}\rbrace$ consist of one vertex $t^i$ and the vertices of~$V^i$.
  Each gadget $H^i$ contains two types of arcs.
  First, it has arcs $(v^i_r, t^i)$ of weight $0$ for every $v^i_r \in V^i$, and second, arcs $(v^i_r, v^i_{r+1})$ of weight~$-n$ for every $r \in \lbrace 1, \hdots, |V^i|\rbrace$, where we define $v^i_{|V^i| +1} = v^i_1$.
		
  The gadgets $H^{i,j}$ representing $E^{i,j} = \lbrace e^{i,j}_1, \hdots, e^{i,j}_{|E^{i,j}|}\rbrace$ consist of one vertex $s^{i,j}$ and one vertex for every $e^{i,j}_r \in E^{i,j}$.
  Each gadget $H^{i,j}$ contains two types of arcs.
  First, it has arcs $(s^{i,j}, e^{i,j}_r)$ of weight $0$ for every $e^{i,j}_r \in E^{i,j}$, and second, it has arcs $(e^{i,j}_r, e^{i,j}_{r+1})$ of weight $-n$ for every $r \in \lbrace 1, \hdots, |E^{i,j}|\rbrace$, where we define $e^{i,j}_{|E^{i,j}| +1} = e^{i,j}_1$.
	
  Finally, the consistency gadgets $C_i$ consist of the vertices $n^i$ and $p^i$ and interconnect the existing gadgets as follows.
  First there is an arc $(t^i, s^{i,j})$ of weight $n (|V^i| + |E^{i,j}| - 1)$ for every $j \in \lbrace 1, \hdots, k \rbrace \setminus \lbrace i \rbrace$.
  Then we need a bijective function $\varphi_i: V^i \to \lbrace 1, \hdots, |V^i|\rbrace$ (for example, $\varphi_i(v^i_r) = r$ will do).
  We abuse the notation and that every edge of $E$ has at most one endpoint in every $V^i$ and define for $e^{i,j} = \lbrace v^i, v^j\rbrace \in E^{i,j}$ that $\varphi_i(e^{i,j}) = \varphi_i(v^i)$.
  With this definition in mind, we add the following arcs to $C^i$:
  \begin{itemize}
    \item $(e^{i,j}, n^i)$ of weight $-(n - \varphi_i(e^{i,j}))$ for every $j \in \lbrace 1, \hdots, k \rbrace \setminus \lbrace i \rbrace$ and $e^{i,j} \in E^{i,j}$,
    \item $(e^{i,j}, p^i)$ of weight $-\varphi_i(e^{i,j})$ for every $j \in \lbrace 1, \hdots, k \rbrace \setminus \lbrace i \rbrace$ and $e^{i,j} \in E^{i,j}$,
    \item $(n^i, v^i)$ of weight $-\varphi_i(v^i)$ for every $v^i \in V^i$, and
    \item $(p^i, v^i)$ of weight $-(n - \varphi_i(v^i))$ for every $v^i \in V^i$.
  \end{itemize}
	
  Let $H$ denote the resulting graph consisting of all the gadgets, and let $w$ denote the corresponding weight function.
  By choosing $d = k + \binom{k}{2}$, we obtain our \NDFAS{} instance $(H, w, d)$.
	
  \smallskip
  \noindent
  \textbf{Negative DFAS to Multicolored Clique.}
  Let $S$ be a solution to $(H, w, d)$.
  The gadgets $H^i$ and $H^{i,j}$ each contain a negative cycle consisting of the arcs $(v^i_r, v^i_{r+1})$ or $(e^{i,j}_r, e^{i,j}_{r+1})$, respectively.
  As there are exactly $k + \binom{k}{2}$ such gadgets and their cycles are disjoint, $S$ has size exactly $d$ and contains exactly one arc $(v^i_{r_i}, v^i_{r_i+1})$ for every $i \in \lbrace 1, \hdots, k\rbrace$ and one arc $(e^{i,j}_{r_{i,j}}, e^{i,j}_{r_{i,j}+1})$ for every $(i,j)$ with $1 \leq i < j \leq k$.
  Define $K = (\lbrace v^i_r \mid i \in \lbrace 1, \hdots, k\rbrace, r \in \lbrace 2, \hdots, |V^i|+1\rbrace, (v^i_{r-1}, v^i_{r}) \in S\rbrace, \lbrace e^{i,j}_r \mid 1 \leq i < j \leq k, r \in \lbrace 1, \hdots, |E^{i,j}|\rbrace, (e^{i,j}_r, e^{i,j}_{r+1}) \in S\rbrace)$.
  We claim that $K$ is a solution to $(G, k)$ i.e. a multicolored clique in $G$.
  As~$K$ contains exactly one element of every $V^i$ and $E^{i,j}$, we only have to check that the edges really connect to the vertices of $K$.
	
  Suppose, for sake of contradiction, that this is not the case, i.e. there is a vertex $v^i \in V^i \cap V(K)$ and an edge $e^{i,j} \in E^{i,j} \cap E(K)$ for some $1 \leq i,j \leq k, i \neq j$, such that $e^{i,j} = \lbrace u^i, w^j\rbrace$ with $u^i \neq v^i$.
  For $e^{i,j}$ we have deleted only the outgoing arc in gadget $H^{i,j}$, so there is an $s^{i,j} \to e^{i,j}$-path $P$ of length $|E^{i,j}|$ that uses one arc of weight $0$ and $(|E^{i,j}| -1)$ many arcs of weight $-n$.
  So, $w(P) = -n(|E^{i,j}| -1)$.
	
  Likewise, we deleted only the incoming arc of $v^i$ in the gadget $H^i$.
  Therefore, there is an $v^i \to t^i$-path $Q$ of length~$|V^i|$ that uses one arc of weight $0$ and $(|V^i| - 1)$ many arcs of weight $-n$.
  Thus, $w(Q) = -n(|V^i| - 1)$.
  We link $P$ and $Q$ together with the arc $(t^i, s^{i,j})$ of weight $n (|V^i| + |E^{i,j}| - 1)$ to an overall $v_i \to e^{i,j}$-path $R = Q \circ (t^i, s^{i,j}) \circ P$ of weight $n (|V^i| + |E^{i,j}| - 1) - n(|E^{i,j}| -1) - n(|V^i| - 1) = n$.
	
  Let us now consider the unique $e^{i,j} \to v^i$-paths
  \begin{equation*}
    W^+ = \lbrace (e^{i,j}, p^i), (p^i, v^i)\rbrace \text{ and } W^- = \lbrace (e^{i,j}, n^i), (n^i, v^i)\rbrace
  \end{equation*}
  in the consistency gadget $C_i$.
  These are using only the internal vertices $p^i$ or $n^i$, respectively.
  Their weights are
  \begin{align*}
    w(W^+) &= -(n -\varphi_i(v^i)) -\varphi_i(e^{ij}) = (\varphi_i(v^i) - \varphi_i(u^i)) - n \qquad \text{ and}\\
    w(W^-) &= -(n -\varphi_i(e^{ij})) - \varphi_i(v^i) = (\varphi_i(u^i) - \varphi_i(v^i)) - n \enspace .
  \end{align*}
  If we join these two paths with $R$, we obtain two cycles with weights
  \begin{align*}
    w(W^+ \circ R) &= (\varphi_i(v^i) - \varphi_i(u^i)) - n + n = \varphi_i(v^i) - \varphi_i(u^i) \qquad \text{ and}\\
    w(W^- \circ R) &= (\varphi_i(u^i) - \varphi_i(v^i)) - n + n = \varphi_i(u^i) - \varphi_i(v^i) \enspace .
  \end{align*}
  As $v^i \neq u^i$ and $\varphi_i$ is bijective, we have that either $w(W^+ \circ R)$ or $w(W^- \circ R)$ is negative.
  Thus, $R \circ W^+$ or $R \circ W^-$ is a negative cycle, a contradiction to $H - S$ containing no negative cycles.
  Hence, $K$ must be consistent and a clique in $G$.

  \smallskip
  \noindent
  \textbf{Multicolored Clique to Negative DFAS.}
  Now suppose that $K$ is a multicolored clique on $k$ vertices in $G$, i.e. a solution to $(G, k)$.
  We choose $S$ to be the set $\lbrace (v^i_{r_i - 1}, v^i_{r_i}) \mid i \in \lbrace 1, \hdots, k\rbrace, v^i_{r_i} \in V^i \cap V(K)\rbrace \cup \lbrace (e^{i,j}_{r_{i,j}}, e^{i,j}_{r_{i,j}+1}) \mid 1 \leq i < j \leq k, e^{i,j}_{r_{i,j}} \in E^{i,j} \cap E(K)\rbrace$.
  Note that now every gadget $H^i$ and $H^{i,j}$ is acyclic on its own.
  Furthermore, the whole graph $H$ without our set $S$ and the arcs $(t^i, s^{i,j})$ is acyclic.
  \begin{claim}
    Every $s^{i,j} \to t^{i'}$-path in $H - S$ has weight at least $-n (|V^{i'}| + |E^{i,j}| - 1)$.
  \end{claim}
  \begin{claimproof}
    Suppose, for sake of contradiction, that there is an $s^{i,j} \to t^{i'}$-path in $H - S$ of weight lower than $-n (|V^{i'}| + |E^{i,j}| - 1)$.
    Let $P$ the shortest such path with respect to the number of arcs.
    Then $P$ does not use any arc of type $(t^p, s^{p,q})$, as they have weight $n (|V^p| + |E^{p,q}| - 1)$ and therefore either $w(P[s^{i,j}, t^p]) < -n (|V^p| + |E^{i,j}| - 1)$ or $P[s^{p,q}, t^{i'}] < -n (|V^{i'}| + |E^{p,q}| - 1)$ and those have fewer arcs.
    But as $P$ does not use any arc of type $(t^p, s^{p,q})$ it can only go from the gadget $H^{i,j}$ to $H^{i'}$ once with the help of $C^{i'}$ and does not use other gadgets.
    Moreover, we have that $i' \in \lbrace i, j\rbrace$.
    We can assume, without loss of generality, that $i' = i$.
    Then $P$ leaves~$H^{i,j}$ exactly once and enters $H^i$ exactly once.
    Let $e^{i,j}_a$ be the last vertex of $H^{i,j}$ and $v^i_b$ the first vertex of $H^i$ our path $P$ visits.
    Note that $w(P[s^{i,j}, e^{i,j}_a])$ and $w(P[v^i_b, t^i])$ are always multiples of $n$ and have weight at least $-n(|E^{i,j}| - 1)$ and $-n(|V^{i}| - 1)$, respectively.
    Also, this weight is only attained if and only if $a = r_{i,j}$ or $b = r_i$, respectively.
		
    In any case, $P[e^{i,j}_a, v^i_b]$ consists of three vertices with $n^i$ or $p^i$ being the middle one.
    By choice of weights in $C^i$, the lowest weight $P[e^{i,j}_a, v^i_b]$ can achieve is $1 - |V^i| - n > -2n$.
    Also, if $a = r_{i,j}$ and $b = r_i$ the weight is $-n$.
    So in any case the composition $P = P[s^{i,j}, e^{i,j}_a] \circ P[e^{i,j}_a, v^i_b] \circ P[v^i_b, t^i]$ has weight at least $-n (|V^{i}| + |E^{i,j}| - 1)$---a contradiction to the choice of $P$.
    This completes the proof of the claim.
  \end{claimproof}
  As every cycle in $H - S$ has to use an arc $(t^i, s^{i,j})$ of weight \mbox{$n (|V^{i}| + |E^{i,j}| - 1)$} and every path completing this arc to a cycle has weight at least $-n (|V^{i'}| + |E^{i,j}| - 1)$, the graph $H - S$ contains no negative cycle.
	
  For the size bound note that $S$ contains exactly one arc from every $H^i$ and $H^{i,j}$.
  Thus, it has size $k + \binom{k}{2} = d$.
  So $S$ is a solution to $(H, w, d)$.

  \smallskip
  \noindent
  \textbf{Bounding the pathwidth of $H$.}
  The pathwidth of $H$ is defined as the pathwidth of the underlying undirected graph~$\bar{H}$.
  We will show that \mbox{$\pw(\bar{H}) \in \mathcal{O}(k^2)$}.
	
  Note that, if for some set $U \subseteq V(\bar{H})$ we get a path decomposition of $\bar{H} - U$ with bag size at most $b$, we get a path decomposition for $\bar{H}$ with bag sizes at most $b + |U|$ by adding $U$ to every bag.
  Thus, $\pw(\bar{H}) \leq \pw(\bar{H} - U) + |U|$.
  Let
  \begin{equation*}
    U = \bigcup_{i=1}^k \lbrace t^i, n^i, p^i, v^i_{|V^i|} \rbrace \cup \bigcup_{1 \leq i < j \leq k} \lbrace s^{i,j}, e^{i,j}_{|E^{i,j}|} \rbrace.
  \end{equation*}
  This set $U$ has size $4k + k (k -1) = k^2 + 3k$.
  Observe that the graph $\bar{H} - U$ consist exactly of paths of the following form.
  Each vertex gadget $H^i$ turned into a path consisting of the edges $\lbrace v^i_\ell, v^i_{\ell +1}\rbrace$ for $\ell \in \lbrace 1, \hdots |V^i| -2\rbrace$.
  Meanwhile, each edge gadget $H^{i,j}$ turned into a path consisting of the edges $\lbrace e^{i,j}_\ell, v^{i,j}_{\ell +1}\rbrace$ for $\ell \in \lbrace 1, \hdots |E^{i,j}| -2\rbrace$.
	
  Furthermore, note that if we have a path decomposition for every connected component, then we can concatenate them to one path decomposition.
  Thus, the pathwidth of $\bar{H} - U$ is the maximum pathwidth of its connected components.
  But these are all paths and have pathwidth one.
  Thus, $\pw(\bar{H}) \leq \pw(\bar{H} - U) + |U| = k^2 + 3k + 1 \in \mathcal{O}(k^2)$.
	
  Overall we have proven that our reduction produces an instance whose pathwidth and deletion size are bounded by some function in the parameter of the original \MulticoloredClique{} instance.
  Moreover, by choice of $\varphi$, the only arcs of positive weight are the arcs of type $(t^i, s^{i,j})$ of which only $k(k-1)$ many exist.
  So we have a parameterized reduction from \MulticoloredClique{} with parameter $k$ to \NDFAS{} with parameters pathwidth, deletion size and number of positive arcs.
	
  To see that the weights have the claimed form, first note that all of them are integral.
  Moreover, the highest absolute weight have the arcs $(t^i, s^{i,j})$, which have weight $n (|V^i| + |E^{i,j}| - 1)$.
  As $|V(H)| \geq |V^i| + |E^{i,j}|$ and $|V(H)| \geq n$, we have that indeed $w: A(H) \to \mathbb{Z} \cap [-|V(H)|^2, |V(H)|^2]$, proving the theorem.
\end{proof}

From this theorem we  are also able to infer \Whard{1}ness for \NDFAS{} parameterized in $\pw(G) + k$ on instances with weights restricted to $w: A(G) \to \lbrace -1, 0, 1\rbrace$.

\begin{theorem}
\label{thm:W1_hardness_for_pathwidth_deletion_size_and_normalized_weights}
  {\sc Negative DFAS} is \Whard{1} when parameterized in the pathwidth and deletion size, even on instances with weights of the form $w: A(G) \to \lbrace -1, 0, 1\rbrace$.
\end{theorem}
\begin{proof}
  We proof the theorem by showing how an instance described by \Cref{thm:W1_hardness_for_pathwidth_and_deletion_size} can be turned into one with its pathwidth increased by two and weights of the form $w: A(G) \to \lbrace -1, 0, 1\rbrace$ while keeping the deletion size the same.
  In the process, it loses the property of having only a bounded number of positive arcs.
  Let $(G, w, k)$ be a \NDFAS{} instance as described in \Cref{thm:W1_hardness_for_pathwidth_and_deletion_size}.
  Especially, we have $w: A(G) \to \mathbb{Z} \cap [-|V(G)|^2, |V(G)|^2]$.
  From this instance we obtain an instance $(G', w', k)$ by replacing every non-zero arc $a \in A(G)$ by a path $P_a$ of $|w(a)|$-many arcs all having weight $1$ if $w(a) > 0$ or $-1$ if $w(a) < 0$.
  We denote by $(G', w')$ the resulting graph with arc weights.
  It remains to show that $(G, w, k)$ and $(G', w', k')$ are solution equivalent and $\pw(G') \leq \pw(G) + 2$.
	
  To see the solution equivalence, note that any cycle in $G'$ uses a newly created path $P_a$ either as a whole or not at all.
  Thus, deleting an arbitrary arc of a path $P_a$ is equivalent to deleting the whole path and this is equivalent to deleting the arc $a$ in $G$.
  For the pathwidth, note that replacing arcs by paths can be incorporated into the path-decomposition as follows.
  Let $B_i$ the bag that contains the arc $a$ in the path decomposition of $G$ that defines $\pw(G)$.
  Let $x_1, \hdots, x_{|w(a)| - 1}$ the internal vertices of~$P_a$.
  Then we modify the path decomposition of $G$ by replacing $B_i$ with the sequence $B_i, B_i \cup \lbrace x_1, x_2 \rbrace, \hdots, B_i \cup \lbrace x_{|w(a)| - 2}, x_{|w(a)| - 1} \rbrace, B_i$.
  This covers all newly created arcs for~$P_a$.
  By doing this replacement for all arcs $a$ (while using only $B_i$ of the original decomposition), this yields a new path decomposition for~$G'$, where the maximum bag size increased by at most two.
  This shows that  $\pw(G') \leq \pw(G) + 2$, concluding the proof of the theorem.
\end{proof}

\subsection{\texorpdfstring{\Whard{1}ness}{W[1]-hardness} for Pathwidth, Deletion Size and Few Negative Arcs}
In this section we are going to prove \Whard{1}ness for \textsc{MinFB} parameterized in $w_-$, $k$ and $\pw(G)$ for general arc weights.
We will do this by doing an intermediate step showing \Whard{1}ness for the \textsc{Bounded Edge Directed $(s,t)$-Cut} problem in DAGs when parameterized in $k$ and $\pw(G)$.

\problemdef{Bounded Edge Directed $(s,t)$-Cut}
  {A graph $G$, vertices $s, t \in V(G)$ and two integers $k, \ell \in \mathbb{Z}_{\geq 0}$.}
  {Find a set $S \subseteq V(G)$ of size at most $k$ such that every $s \to t$-path of $G - S$ has length more than~$\ell$ or decide that no such set exists.}

For the undirected \textsc{Bounded Edge $(s,t)$-Cut} problem, Bentert et al. showed \Whard{1}ness for the parameters maximum vertex degree and pathwidth \cite[Theorem~1]{BentertHK2021}.
By inspecting their reduction closely, one easily sees that their proof actually shows hardness for parameters maximum vertex degree, pathwidth and deletion size.
The hardness for the directed case follows from this by replacing every edge by a forward and backward arc.
For completeness (and bound on the pathwidth) we choose to do a direct reduction to the directed version here.
It has the nice property that it produces acyclic digraphs, a property we cannot achieve by replacing arcs with a \emph{cycle} of length two.
However, note that all important ideas are already present in the paper by Bentert et al. \cite{BentertHK2021}.

\begin{lemma}
\label{lem:BEDC_is_W1_hard_parameterized_in_pathwidth_and_deletion_size}
  The {\sc Bounded Edge Directed $(s,t)$-Cut} problem parameterized in the pathwidth~$\pw(G)$, deletion size~$d$ and maximum degree of the underlying undirected graph $\Delta(G)$ is \mbox{\Whard{1}}, even when the graph is restricted to be acyclic.
\end{lemma}
\begin{proof}
  Let now $(G, k)$ be a \textsc{Clique} instance.
  Recall that \textsc{Clique} is \mbox{\Whard{1}} when parameterized by $k$.
  We want to construct an equivalent instance $(H,d,\Lambda)$ of \textsc{Bounded Edge Directed $(s,t)$-Cut} whose parameters $\pw(H)$, $d$ and $\Delta(G)$ are bounded by a function of $k$ only.
  For notation, we fix some arbitrary ordering $v_1, \hdots, v_n$ of the vertices $V(G)$.
  For the edges $e \in E(G)$ we define their canonical representation $\lbrace v_i, v_j\rbrace$ to be the order of vertices with $i \leq j$.
  Setting this canonical representation then defines an ordering $e_1 \prec \hdots \prec e_m$ of the edges $E(G)$ by $e = \lbrace v_i, v_j\rbrace \preceq e' = \lbrace v_{i'}, v_{j'}\rbrace$ iff $(i,j) \preceq_ \textsf{lex} (i', j')$, where $\preceq_ \textsf{lex}$ is the lexicographical ordering on vectors of length two.
	
  As usual in reduction from cliques we construct two kinds of gadgets, vertex gadgets~$H^i$ and edge gadgets $H^{i,j}$.
  The vertex gadgets each choose a vertex in the graph and the edge gadgets check whether there is indeed an edge between them.
  For all these constructions we will use paths of different lengths.
  For ease of notation we will replace these paths by arcs with positive integral weights.
  We make sure that these weights are bounded by some function of $n + m$ and that parallel arcs have weight at least two. In this way, replacing an arc of weight $w$ by a directed path of length $w$ gives us a still polynomial sized simple graph for \textsc{Bounded Edge Directed $(s,t)$-Cut}.
  Also, deletion of an arc of a path whose internal vertices have in-degree and out-degree one and are disjoint from $s$ and $t$, means that no arc on this path is part of an $s \to t$-path anymore.
  So a solution does not delete more paths by choosing several arcs on the same replaced path.
  This preserves solution equivalence, between the weighted and the replaced graph.
	
  An important role for the weights plays an integer $M$ which we will choose later.
  Our choice of $\Lambda$ will depend on $M$, so we also defer its choice until later.
  For know think of $M$ as a big integer.

  \smallskip
  \noindent
  \textbf{Gadgets.}
  Before introducing the different vertex and edge gadgets, we discuss some gadget that will be the foundation for both.
  For any positive integers $a,b \in \mathbb{Z}_{>0}$ and any set  $\lbrace \varphi^q: \lbrace 1,\hdots, b\rbrace \to \mathbb{Z} \cap [-M, M] \mid q\in\lbrace1, \hdots, a\rbrace\rbrace$ of functions, we introduce the gadget $R^{a,b}$.
  The gadget $R^{a,b}$ consists of $a$ many disjoint paths $P^1, \hdots, P^a$ of length $2b + 2$.
  For each path $P^q$ with $q \in \lbrace 1,\hdots,a\rbrace$ all arcs of $P^q$ have weight $1$, and we denote the vertices of $P^q$ in order of there appearance by $p^q_\textsf{in}, p^q_1, \hat{p}^q_1, p^q_2, \hat{p}^q_2 \hdots, p^q_{b+1}, p^q_\textsf{out}$.
  We add the following ``detours'' along the paths $P^q$.
  For every $r \in \lbrace 1, \hdots, b\rbrace$ there is an arc $(p^q_r, p^q_{r+1})$ of weight $5M + \varphi^q(r)$.
  Moreover, we interlink the paths $P^q$ and $P^{q+1}$ for every $q \in \lbrace 1, \hdots, a-1\rbrace$ by the following arcs.
  For every $r \in \lbrace 1, \hdots, b\rbrace$ there are arcs $(\hat{p}^q_r, p^{q+1}_{r + 1})$ and $(\hat{p}^{q+1}_r, p^{q}_{r + 1})$ of weight $3M$ each.
  Finally, the gadget~$R^{a,b}$ contains the vertices $s$ and $t$.
  To connect those vertices to the rest of the graph, for every $q \in \lbrace 1,\hdots,a\rbrace$, there are two parallel arcs $(s, p^q_1)$ of weight $2$ and two parallel arcs $(p^q_{b+1}, t)$ of weight $\Lambda - 4M$.
  This concludes the construction of our gadget.
  Note that for $M \geq 1$ the graph has positive arc weights and for $\Lambda \geq 4M + 2$ all parallel arcs are of weight at least~$2$.
	
  We now prove the crucial properties of this gadget:
  \begin{claim}
  \label{claim:BEDC_meta_gadget}
    Let $M \geq  2b + 2$ and $\Lambda \geq 4M + 2$ be integers.
    For every set $S \subseteq A(R^{a,b})$ such that $R^{a,b} - S$ contains no $s \to t$-path of weight at most $\Lambda$, we have that $|S| \geq a$.
		
    Moreover, for every such set $S \subseteq A(R^{a,b})$ of size $|S| = a$, there exists a unique \mbox{$x \in \lbrace 1, \hdots, b\rbrace$} such that for every $q \in \lbrace 1, \hdots, a\rbrace$ any minimum-weight $s \to p^q_\textsf{out}$-path has weight $5M + 2b + 1 + \varphi^q(x)$ and any minimum-weight $p^q_\textsf{in} \to t$-path has weight $\Lambda + M + 2b - 1  + \varphi^q(x)$.
		
    Also, for every $x \in \lbrace 1, \hdots, b\rbrace$, we can choose a set $S \subseteq A(R^{a,b})$ of size $a$ such that $R^{a,b} - S$ contains no $s \to t$-path of weight $\leq \Lambda$, any minimum-weight $s \to p^q_\textsf{out}$-path has weight $5M + 2b + 1  + \varphi^q(x)$ and any minimum-weight $p^q_\textsf{in} \to t$-path has weight $\Lambda + M + 2b - 1  + \varphi^q(x)$.
  \end{claim}
  \pagebreak[1]
  \begin{claimproof}
    Consider for every $q \in \lbrace 1, \hdots, a\rbrace$ the path $s, p^q_1, \hat{p}^q_1, \hdots, p^q_{b+1}, t$.
    This path has weight $2 + 2b + \Lambda - 4M \leq \Lambda$.
    Thus, any set $S$ as in the theorem statement has to contain an arc of all of these $s \to t$-paths.
    As these paths are arc-disjoint, we have that all sets have size at least $a$.

    Now we turn to the sets $S$ of size exactly $a$.
    Consider the paths from above again.
    As there are two parallel arcs for each of $(s, p^q_1)$ and $(p^q_{b+1},t)$, our set $S$ has to intersect each $P^q$ with exactly one edge that is not $(p^q_\textsf{in}, p^q_1)$ or $(p^q_{b+1}, p^q_\textsf{out})$.
		
    We claim that there is an $x \in \lbrace 1, \hdots, b\rbrace$ such that $S = \lbrace (p^q_x, \hat{p}^q_x)\mid q \in \lbrace 1, \hdots, a\rbrace\rbrace$.
    Suppose that not, then for some pair $(p,q)$ with $q \in \lbrace 1, \hdots, a-1\rbrace$ and $y \in \lbrace 1, \hdots, n\rbrace$ we have that either $P^q[p^q_1, \hat{p}^q_y]$ and $P^{q+1}[p^{q+1}_{y+1}, p^{q+1}_{b+1}]$ are disjoint from~$S$ or $P^{q+1}[p^{q+1}_1, \hat{p}^{q+1}_y]$ and $P^q[p^q_{y+1}, p^q_{b+1}]$ are disjoint from~$S$.
    In the former case, the path
    \begin{equation*}
      (s,p^q_1) \circ P^q[p^q_1, \hat{p}^q_y] \circ (\hat{p}^q_y, p^{q+1}_{y+1}) \circ P^{q+1}[p^{q+1}_{y+1}, p^{q+1}_{b+1}] \circ (p^{q+1}_{b+1},t)
    \end{equation*}
    is an $s \to t$-path in $R^{a,b} - S$. 
    In the latter case, the path
	\begin{equation*}
      (s,p^{q+1}_1) \circ P^{q+1}[p^{q+1}_1, \hat{p}^{q+1}_y] \circ (\hat{p}^{q+1}_y, p^q_{y+1}) \circ P^q[p^q_{y+1}, p^q_{b+1}] \circ (p^q_{b+1},t)
    \end{equation*}
    is an $s \to t$-path in $R^{a,b} - S$.
    So one of these $s \to t$-paths exists in $R^{a,b} - S$.
    Note that both paths have weight
    \begin{equation*}
      2 + (2(y- 1) + 1) + 3M + 2(b - y) + (\Lambda - 4M) = 2b + 1 + 3M + \Lambda - 4M \leq \Lambda,
    \end{equation*}
    which is a contradiction to the choice of $S$.
		
    So, if there is a set $S$ such that $R^{a,b} - S$ contains no $s \to t$-path of weight $\leq \Lambda$ that has size $a$, it has the form $\lbrace (p^q_x, \hat{p}^q_x)\mid q \in \lbrace 1, \hdots, a\rbrace\rbrace$.
    Indeed, for such a set $S$, we have that any $s \to t$-path in $R^{a,b} - S$ has to use one of the detour arcs $(p^q_x, p^q_{x+1})$ of weight $5M + \varphi^q(x) \geq 4M$.
    Moreover, any $s \to t$-path has to use an arc incident to $s$ and~$t$ of weight $2$ and $\Lambda - 4M$ respectively.
    So for any $x\in \lbrace 1, \hdots, b\rbrace$, each $s \to t$-path in $R^{a,b} -\lbrace (p^q_x, \hat{p}^q_x)\mid q \in \lbrace 1, \hdots, a\rbrace\rbrace$ has overall weight at least $2 + 4M + \Lambda - 4M > \Lambda$.
		
    It only remains to prove that, for solutions $S = \lbrace (p^q_x, \hat{p}^q_x)\mid q \in \lbrace 1, \hdots, a\rbrace\rbrace$, any minimum-weight $s \to p^q_\textsf{out}$-path has weight $5M + 2b + 1 + \varphi^q(x)$ and any minimum-weight $p^q_\textsf{in} \to t$-path has weight $\Lambda + M + 2b - 1  + \varphi^q(x)$ in $R^{a,b} - S$.
    Any $s \to p^q_\textsf{out}$-path uses exactly one incident edge of $s$ (which has weight $2$) and $p^q_\textsf{out}$ (which has weight $1$).
    So, for these paths it is enough to prove that any $p^h_1 \to p^q_{b+1}$-path in $R^{a,b} - S$ for $h \in \lbrace 1, \hdots, a\rbrace$ has weight at least $5M + 2(b -1) + \varphi^q(x)$ and one path achieves that weight.
    Moreover, any $p^q_\textsf{in} \to t$-path uses exactly one incident edge of $p^q_\textsf{in}$ (which has weight~$1$) and $t$ (which has weight $\Lambda - 4M$).
    Thus, for these paths it is enough to prove that any $p^q_1 \to p^g_{b+1}$-path in $R^{a,b} - S$ for $g \in \lbrace 1, \hdots, a\rbrace$ has weight at least $5M + 2(b -1) + \varphi^q(x)$ and one path achieves that weight.
    Together it suffices to show that any $p^h_1 \to p^g_{b+1}$-path in $R^{a,b} - S$ has weight at least $5M + 2(b -1) + \varphi^q(x)$ and there is a $p^q_1 \to p^q_{b+1}$-path that achieves this weight.
		
    Notice that any $p^h_1 \to p^g_{b+1}$-path contains a detour arc $(p^r_x, p^r_{x+1})$ of weight $5M + \varphi^h(x) \geq 4M$.
    If a $p^h_1 \to p^g_{b+1}$-path contains another detour arc or an interconnecting arc, it occurs an additional weight of at least $3M$.
    So any $p^h_1 \to p^g_{b+1}$-path that has uses arcs other than those of a single $P^q$ and the arc $(p^q_x, p^q_{x+1})$ has weight at least $8M$.
    But $5M + 2b + 1 + \varphi^q(x) \leq 7M < 8M$.
    So it is enough to show that there is a $p^q_1 \to p^q_{b+1}$-path of weight $5M + 2(b -1) + \varphi^q(x)$.
    Notice that $P^q[p^q_1, p^q_x] \circ (p^q_x, p^q_{x+1}) \circ P^q[p^q_{x+1}, p^q_{b+1}]$ has weight $2(x-1) + 5M + \varphi^q(x) + 2(b-x) = 5M + 2(b -1) + \varphi^q(x)$ and thus is exactly such a path.
    This completes the proof of \Cref{claim:BEDC_meta_gadget}.
  \end{claimproof}
	
  We are now ready to define the vertex gadgets and edge gadgets.
  All vertex gadgets will be identical and all edge gadgets will be identical, just the interconnection between them differs.
  For all $i \in \lbrace 1, \hdots, k\rbrace$, we introduce a vertex gadget $H^i$ by setting $a = 2$, $b = n$, $\varphi^1(x) = x$ and $\varphi^2(x) = -x$.
  Note that the image of $\varphi^1$ and $\varphi^2$ lies in $\mathbb{Z} \cap [-M, M]$ for~$M \geq n$.
  Thus, the graph $R^{2, n}$ is well-defined.
  To obtain our gadget $H^i$ from this we forget about the vertices $p^1_\textsf{in}$ and~$p^2_\textsf{in}$.
  Furthermore, we rename the vertices $p^1_\textsf{out}$ and $p^2_\textsf{out}$ to $v^i_+$ and $v^i_-$.
	
  For all $i,j \in \lbrace 1, \hdots, k\rbrace$ with $i < j$, we introduce an edge gadget $H^{i,j}$.
  For this we set $a=4$, $b = m$ and for every $x \in \lbrace 1, \hdots, m\rbrace$ and edge $e_x$ with canonical representation $(v_g, v_h)$, we set $\varphi^1(x) = -g$, $\varphi^2(x) = g$, $\varphi^3(x) = -h$ and $\varphi^4(x) = h$.
  Again the image of $\varphi^1, \hdots, \varphi^4$ lies in $\mathbb{Z} \cap [-M,M]$ for $M \geq m$.
  Thus, the graph $R^{4, m}$ is well-defined.
  By forgetting the vertices $p^q_\textsf{out}$ and renaming $p^1_\textsf{in}$, $p^2_\textsf{in}$, $p^3_\textsf{in}$ and $p^4_\textsf{in}$ to $v^i_+$, $v^i_-$, $v^j_+$ and $v^j_-$.
	
  From these gadgets we obtain our complete graph $H$ by taking one copy of each gadget and identifying all vertices of the same name ($s$, $t$, $v^i_+$ and $v^i_-$).
  Finally, we set $d = 2n + 2n(n-1)$, $M = 2(n+m) + 2$ and $\Lambda = 10M + 2(n -1) + 2(m -1) - 1$.
  This completes our construction.

  \smallskip
  \noindent
  \textbf{Correctness.}
  We now want to prove that $(G, k)$ as instance of \textsc{Clique} is a ``yes''-instance if and only if $(H, d, \Lambda)$ as instance of \textsc{Bounded Edge Directed $(s,t)$-Cut} is a ``yes''-instance.
	
  For the forward direction let $U \subseteq V(G)$ be a clique in $G$ with corresponding edge set $F$.
  Let $u_1, \hdots, u_k$ be the indices of the vertices of $v_1, \hdots, v_n$ corresponding to the vertices of $U$ in increasing order.
  Also, let $f_{i,j}$ be the index of the edge $\lbrace v_{u_i}, v_{u_j}\rbrace$ in the fixed ordering of $E(G)$.
  By \Cref{claim:BEDC_meta_gadget}, there is a solution $S^i$ of size two to each of the $H^i$'s such that each $s \to v^i_+$-path has weight at least $5M + 2(n-1) + u_i$ and each $s \to v^i_-$-path has weight at least $5M + 2(n-1) - u_i$.
  Moreover, again by \Cref{claim:BEDC_meta_gadget}, there is a solution $S^{i,j}$ of size four to each of the $H^{i,j}$'s such that each $v^i_+ \to t$-path has weight at least $5M + 2(m-1) - u_i$, each $v^i_- \to t$-path has weight at least $5M + 2(m-1) + u_i$, each $v^j_+ \to t$-path has weight at least $5M + 2(m-1) - u_j$ and each $v^j_- \to t$-path has weight at least $5M + 2(m-1) + u_j$.
  We define our solution~$S$ as $\bigcup_{i=1}^k S^i \cup \bigcup_{1 \leq i < j \leq k} S^{i,j}$.
  By definition our $S$ is a solution when restricted to a single vertex gadget or a single edge gadget.
  The only $s \to t$-paths that are not contained in a single gadget consist of an $s \to v^i_+$-path in some vertex gadget and a $v^i_+ \to t$-path in some edge gadget or of an $s \to v^i_-$-path in some vertex gadget and a $v^i_- \to t$-path in some edge gadget.
  In each case, as seen above, these paths have length at least $5M + 2(n-1) + u_i +5M + 2(m-1) - u_i = \Lambda +1$.
  It only remains to check the size of $S$ which is exactly $2k + 4 \frac{k(k-1)}{2} = d$.
  \smallskip
		
  For the backwards direction, let $S$ be a solution to $(H, d, \Lambda)$.
  From \Cref{claim:BEDC_meta_gadget}, we get that any solution to a vertex gadget has size at least two and any solution to an edge gadget has size at least four.
  As a solution must be a solution to every single gadget and our gadgets are arc disjoint, we have that any solution must have size at least $2k + 4 \frac{k(k-1)}{2} = d$.
  From this it follows that our solution must have size exactly two in every vertex gadget and size exactly four in every edge gadget.
  \Cref{claim:BEDC_meta_gadget} states that the structure of these solutions in the vertex gadgets must be such that in $H^i - S$ any minimum weight $s \to v^i_+$-path in $V^i - S$ has weight $5M + 2(n-1) + u_i$ and any minimum weight $s \to v^i_-$-path in $V^i$ has weight $5M + 2(n-1) - u_i$ for some $u_i \in \lbrace 1, \hdots, n\rbrace$.
  We define the vertex set $U \subseteq V(G)$ to be the set $\lbrace v_{u_i} \mid i \in \lbrace 1, \hdots, k\rbrace\rbrace$.
	
  We can apply \Cref{claim:BEDC_meta_gadget} also to the structure of $S$ with respect to $H^{i,j} - S$.
  This tells us that for every $1 \leq i < j \leq k$ there is an index $f_{i,j} \in \lbrace 1, \hdots,m\rbrace$ such that for the canonical representation $\lbrace v_g, v_h\rbrace$ of $e_{f_{i,j}}$, in $H^{i,j} - S$ we have that any minimum weight $v^i_+ \to t$-path has weight $5M + 2(m-1) - g$, any minimum weight $v^i_- \to t$-path has weight $5M + 2(m-1) + g$, any minimum weight $v^j_+ \to t$-path has weight $5M + 2(m-1) - h$ and any minimum weight $v^j_- \to t$-path has weight $5M + 2(m-1) + h$.
  If we concatenate a minimum weight $s \to v^i_+$-path in $V^i - S$ with a minimum weight $v^i_+ \to t$-path in $H^{i,j} - S$, we get an $s \to t$-path of weight $10M + 2(n -1) + 2(m -1) + u_i-g$.
  As $S$ is a solution to $(H, d, \Lambda)$, this path has length more than $\Lambda$, which is equivalent to $u_i + 1 > g$ or $u_i \geq g$.
  By applying the same argument to minimum weight paths that concatenate at $v^i_-$ instead of $v^i_+$, we get the inequality $u_i \leq g$.
  Together this yields $u_i = g$.
  If we swap $H^i$ for $H^j$ and $v^i_+$ and $v^i_-$ for $v^j_+$ and $v^j_-$, we get $u_j = h$ in the same way.
  So the edge $\lbrace v_{u_i}, v_{u_j}\rbrace = \lbrace v_g, v_h\rbrace = e_{f_{i,j}}$ exists in $G$.
  Since this is true for all $(i,j)$ with $1 \leq i < j \leq k$, we have that $U$ is indeed the vertex set of a clique.

  \smallskip
  \noindent
  \textbf{Bound on Parameters and Acyclicity.}
  For pathwidth consider the vertex set $U = \lbrace s, t \rbrace \cup \bigcup_{i=1}^k \lbrace v^i_+, v^i_-\rbrace$.
  This is a set of $2(k+1)$ vertices and will be contained in every bag of our path decomposition.
  Note that all vertices of $H - U$ have no neighbors outside the gadget they are contained in.
  So it suffices to find a path decomposition for every gadget (minus vertices in~$U$), concatenate them and add $U$ to every set.
  The pathwidth of this decomposition is then bounded by $|U|$ plus the maximum pathwidth of any gadget path decomposition.
  We do this path decomposition on the meta-gadget~$R^{a,b}$.
  For every $r \in \lbrace 1, \hdots, b\rbrace$ let $X_{2r-1} = \lbrace p^q_r, \hat{p}^q_r \mid q \in \lbrace 1, \hdots, a\rbrace\rbrace$ and for every $r \in \lbrace 1, \hdots, b-1\rbrace$ let $X_{2r} = \lbrace \hat{p}^q_r, p^q_{r+1} \mid q \in \lbrace 1, \hdots, a\rbrace\rbrace$.
  Then every arc of $R^{a,b} - S$ is contained in one of the bags $X_i$ and each vertex in at most two consecutive bags.
  Thus, $X_1, \hdots, X_{2b-1}$ is a path decomposition.
  The maximum size of a bag is $2a$.
  For our specific gadgets $a$ is at most four.
  Thus, we get that the pathwidth of $H$ is at most $8 + |U| = 2k + 10$.
  The maximum vertex degree inside a gadget is $6$ for the $p^q_r$ vertices.
  However, the vertices $v^i_+$, $v^i_-$, $s$ and $t$ have neighbors in several gadgets.
  Nodes $v^i_+$, $v^i_-$ have degree $k + \frac{k(k-1)}{2}$, $s$ has degree $4k$ and $t$ has degree $16 \frac{k(k-1)}{2}$.
  In any case the maximum is bounded by a function of $k$.
  Last but not least the deletion size is $d = 2k + 4 \frac{k(k-1)}{2}$.
	
  For the acyclicity note that the gadgets $R^{a,b}$ are acyclic (as all arcs only run to increasing positions along the paths $P^q$).
  So it only remains to prove that the boundary vertices $s$, $t$, $v^i_+$ and $v^i_-$ are connected in a way that preserves acyclicity.
  Node $s$ has only outgoing arcs and $t$ has only incoming arcs.
  Moreover, all $v^i_+$ and $v^i_-$ the ingoing arcs come exactly from the vertex gadgets and the outgoing arcs come exactly form the edge gadgets and these gadgets are not connected by any other vertices.
\end{proof}

Now we use \Cref{lem:BEDC_is_W1_hard_parameterized_in_pathwidth_and_deletion_size} to show hardness for \NDFAS{} by weighting all arcs of an \textsc{Bounded Edge Directed $(s,t)$-Cut} instance with weight one and introducing $(k+1)$-many $(t,s)$ arcs of weight $-(\ell + 1)$.

\begin{theorem}
\label{thm:W1_hardness_for_pathwidth_number_negative_arcs_and_deletion_size}
  {\sc Negative DFAS} is \Whard{1} when parameterized in the pathwidth, deletion size and number of negative arcs.
\end{theorem}
\begin{proof}
  We show the theorem by reduction from \textsc{Bounded Edge Directed $(s,t)$-Cut} on DAGs parameterized by pathwidth and deletion size, which is \Whard{1} by \Cref{lem:BEDC_is_W1_hard_parameterized_in_pathwidth_and_deletion_size}.
  Let $(G, s, t, k, \ell)$ be a \textsc{Bounded Edge Directed $(s,t)$-Cut} instance with $G$ being a DAG.
  We construct a solution equivalent instance $(H, w, k)$ of \NDFAS{} as follows.
  To construct~$H$, give all arcs in $G$ weight $1$ and introduce $k+1$ arcs $(t, s)$ of weight $-(\ell + 1)$.
  This completes the construction of $(H, w, k)$.
  We now show the solution equivalence.

  \smallskip
  \noindent
  \textbf{Bounded Edge Directed $(s,t)$-Cut to Negative Directed Feedback Arc Set.}
  Let $S$ be a solution to $(G, s, t, k, \ell)$.
  We claim that $S$ is also a solution to $(H, w, k)$.
  The set~$S$ trivially fulfills the size bound of $k$.
  Suppose, for sake of contradiction, that there is a negative cycle $C$ in $H - S$.
  As the only negative arcs of $H$ are those from $t$ to $s$, $C$ contains exactly one of these.
  So $C[s,t]$ is an $s \to t$-path in $G - S$ and has weight less than $-w((t,s)) = \ell + 1$.
  Now, all weights in $G$ are $1$ and thus $C[s,t]$ is an $s \to t$-path in $G - S$ of length at most $\ell$, a contradiction to $S$ being a solution to $(G, s, t, k, \ell)$.
  Thus,~$S$ is a solution to $(H, w, k)$.

  \smallskip
  \noindent
  \textbf{Negative DFAS to Bounded Edge Directed $(s,t)$-Cut.}
  Let $S$ be an inclusion-wise minimal solution to $(H, w, k)$.
  As there are $(k+1)$-many parallel arcs $(t, s)$, one of them does not lie in~$S$.
  Any cycle in $H$ can use at most one of these arcs and thus as $S$ is inclusion-wise minimal, $S$ contains non of them, i.e. it contains only arcs of $G$.
  We claim that $S$ is a solution to $(G, s, t, k, \ell)$.
  It trivially fulfills the size bound.
  Now assume for contradiction that $G - S$ contains an $s \to t$-path $P$ of length at most $\ell$.
  Then $P \circ (t, s)$ is a cycle of weight $w(P) - (\ell + 1) < 0$ in $H - S$, which is a contradiction to $S$ being a solution to $(H, w, k)$.

  \smallskip
  \noindent
  \textbf{Bounding the parameters.}
  The deletion size of our new instance stays the same and thus is bounded by the deletion size of the old instance.
  The same holds for the number of negative arcs, as we only introduced $k+1$ of them.
  Finally, note that we can turn any pathwidth decomposition of $G$ into one of $H$ by adding $\lbrace s,t \rbrace$ to all bags.
  This shows $\pw(H) \leq \pw(G) + 2$.
\end{proof}

\subsection{No Admission of a Polynomial Compression}\label{sec:noPolynomialCompression}
In this section, we exhibit that there is no polynomial compression for {\sc MinFB} parameterized by~$k + w_-$. We give a reduction from the {\sc Bounded Edge Directed $(s,t)$-Cut} problem.
Recall from the previous section that this problem takes as input a digraph~$G$, vertices $s,t$, and integers $k,\ell\in\mathbb N$, and asks for a set $X \subseteq E(G)$ of size at most $k$ such that $G-X$ contains no $s$-$t$-paths of length at most $\ell$. 
Fluschnik et al.~\cite{FluschnikEtAl2016} showed that {\sc Bounded Edge Directed $(s,t)$-Cut} does not admit a kernel of size polynomial in~$k$ and $\ell$, assuming $\mathsf{NP}\not\subseteq \mathsf{coNP}/poly$, even for acyclic input digraphs.
In fact, their construction allows for a stronger result, ruling out a polynomial compression.

\begin{proposition}[Fluschnik et al.~\cite{FluschnikEtAl2016}]
\label{thm:BEDChasnoPolComp}
  Assuming $\mathsf{NP}\not\subseteq \mathsf{coNP}/poly$, {\sc Bounded Edge Directed $(s,t)$-Cut} does not admit a polynomial compression in $k + \ell$ even when $G$ is a DAG.
\end{proposition}

To get to the {\sc MinFB} problem, we first consider as an intermediate step the {\sc Directed Small Cycle Transversal}  problem.

\problemdef{Directed Small Cycle Transversal}
  {A directed graph $G$ and integers $k,\ell$.}
  {Find a set $X \subseteq E(G)$ of size at most $k$ such that $G-X$ contains no cycles of length at most $\ell$.}

Fluschnik et al.~\cite{FluschnikEtAl2016} showed that {\sc Directed Small Cycle Transversal} does not admit a kernel of size polynomial in~$k$ and $\ell$, unless $\mathsf{NP}\subseteq \mathsf{coNP}/poly$.
Again their result can be strengthened to not admitting a polynomial compression.
They further observed that {\sc Directed Small Cycle Tranversal} admits a simple branching algorithm that runs in time $\mathcal{O}(\ell^k\cdot n\cdot (n+m))$.
Here we argue that a dependence on both parameters $k$ and $\ell$ is necessary for fixed-parameter tractability:

\begin{lemma}
\label{thm:PPTtoDSCT}
  There is a polynomial parameter transformation from {\sc Bounded Edge Directed $(s,t)$-Cut} in DAGs with parameter $k$ (parameter $\ell$, parameter $k + \ell$) to {\sc Directed Small Cycle Transversal} with parameter $k$ (resp. $\ell$, resp. $k + \ell$).
	
  This even holds if {\sc Directed Small Cycle Transversal} is restricted to instances where there are vertices $s, t \in V(G)$ such that all cycles of $G$ consist of an $s$-$t$-path and an arc of the form $(t,s)$.
\end{lemma}
\begin{proof}
  Let $(G, s, t, k, \ell)$ be a {\sc Bounded Edge Directed $(s,t)$-Cut} instance where $G$ is a DAG.
  As $G$ is a DAG, it admits a topological ordering $v_1, \hdots, v_{|V(G)|}$ of its vertices, so that there are no arcs $(v_i, v_j)$ for $j < i$.
  Without loss of generality, let $v_1 = s$ and $v_{|V(G)|} = t$, as vertices before~$s$ or after $t$ in a topological ordering are never part of any $s$-$t$-path.
	
  We now create a directed graph $G'$ from $G$ by adding $k +1$ parallel arcs $a_1,\dots,a_{k+1}$ from~$t$ to $s$.
  Then every cycle in~$G'$ consists of an $s$-$t$-path and an arc $a_i$, as $G$ was acyclic.
  Set $\ell' = \ell + 1$.
  Then $(G', k, \ell')$ is an instance of {\sc Directed Small Cycle Transversal}.
  The above transformation can be done in polynomial time and the parameter increases by at most one (depending on whether $\ell$ is part of the parameter).
	
  It remains to show that $(G', k, \ell')$ is a ``yes''-instance of {\sc Directed Small Cycle Transversal} if and only if $(G, s, t, k, \ell)$ is a ``yes''-instance of {\sc Bounded Edge Directed $(s,t)$-Cut}.
	
  For the forward direction, let $X'$ be a solution to $(G', k, \ell')$.
  Consider $X = X' \setminus \{a_1,\dots,a_{k+1}\}$.
  For sake of contradiction, suppose that $G - X$ contains an $s$-$t$-path $P$ of length at most $\ell$.
  As $|X'| \leq k$, it can not contain all arcs $a_i$.
  Without loss of generality, $a_1 \not \in X'$.
  Then $P$ followed by $a_1$ is a cycle in $G' - X'$ of length at most $\ell +1 = \ell'$---a contradiction to $X'$ being a solution to $(G', k, \ell')$.
	
  For the reverse direction, let $X$ be a solution to $(G, s, t, k, \ell)$.
  Then $X$ is also a solution to $(G', k, \ell')$ by the following argument.
  Suppose, for sake of contradiction, that $G' - X$ contains a cycle $C$ of length at most $\ell'$.
  By the structure of $G'$, $C$ consists of an $s$-$t$-path $P$ in $G - X$ and an arc~$a_i$.
  Then $|P| = |C| -1 \leq \ell$, contradicting that $X$ is a solution to~$(G, k, \ell)$.
\end{proof}

\begin{lemma}
\label{thm:PPTtoNDFAS}
  There is a PPT from {\sc Directed Small Cycle Transversal} on instances where every cycle uses an arc of type $(t,s)$ when parameterized by $k$ (by $k + \ell$) to {\sc Negative DFAS} parameterized by $k$ (resp. $k + w_-$, where $w_-$ is the number of arcs with negative weight); this even holds in the case where $w: A(G) \rightarrow \{\pm 1\}$.
\end{lemma}
\begin{proof}
  We start with a {\sc Directed Small Cycle Transversal} instance $(G, k, \ell)$ as described in the lemma.
  Let $a_1,\dots,a_p$ be the arcs of the form $(t, s)$.
  For any $p \geq k+1$ there is always an arc which survives the deletion of some arc set of at most $k$ elements.
  So we can assume $p \leq k + 1$ as deleting superfluous arcs does not change the solution.
  Set $A_{+1} = A(G) \setminus \{a_1,\dots,a_p\}$.
  Now replace the $a_i$ by mutually disjoint (except for $s$ and $t$) paths $P_i$ of length $\ell$.
  Call the resulting directed graph~$G'$ and let $A_{-1} = \cup_{i=1}^pA(P_i)$.
  Finally, define
  $w(a) = 1$ for $a\in A_{+1}$ and $w(a) = -1$ for $a\in A_{-1}$.
  As $A(G') = A_{-1} \uplus A_{+1}$, the function $w:A(G') \rightarrow \{-1,+1\}$ is well defined.
	
  The instance $(G', w, k)$ has the required form.
  Also the transformation can be made in polynomial time.
  As $k$ remains unchanged and $w_- = |A_{-1}| = \ell \cdot p \leq  \ell \cdot (k + 1)$ is bounded by a polynomial in $k + \ell$, the parameter restrictions of PPTs are fulfilled.
  It remains to prove that $(G', w, k)$ is a ``yes''-instance of {\sc Negative DFAS} if and only if $(G, k, \ell)$ is a ``yes''-instance of {\sc Directed Small Cycle Transversal}.
	
  For the forward direction, let $X'$ be a solution of $(G', w, k)$.
  Let $X$ be the set where every arc of~$X'$ which is part of some $P_i$ is replaced by $a_i$ (and duplicates are removed).
  Clearly, $|X| \leq |X'| \leq k$.
  Suppose there is a cycle $C$ of length at most $\ell$ in $G - X$.
  Then $C$ contains a unique arc $a_j$.
  Let $C'$ be the cycle in $G'$ resulting from the replacement of $a_j$ by~$P_j$ in~$C$.
  Then~$C'$ is also in $G' - X'$ by choice of $X$ and contains at most $\ell -1$ arcs in $A_{+1}$ and~$\ell$ arcs in~$A_{-1}$.
  This yields the contradiction $w(C') = |A(C') \cap A_{+1}| - |A(C') \cap A_{-1}| \leq \ell -1 - \ell = -1$.
	
  For the reverse direction, let $X$ be a solution of $(G, k, \ell)$.
  Let $X'$ the set $X$ where every arc~$a_i$ is replaced by the first arc of~$P_i$.
  By definition of $G'$, $X' \subseteq A(G')$ and $|X'| = |X| \leq k$ holds.
  Now suppose there is a cycle $C'$ in $G' - X'$ with $w(C') < 0$.
  As the paths $P_i$ are mutually disjoint (except the end vertices) every inner vertex of each $P_i$ has in-degree and out-degree one.
  Thus, if there is an arc of some~$P_i$ inside $C'$ the whole path $P_i$ is.
  Replace each such $P_i$ by $a_i$ to obtain a cycle $C$.
  By construction of~$G'$ and~$X'$, this cycle $C$ is in $G - X$.
  Each cycle in $G$ and therefore also $C$ contains exactly one arc of type~$a_i$.
  Therefore,~$C'$ contains exactly one path $P_i$ and from $0 > w(C') = |A(C') \cap A_{+1}| - |A(C') \cap A_{-1}| = |A(C') \cap A_{+1}|  - \ell$ we get that $|A(C') \cap A_{+1}| < \ell$.
  As $|A(C') \cap A_{+1}|$ is integral we can sharpen the bound to $|A(C') \cap A_{+1}| \leq \ell - 1$.
  By $A(C) = \left(A(C') \cap A_{+1}\right) \uplus \{a_i\}$, we get that $|A(C)| = |A(C') \cap A_{+1}|  + 1 \leq \ell$---a contradiction.
\end{proof}

Concatenating all reductions above and combing with \Cref{thm:NDFASisequivalenttospecialMIS}, we obtain the following corollary:
\begin{corollary}
\label{thm:PPTfromBEDCtoMIS}
  There is a PPT from {\sc Bounded Edge Directed $(s,t)$-Cut} parameterized by $k$ (in $k + \ell$) to {\sc MinFB} parameterized by $k$ (in $k + w_-$).
  This even holds for instances of {\sc MinFB} where $A$ is a system of difference constraints and $w\in\{\pm 1\}^m$.
\end{corollary}

This corollary yields the desired result.
\begin{theorem}
\label{thm:nopolycompressionkb-}
  Assuming $\mathsf{NP}\not\subseteq \mathsf{coNP}/poly$, {\sc MinFB} does not admit a polynomial compression when parameterized by $k + w_-$ even for systems $A$ of difference constraints and right-hand sides $w\in\{\pm 1\}^m$.
\end{theorem}
\begin{proof}
  This follows by combining \Cref{thm:BEDChasnoPolComp} and \Cref{thm:PPTfromBEDCtoMIS} with the help of \Cref{thm:PPTsShowW1andPolComp}.
\end{proof}

\section{Discussion}
\label{sec:discussion}
We investigated the {\sc MinFB} problem from the perspective of parameterized complexity.
We settled the complexity of the problem for various choices of parameters, such as the size of a deletion set or the numbers of positive and negative right-hand sides.
In addition, structural parameters were contemplated: treewidth, pathwidth and treedepth. 
We also ruled out the existence of a polynomial compression for combined parameters $k + w_-$, assuming that $\mathsf{coNP} \not\subseteq \mathsf{NP}/poly$.

It remains a challenging open problem whether DFAS admits a polynomial compression for parameter $k$, and whether our perspective from the more general {\sc MinFB} problem can help in answering that.
A polynomial compression for {\sc MinFB} for parameter $k$ or $k + w_+$, even for node-arc incidence matrices, would imply a polynomial compression for {\sc DFAS} (as the reduction from {\sc DFAS} to {\sc MinFB} does not increase the parameter).

Another interesting open question is whether the problem is tractable when parameterized by pathwidth whenever the pathwidth is at most 5.
Recall that in \Cref{sec:Patwidth_NPHard}, we show that the problem is $\mathsf{NP}$-hard when the parameter pathwidth has value 6.


\medskip
\noindent
\textbf{Acknowledgements.}
  K.B. and L. M. M.-C. were supported by the Lend{\"u}let Programme of the Hungarian Academy of Sciences, grant number LP2021-1/2021, and by the Hungarian National Research, Development and Innovation Oﬃce – NKFIH, grant number FK128673.
  The ``Application Domain Speciﬁc Highly Reliable IT Solutions'' project has been implemented with the support provided from the National Research, Development and Innovation Fund of Hungary, ﬁnanced under the Thematic Excellence Programme TKP2020-NKA-06 (National Challenges Sub-programme) funding scheme.
  A. G. was supported by DFG grant MN 59/1-1.
  The research was supported by the DAAD with funds from the Bundesministerium f{\"u}r Bildung und Forschung.

\bibliographystyle{abbrvnat}
\bibliography{minFB_Bib}





\end{document}